\documentclass[a4paper,12pt]{article}

\usepackage[utf8]{inputenc}
\usepackage{amsmath}
\usepackage{mathtools}
\mathtoolsset{centercolon}
\usepackage{amsfonts}
\usepackage{amsthm}
\usepackage{amssymb}
\usepackage{tikz}
\usetikzlibrary{arrows}
\usetikzlibrary{calc}
\usetikzlibrary{shapes}
\usepackage{url}
\usepackage{stmaryrd}
\usepackage[noend, algoruled]{algorithm2e}
\usepackage{fullpage}
\usepackage{xparse}

\usepackage{xcolor}
\usepackage{enumitem}
\usepackage[colorlinks]{hyperref}
\definecolor{lightblue}{rgb}{0.5,0.5,1.0}
\definecolor{darkred}{rgb}{0.5,0,0}
\definecolor{darkgreen}{rgb}{0,0.5,0}
\definecolor{darkblue}{rgb}{0,0,0.5}

\renewcommand{\phi}{\varphi}
\newcommand{\nat}{\mathbb{N}}

\newcommand{\defining}[1]{\emph{#1}}

\newcommand{\iso}{\cong}

\newcommand{\disunion}{\uplus}

\newcommand{\Struct}{\mathcal{H}}

\newcommand{\autgroupsym}{\text{Aut}}
\newcommand{\autgroup}[1]{\autgroupsym(#1)}

\newcommand{\kernel}[1]{\text{ker}(#1)}

\newcommand{\proj}{\pi}
\newcommand{\outproj}{\overline{\pi}}

\newcommand{\intersect}{\cap}

\newcommand{\wreath}{\wr}
\newcommand{\direct}{\times}

\newcommand{\inv}[1]{#1^{-1}}

\newcommand{\normal}{\triangleleft}

\newcommand{\perm}{\sigma}

\newcommand{\stab}[2]{\mathrm{Stab}_{#1}(#2)}
\newcommand{\restrictGroup}[2]{#1|_{#2}}

\newcommand{\rot} {r}
\newcommand{\rotA} {\rot}
\newcommand{\rotB} {s}
\newcommand{\refl} {\alpha}
\newcommand{\reflA} {\refl}
\newcommand{\reflB} {\beta}
\newcommand{\group}{\Gamma}
\newcommand{\groupA}{\group}
\newcommand{\groupB}{\Delta}
\newcommand{\rotsubgroup}[1]{\text{Rot}(#1)}

\newcommand{\CFIgroup}[0]{\Gamma_\mathsf{CFI}}
\newcommand{\DoubleCFIgroup}[0]{\Gamma_\mathsf{2CFI}}

\newcommand{\ord}[1]{\operatorname{ord}(#1)}

\newcommand{\DihedralGroup}[1]{\mathsf{D}_{#1}}
\newcommand{\CyclicGroup}[1]{\mathsf{C}_{#1}}
\newcommand{\SymGroup}[1]{\mathsf{S}_{#1}}
\newcommand{\SymSetGroup}[1]{\text{Sym}(#1)}

\newcommand{\orbsym}{\text{orb}}
\newcommand{\orbit}[1]{\orbsym(#1)}
\newcommand{\orbitGroup}[2]{\orbsym_{#1}(#2)}
\newcommand{\orbitpart}[1]{\orbsym(#1)}

\newcommand{\ZZ}{\mathbb{Z}}

\newcommand{\set}[1] {\left\lbrace #1 \right\rbrace }
\newcommand{\setcondition}[2] {\left\lbrace #1 \mid #2 \right\rbrace }

\newcommand{\cansym}{\mathsf{can}}
\newcommand{\can}[1]{\cansym(#1)}

\newcommand{\types}{\mathbb{T}}	
\newcommand{\type}{T}

\newcommand{\ces}{\mathcal{S}}
\newcommand{\tces}{\mathcal{T}}

\newcommand{\var}{u}
\newcommand{\varA}{\var}
\newcommand{\varB}{v}
\newcommand{\varC}{w}

\newcommand{\vect}{x}
\newcommand{\vectA}{\vect}
\newcommand{\vectB}{y}
\newcommand{\vectC}{z}

\newcommand{\Vars}{V}
\newcommand{\BVars}{B}

\newcommand{\restrictVect}[2]{#1|_{#2}}
\newcommand{\solutions}[1]{L(#1)}

\newcommand{\isos}[2]{\mathrm{Iso}(#1,#2)}
\newcommand{\img}[1]{\mathrm{im}(#1)}
\newcommand{\orderings}[1]{\mathcal{O}(#1)}

\newcommand{\sig}{\tau}
\newcommand{\rel}{R}
\newcommand{\StructP}{H}
\newcommand{\spleq}{\preceq}
\newcommand{\spless}{\prec}

\newcommand{\colorclasses}[1]{\mathbb{C}_{#1}}

\newcommand{\colclass}{C}
\newcommand{\colstruct}{\mathcal{C}}

\newcommand{\groupSym}{\mathsf{gr}}
\newcommand{\extensionSym}{\mathsf{ex}}

\newcommand{\groupvertices}{\StructP_{\groupSym}}
\newcommand{\extensionvertices}{\StructP_{\extensionSym}}

\newcommand{\typegroupSym}{\group}
\newcommand{\typegroup}[1]{\typegroupSym_{#1}}
\newcommand{\typegroupStruct}[2]{\typegroupSym^{#1}_{#2}}
\newcommand{\grouptypes}{\types_\groupSym}

\def\unionordered{\mathbin{\ooalign{\hss\raisebox{3pt}{\tiny$<$}\hss\cr$\cup$}}}

\newcommand{\tcesunion}[2]{#1 \unionordered #2}

\newcommand{\extExpl}[2]{\text{ext}_{#1}(#2)}

\newcommand{\borderclasses}[1]{B(#1)}

\DeclareMathOperator{\lcm}{lcm}

\newtheorem{theorem}{Theorem}
\newtheorem{definition}[theorem]{Definition}
\newtheorem{lemma}[theorem]{Lemma}
\newtheorem{corollary}[theorem]{Corollary}

\tikzstyle{vertex} = [circle, fill=black, inner sep=0.5mm]

\title{Canonization for Bounded and Dihedral Color Classes in Choiceless Polynomial Time}

\author{Moritz Lichter\\
	TU Kaiserslautern\\
	\texttt{lichter@cs.uni-kl.de}
	\and
	Pascal Schweitzer\\
	TU Kaiserslautern\\
	\texttt{schweitzer@cs.uni-kl.de}
}

\def\PTime{\textsc{Ptime}}

\newcommand\blfootnote[1]{%
\begingroup
\renewcommand\thefootnote{}\footnote{#1}%
\addtocounter{footnote}{-1}%
\endgroup
}

\hypersetup{colorlinks
  ,linkcolor=darkred
  ,filecolor=darkgreen
  ,urlcolor=darkred
  ,citecolor=darkblue,
  linktocpage}
  
\begin{document}
	
\maketitle

\begin{abstract}
	In the quest for a logic capturing \PTime{} the next natural classes of structures to consider are those with bounded color class size.
	We present a canonization procedure for graphs with dihedral color classes of bounded size in the logic of Choiceless Polynomial Time (CPT), which then captures \PTime{} on this class of structures. This is the first result of this form for non-abelian color classes.
	
	The first step proposes a normal form which comprises a ``rigid assemblage''. This roughly means that the local automorphism groups form 
	$2$-injective $3$-factor subdirect products.
	Structures with color classes of bounded size can be reduced canonization preservingly to normal form in CPT.
	
	In the second step, we show that for graphs in normal form with dihedral color classes of bounded size, the canonization problem can be solved in CPT.
	We also show the same statement for general ternary structures in normal form if the dihedral groups are defined over odd domains.\blfootnote{\noindent The research leading to these results has received funding from the European Research Council (ERC) under the European Union’s Horizon 2020 research and innovation programme (EngageS: grant agreement No.\ 820148).}
\end{abstract}

\section{Introduction}
One of the central open questions in the field of descriptive complexity theory asks about the existence of a logic that captures polynomial time (\PTime{})~\cite{Grohe2008}.
This question goes back to Chandra and Harel~\cite{ChandraHarel82}.
They ask whether there is a logic within which we can define exactly the
polynomial-time computable properties of relational structures.
For the complexity class NP such a logic is known, namely existential second order logic. This was shown by Fagin in his famous theorem~\cite{MR0371622}.
However, for the class \PTime{}, the question has been open now for more than 35 years.
A fundamental difficulty at its heart is a mismatch
between logics and Turing machines.
An input has to be written onto a tape to provide it to a Turing machine.
So all inputs are necessarily ordered by the position of each character on the tape. This is the case 
even when there is no natural order to begin with, which for example happens with the vertices of a graph that is encoded.
In contrast to this, such an order is typically not given for a logic.
In fact, if an order is given a priori then there is a logic capturing \PTime{}, for example on totally ordered structures.
Indeed, IFP (first order logic enriched with a fixed-point operator) is such a logic as shown by the Immerman-Vardi Theorem~\cite{Immerman87}.

In the ongoing search for a logic for unordered structures, one of the most promising candidates is the logic Choiceless Polynomial Time (CPT).
It manages to capture an important aspect demanded from a ``reasonable logic''
in the sense of Gurevich~\cite{Gurevich1988},
namely that such a logic cannot make arbitrary choices.
Whenever there are multiple indistinguishable elements,
a logic can either process all or none of them. It is impossible to pick one element arbitrarily and process just that.
This is common for many algorithms executed on Turing machines
exploiting the order given by the tape.
In the form originally defined by Blass, Gurevich, and Shelah~\cite{BlassGurevichShelah99},
CPT has a pseudocode-like syntax for processing hereditarily finite sets.
Most importantly there is a construct to process all elements of a set in parallel, because we cannot choose one to process first.
Subsequently,
there were definitions by Rossman~\cite{rossman2010}
and Grädel and Grohe~\cite{GradelGrohe2015}
in a more ``logical'' way using iteration terms or fixed points.

The question of whether a logic capturing \PTime{} on a class of structures exists
is closely linked to the problem of canonization.
Suppose it is possible to canonize  
input structures from a particular class in a logic
(i.e., to define an isomorphic copy enriched by a total order).
Then the logic (extended by IFP)  captures \PTime{} on this class
by the already mentioned Immerman-Vardi Theorem.
This yields a general approach to show that some logic captures \PTime{}
on a class of structures:
proving that canonization of the structures is definable in this logic.
This approach has been the method of choice for numerous results in descriptive complexity theory.
To this end, it was shown that canonization is IFP+C (IFP with counting) definable on 
interval graphs~\cite{Laubner2010},
graphs with excluded minors~\cite{DBLP:conf/stoc/Grohe00, Grohe2010, Grohe2017,GroheMarino1999},
and graphs with bounded rank width~\cite{DBLP:conf/lics/GroheN19}. Thus IFP+C captures \PTime{} on these classes.
Regarding CPT, all CPT-definable properties and transformations (e.g.,~canonization)
are in particular polynomial-time computable.
If canonization is CPT-definable for a graph class
then CPT captures \PTime{} on this class
because CPT subsumes IFP+C.

Closely linked to the problem of canonization is the problem of isomorphism testing.
A polynomial-time canonization algorithm implies a polynomial-time isomorphism test.
While we do not know of a formal reduction the other way around,
we usually have efficient canonization algorithms for all classes
for which an efficient isomorphism test is known
(see~\cite{DBLP:conf/stoc/SchweitzerW19} for an overview).
This statement can even be proven unconditionally for classes of vertex-colored structures with a CPT-definable isomorphism problem~\cite{deepWL}.
Accepting for the moment that canonization and isomorphism testing are algorithmically very related,
we arrive at the following observation:
if isomorphism testing is polynomial-time solvable
on a class of structures then, to capture \PTime{}, we must ``solve'' the isomorphism problem in the logic anyway. 
If we do so in CPT we (almost) immediately obtain a logic capturing \PTime{}. In summary, it appears the question of a logic for \PTime{} boils down to isomorphism testing within a logic. 

There is a notable class for which we have polynomial-time isomorphism testing and canonization algorithms, but for which we do not know
how to canonize them in CPT.
This is the class of structures with bounded color class size.
Specifically, for an integer~$q$, a $q$\nobreakdash-bounded structure is a vertex-colored structure, 
where at most $q$ vertices have the same color
and a total order on the colors is given.
Structures with bounded color class size can be canonized in polynomial time
with group-theoretic techniques (see~\cite{FurstHopcroftLuks80, bab79,  DBLP:conf/stoc/BabaiL83}).
The introduction of group-theoretic techniques  marks an important step for the design of canonization algorithms.
The use of algorithmic group theory turned out to be very fruitful
and subsequently lead to Luks's famous polynomial-time isomorphism test for graphs of bounded degree \cite{DBLP:journals/jcss/Luks82}.
It uses a more general and more complicated machinery than needed for bounded color classes.

For the purpose of isomorphism testing,
these group-theoretic techniques inherently rely on choosing generating sets,
and it is not clear how this can be done in a choiceless logic. 
A well-known construction of Cai, Fürer, and Immerman~\cite{CaiFuererImmerman1992} shows 
that IFP+C does not provide us with a \PTime{} logic for $2$-bounded structures.
Finding an natural, alternative logic for $q$-bounded structures is
still an open problem.
Studying structures with bounded color class size is a reasonable a next step,
because the canonization algorithm for them makes use of comparatively easy group theory
but we still do not know how to transfer these techniques into logics.

A first result towards the canonization of
structures with bounded color class size in CPT
was the canonization of structures with abelian colors,
that is the automorphism group of every color class is abelian,
due to Abu Zaid, Grädel, Grohe, and Pakusa~\cite{AbuZaidGraedelGrohePakusa2014}.
They use a certain class of linear equation systems
to encode the group-theoretic structure of abelian color classes
and solve these systems in CPT.
In particular, they show that CPT captures \PTime{} on $2$-bounded structures.
Considering dihedral groups is a next natural step
because dihedral groups are extensions of abelian groups by abelian groups.

\subparagraph{Contribution.}
This paper presents a canonization procedure in CPT for finite $q$-bounded  structures with dihedral colors.
A color class is dihedral (resp.~cyclic) if it induces a substructure
whose automorphism group is dihedral (resp.~cyclic).
A dihedral group is the automorphism group of a regular $n$-gon consisting of rotations and reflections and we call it odd if $n$ is odd.
Dihedral groups are non-abelian for $n>2$.
We thereby provide the first canonization procedure 
for a class of $q$-bounded structures with
non-abelian color classes
and in particular show that CPT captures \PTime{} on it.
Overall, we prove the following theorem:
\begin{theorem}
	\label{thm:canonize-structures-CPT}
	The following structures can be canonized in CPT:
	\begin{enumerate}
		\item $q$-bounded relational structures of arity at most $3$
		with odd dihedral or cyclic colors
		\item $q$-bounded graphs with dihedral or cyclic colors.
	\end{enumerate}
\end{theorem}

Our approach consists of two steps.
As a first step, we propose a normal form for arbitrary finite $q$-bounded structures.
Then, in a second step, we use group-theoretic arguments
to canonize structures with dihedral colors given in the aforementioned normal form. 

Concretely, the first step is a reduction
transforming the input structure into a normal form,
which ensures that
a color class and its adjacent color classes form a ``rigid assemblage''.
That is, locally the automorphism groups form 
$2$-injective $3$-factor subdirect products
or are quotient groups of other color classes.
In the the case of $2$-injective $3$-factor subdirect products, the automorphisms of three adjacent color classes
are not independent of each other.
This means that every nontrivial automorphism of
the substructure induced by these three color classes
is never constant on two of them.
More precisely, we prove the following theorem (formal definitions and proofs are given later in the paper).
\begin{theorem}
	\label{thm:convert-to-normal-form-light}
	For every $q$ and signature $\sig$
	there is a $q'$ and another signalure $\sig'$,
	such that
	relational $q$-bounded $\sig$-structures of arity $r$
	can be reduced canonization preservingly in CPT
	to $q'$-bounded $2$-injective
	quotient $\sig'$-structures.
\end{theorem}

It was not necessary to consider $2$-injective groups for abelian colors yet, but it is for non-abelian colors.
Towards a reduction step, a purely group-theoretic analysis of
$2$-injective groups is given in~\cite{NeuenSchweitzer2016}. The main insight is basically that such groups
decompose naturally into structurally simpler parts which are related via a common abelian normal subgroup.
We extend the techniques to canonize abelian color classes
and show how they can be combined with the analysis of $2$-injective groups
to obtain a canonization procedure for said structures with dihedral colors in CPT. 
That is, we provide new methods to integrate group-theoretic reasoning,
which is at the core of canonizing $q$-bounded structures algorithmically,
into logics.

\subparagraph{Our Technique.}
The strategy of our canonization procedure 
is to reduce the dihedral groups
in some way to abelian groups and then exploit
the canonization procedure of~\cite{AbuZaidGraedelGrohePakusa2014}.

Since the automorphism groups of the color classes are restricted
to be dihedral or cyclic,
we can characterize all occurring
$2$-injective $3$-factor subdirect products.
Using this characterization we show that 
we can partition the input structure into parts we call reflection components.
These reflection components have the property that automorphisms either simultaneously reflect the points in all color classes of the component or in none. In the latter case they rotate the points in all color classes.
We use this property to force all groups in a reflection component to become abelian: we prohibit reflections in one color class of the reflection component
and this automatically prohibits reflections in all other classes of the component, too.
Once all reflections are removed, the remaining groups are abelian.
Then we apply the canonization procedure for structures with bounded abelian color classes  from~\cite{AbuZaidGraedelGrohePakusa2014} to the entire reflection component.

Not limiting ourselves to dihedral groups but also allowing cyclic groups has the benefit that the class of occurring groups
is closed under quotients and subgroups.
Quotient and subgroups of the input color classes occur naturally in our reduction process to the normal form.
For dihedral groups it turns out that odd dihedral groups
are easier to handle than the other ones.
In fact, reflection components are not really independent but can have ``global'' dependencies.
We show that for ``odd'' dihedral groups, reflection components can only be related
via color classes with abelian automorphism groups,
that is their global dependencies are abelian.
For ``even'' dihedral groups, there is a single
non-abelian exception that can connect reflection components, which complicates matters.
For even dihedral color classes this restricts us to the treatment of graphs (see Theorem~\ref{thm:canonize-structures-CPT} above).

Towards generalization,
it unfortunately becomes cumbersome to exploit the group structure theory in CPT, which is heavily required to execute the approach.
Extending the treatment of linear equation systems, which is a subroutine in~\cite{AbuZaidGraedelGrohePakusa2014},
to dihedral groups requires significant work already.
We still follow the strategy of~\cite{AbuZaidGraedelGrohePakusa2014}
and use a certain class of equation systems to encode the global dependencies.
However, we need to generalize the equation systems.
Consequently, we have to adapt all operations used on these equations systems
to work in the more general setting (e.g. the check for consistency).
This becomes technically even more involved than the techniques of~\cite{AbuZaidGraedelGrohePakusa2014} already are.

\subparagraph{Related Work.}
There already exist various results for CPT regarding structures with bounded color class size in addition to the ones mentioned above:
Cai, Fürer, and Immerman
introduced the so-called CFI graphs. From every base graph, a pair of two non-isomorphic CFI graphs is derived.
The isomorphism problem on these pairs of graphs is used to separate IFP+C from \PTime{}~\cite{CaiFuererImmerman1992}.
Dawar, Richerby, and Rossman showed in~\cite{DawarRicherbyRossman2008}
that the isomorphism problem for the CFI graphs
can be solved in CPT for base graphs of color class size 1.

This result was strengthened by Pakusa, Schalthöfer, and Selman
to base graphs with logarithmic color class size~\cite{PakusaSchalthoeferSelman2016}.
The techniques of~\cite{DawarRicherbyRossman2008}
and~\cite{PakusaSchalthoeferSelman2016}
are used in~\cite{AbuZaidGraedelGrohePakusa2014}
to solve the mentioned equation systems.

The logic IFP+C has a strong connection to the higher dimensional Weisfeiler-Leman algorithm and to an Ehrenfeucht–Fraïssé-like game,
the so-called bijective pebble game.
They are often used to show that IFP+C identifies
graphs in a given graph class,
i.e., that for any two non-isomorphic graphs of the class
there is an IFP+C formula distinguishing them.
It was shown by Otto~\cite{Otto1997}
that if IFP+C captures \PTime{} on a graph class
then IFP+C also identifies all graphs in this class.
The converse direction is open~\cite{Grohe2010}. The capabilities of IFP+C to detect graph decompositions were recently investigated in~\cite{DBLP:conf/mfcs/KieferN19}.

\subparagraph{Structure of this Paper.}
We begin with the characterization of
$2$-injective $3$-factor subdirect products of dihedral and cyclic groups
in Section~\ref{sec:classification-2-inj-products}.
Then we turn to structures and to permutation groups.
We begin with the reduction to the above mentioned normal form
in Section~\ref{sec:normal-forms}
and then preprocess structures with dihedral colors in Section~\ref{sec:dihedral-colors}.
In Section~\ref{sec:cyclic-linear-equations-systems}
we introduce tree-like cyclic linear equation systems (TCES)
and show that a certain subclass of them can be solved in CPT.
Finally,
we define and analyse reflection components in Section~\ref{sec:canonization-dihedral}
and give the CPT-definable canonization procedure for dihedral colors.

\section{Preliminaries}

We write $[k]$ for the set $\set{1,\dots , k}$ for $k \in \nat$. 

Let $\vectA \in M^N$ be a vector over the set $M$ indexed by the set $N$.
For $u \in N$ we write $\vectA(u)$ for the entry of $\vectA$ at position $u$.
Let $N' \subseteq N$.
We denote with $\restrictVect{\vectA}{N'} \in M^{N'}$
the vector obtained from $\vectA$
that satisfies $\vectA(u) = \restrictVect{\vectA}{N'}(u)$ for all $u \in N'$.
In the case that $N' \not \subseteq N$
we write $\restrictVect{\vectA}{N'}$ as an abbreviation for
$\restrictVect{\vectA}{N' \cap N}$.
We extend the notation to subsets of $K \subseteq M^N$:
$\restrictVect{K}{N'} := \setcondition{\restrictVect{\vect}{N'}}{\vect \in K}$.
For a set of vectors $K' \subseteq M^{N'}$
we write $\extExpl{N}{K'} := \setcondition{\vectA \in M^N}{\restrictVect{\vectA}{N'} \in K'}$
for the extension of $K'$ to $N$.

\subparagraph{Groups}
All groups considered in the paper will be finite.
We make use of standard group notation:
Let $\groupA$ be a group.
The identity (or trivial) element of $\groupA$ is written as~$1$.
The group operation is written as multiplication,
i.e., the product of $g,h\in\groupA$ is $gh$.
We write $\ord{g}$ for the order of $g \in \groupA$, that is the smallest $k \in \nat$ such that $g^k = 1$.
The number of elements $|\groupA|$ of $\groupA$ is called to order of $\groupA$.
For elements $g_1, \dots, g_i \in \groupA$
we denote with $\langle g_1, \dots, g_i \rangle$
the group generated by $\{g_1, \dots, g_i\}$.
We additionally use the notation for subgroups
$\groupB_1, \dots, \groupB_i \leq \groupA$,
where $\langle \groupB_1, \dots, \groupB_i \rangle$
is the group generated by all elements contained in the subgroups.
The index $[\groupA : \groupB]$ of a subgroup $\groupB\leq \groupA$ in $\groupA$
is $[\groupA : \groupB] = |\groupA| / |\groupB|$.
A (left) coset of $\groupA$ is a set
$g\groupA := \setcondition{gh}{h \in \groupA}$.
A normal subgroup $N \normal \groupA$ is a subgroup $N \leq \groupA$
satisfying $gN = Ng$ for all $g \in \groupA$.
Let $N \normal \groupA$. The quotient group $\groupA / N$
is the group of cosets $gN$ for all $g \in \group$.
Two groups $\groupA_1$ and $\groupA_2$ are isomorphic,
denoted $\groupA_1 \iso \groupA_2$,
if there is an isomorphism from $\groupA_1$ to~$\groupA_2$.

Of special interest in this paper are permutation groups.
We write $\SymSetGroup{\Omega}$ for the symmetric group
with domain $\Omega$,
i.e., the group of all permutations of $\Omega$.
We abbreviate $\SymGroup{n} := \SymSetGroup{[n]}$.
Let $\group \leq \SymSetGroup{\Omega}$ be a permutation group.
The orbit of $u \in \Omega$ in $\group$
is the set of points onto which $u$ can be mapped,
i.e.,~$\orbitGroup{\group}{u} := \setcondition{v \in \Omega}{\exists \perm \in \group.\ \perm(u) = v}$. 
If the group is clear from the context, we just write $\orbit{u}$.
The set of orbits of $\group$ is
$\orbitpart{\group} := \setcondition{\orbit{u}}{ u \in \Omega}$
and defines a partition of $\Omega$.
The group $\group$ is called transitive if 
$\group$ has only one orbit, i.e.,~$\orbit{\group} = \set{\Omega}$.
An action of a group~$\group$ on a set~$\Omega$ is a homomorphism from~$\Gamma$ to~$\SymSetGroup{\Omega}$. By considering the image of group elements under the homomorphism, we can speak about orbits of group elements on~$\Omega$. 
The $k$-orbits of~$\group$ are the orbits of $\group$ acting on $\Omega^k$ component-wise.
Lastly, a finite permutation group $\group$ is regular if $\group$ is transitive
and $|\group| = |\Omega|$.

Let $\Omega' \subseteq \Omega$.
We denote with
$\stab{\group}{\Omega'} := \setcondition{\perm \in \group}{\perm(\Omega') = \Omega'}$
the setwise stabilizer of $\Omega'$ in $\group$ and write $\stab{\group}{u}$ for $\stab{\group}{\set{u}}$
for $u \in \Omega$. 
We define $\restrictGroup{\group}{\Omega'} := \setcondition{\restrictVect{\perm}{\Omega'}}{\perm \in \stab{\group}{\Omega'}}$
to be the stabilizer of $\Omega'$ restricted to $\Omega'$.

Let $\groupA \leq G_1 \direct G_2 \direct G_3$ be a group.
We denote with $\proj^\groupA_i \colon \groupA \rightarrow G_i$ the projection
onto the $i$-th factor and 
with $\outproj^\groupA_i \colon \groupA \rightarrow G_{i+1} \direct G_{i+2}$
(indices modulo $3$) the projection onto the factors other than the $i$-th one.
If for $i \in [3]$ the groups~$G_i$ are permutation groups acting on pairwise disjoint sets $\Omega_i$,
we also write $\proj^\groupA_{\Omega_i}$ and $\outproj^\groupA_{\Omega_i}$.
We omit the group and only write~$\proj_i$ and $\outproj_i$
if $\group$ is clear from the context.

\subparagraph{Bounded Relational Structures}
A (relational) signature
$\tau = \set{\rel_1, \dots, \rel_k}$ is a set of relation symbols
with associated arities $r_i \in \nat$ for all $i \in [k]$.
We consider signatures containing a binary relation symbol $\spleq$.
A $\tau$-structure $\Struct$ is a tuple
${\Struct = (\StructP, \rel_1^\Struct, \dots, \rel_k^\Struct, \spleq)}$
where $\rel_i^\Struct \subseteq\StructP^{r_i}$ for all $i \in [k]$
and $\spleq \subseteq \StructP^2$ is a total preorder.
Almost all structures considered in this paper will be finite%
\footnote{The hereditary finite expansion of a finite structure used to define CPT is the only occurrence of a non-finite structure in this paper.}.
Note that a $\tau$-structure defines an order on the relations.
The preorder~$\spleq$ partitions~$\StructP$ into equivalence classes,
which we call \defining{color classes}.
The total preorder~$\spleq$ induces a total order
on the color classes.
We denote the set of $\Struct$-color classes by
$\colorclasses{\Struct}$.
For a set $I \subseteq \StructP$
we denote with~$\Struct[I]$ the substructure induced by $I$.
If $I = \colclass$ is a color class, we just write $\colstruct$ for $\Struct[\colclass]$,
if the structure $\Struct$ is clear from the context.
If $I \subseteq \colorclasses{\Struct}$ we also write
$\Struct[I]$ for $\Struct[\bigcup I]$.
We say that two colors classes $\colclass,\colclass' \in \colorclasses{\Struct}$
are \defining{related},
if $\Struct[\colclass \cup \colclass'] \neq \Struct[\colclass] \cup \Struct[\colclass']$,
so there is a relation that contains a tuple consisting of vertices of both
$\colclass$ and $\colclass'$.

We say that a relation $\rel_i^\Struct$
is \defining{homogeneous} if  $\rel_i^\Struct \subseteq \colclass^k$
for some $\colclass \in \colorclasses{\Struct}$ and $k \in \nat$.
A relation is \defining{heterogeneous} if it is not homogeneous.

In the following, we are mostly interested in the arity of heterogeneous relations.
We say that a structure $\Struct$
is of \defining{inter-color-arity} $r$
if the largest arity of a heterogeneous relation is~$r$.
For simplicity we just say that $\Struct$ is of arity $r$. 
A structure is \defining{$q$-bounded}, if $|\colclass| \leq q$
for all $\colclass \in \colorclasses{\Struct}$
and the arity of every homogeneous relation is bounded by $q$. 
We are interested in bounding the arity of heterogeneous relations,
because it bounds the number of color classes we need to consider 
simultaneously. In contrast to this, we could actually allow homogeneous relations of arbitrary arity, since they only ever involve one color class.
In fact, bounding the arity of homogeneous relations is technically not necessary since higher arity implies repeated entries. However, having the bound simplifies our exposition.

An ordered copy of $\Struct$ is a pair $(\Struct', <)$,
such that
$\Struct' = (\StructP', \rel_1^{\Struct'}, \dots, \rel_k^{\Struct'}, \spleq')$,
$\Struct' \iso \Struct$, and $<$ is a total order that refines
$\spleq'$.
A canonical copy $\can{\Struct}$ is an ordered copy of $\Struct$
obtained in a canonical way, i.e.,~defined in CPT in the following.

We write $\autgroup{\Struct}$ for the automorphism group of a structure $\Struct$.
The orbit of a tuple $(v_1, \dots, v_k) \in \StructP^k$
is the orbit of $(v_1,\dots,v_k)$ under the action of $\autgroup{\Struct}$ on $\StructP^k$.
For two structures~$\Struct_1$ and~$\Struct_2$
we write $\isos{\Struct_1}{\Struct_2}$ for the set of isomorphisms between~$\Struct_1$ and~$\Struct_2$.
For a canonical copy $\can{\Struct}$ we call the set $\isos{\Struct}{\can{\Struct}}$
the \defining{canonical labellings}.
We denote with $\orderings{M}$ the set of orderings of $M$,
that is the set of bijections $M \to [|M|]$.

\newcommand{\Atoms}{\textsf{Atoms}}
\newcommand{\Pair}{\textsf{Pair}}
\newcommand{\Union}{\textsf{Union}}
\newcommand{\Unique}{\textsf{Unique}}
\newcommand{\Card}{\textsf{Card}}
\newcommand{\HFsym}{\textsf{HF}}
\newcommand{\HF}[1]{\HFsym(#1)}
\newcommand{\halt}{\text{halt}}
\newcommand{\out}{\text{out}}
\newcommand{\TC}[1]{\textsf{TC}(#1)}
\newcommand{\act}[2]{\text{act}(#1,#2)}
\newcommand{\actStruct}[3]{\text{act}(#1,#2,#3)}
\newcommand{\denotation}[1]{\llbracket #1 \rrbracket}

\newcommand{\false}{\textsf{false}}
\newcommand{\true}{\textsf{true}}

\subparagraph{Choiceless Polynomial Time}

To give a concise definition of CPT,
we follow the definition of~\cite{GradelGrohe2015} and
use the same idea of e.g.~\cite{pakusa2015} to enforce polynomial bounds.

For a set of atoms $A$
we denote with $\HF{A}$ the \defining{hereditarily finite sets over~$A$}.
This is the inclusion-wise smallest set with $A \subseteq \HF{A}$
and $a \in \HF{A}$ for every ${a \subseteq \HF{A}}$.
A~set $a \in \HF{A}$ is called transitive,
if $c \in a$ whenever there is some $b$ with $ c\in b \in a$.
We denote with $\TC{a}$ the \defining{transitive closure} of $a$,
that is the least (with respect to set inclusion)
transitive set $b$ with $a \subseteq b$.

Let $\sig$ be a signature.
We extend $\sig$ by adding set-theoretic function symbols to obtain
$\sig^\HFsym := \sig \disunion \set{\emptyset,\Atoms, \Pair,\Union, \Unique, \Card}$,
where $\emptyset$ and $\Atoms$ are constants,
$\Union, \Unique,$ and $\Card$ are unary, and $\Pair$ is binary.
For a $\sig$-structure $\Struct$, the \defining{hereditarily finite expansion}
$\HF{\Struct}$ is a $\sig^\HFsym$-structure over the universe $\HF{\StructP}$ defined as follows:
all relations in $\sig$ are interpreted as they are in $\Struct$.
The other function symbols have usual set theoretic interpretation:
\begin{itemize}
	\item $\emptyset^{\HF{\Struct}} = \emptyset$ and $\Atoms^{\HF{\Struct}} = \StructP$
	\item $\Pair^{\HF{\Struct}}(a,b) = \set{a,b}$
	\item $\Union^{\HF{\Struct}}(a) = \setcondition{b}{\exists c \in a.~b\in c}$
	\item $\Unique^{\HF{\Struct}}(a) = \begin{cases}
		b & \text{if } a = \set{b}\\
		\emptyset & \text{otherwise}
	\end{cases}$
	\item $\Card^{\HF{\Struct}}(a) = \begin{cases}
		|a| &\text {if } a \notin \StructP\\
		\emptyset &\text{otherwise}
	\end{cases}$,\\
	where the number $|a|$ is encoded as a von Neuman ordinal.
\end{itemize} 
Note that the $\Unique$ function can only be applied to singleton sets
and hence is invariant under automorphisms.

Before defining the logic CPT we define the BGS logic
(after Blass, Gurevich, and Shelah)
and then obtain CPT as its polynomial time fragment:
A \defining{BGS term} is composed as usual from variables and
function symbols from $\sig^{\HFsym}$.
There are two additional constructs:
if $s(\bar{x},y)$ and $t(\bar{x})$ are terms with
a sequence of free variables $\bar{x}$
(and an additional free variable $y$ in the case of $s$)
and $\phi(\bar{x}, y)$ is a formula with free variables $\bar{x}$ and $y$
then $r(\bar{x}) = \setcondition{s(\bar{x}, y)}{y \in t(\bar{x}), \phi(\bar{x}, y)}$
is a \defining{comprehension term} with free variables $\bar{x}$.
If $s(x)$ is a term with a single free variable $x$,
then $s^*$ is an \defining{iteration term} without free variables.
\defining{BGS formulas} are composed of terms $t_1, \dots, t_k$ as
$\rel(t_1, \dots, t_k)$ (for $\rel \in \sig$ of arity $k$),
terms $t_1 = t_2$, and the usual boolean connectives.

Let $\Struct$ be a $\sig$ structure.
BGS terms and formulas are interpreted over $\HF{\StructP}$.
We define the denotation $\denotation{t}^\Struct \colon \HF{\StructP}^k \to \HF{\Struct}$
that for a term $t$ with free variables ${\bar{x} = (x_1, \dots, x_k)}$
maps $\bar{a} = (a_1, \dots ,a_k) \in \HF{\Struct}^k$
to the value of $t$ if we replace $x_i$ with~$a_i$ (for $i \in [k]$).
For a formula $\phi$ with free variables $\bar{x}$
we define $\denotation{\phi}^\Struct$ to be the set of all 
$\bar{a} = (a_1, \dots ,a_k) \in \HF{\Struct}^k$ satisfying~$\phi$.

For the comprehension term $r$ as above, the denotation is defined as follows:
$\denotation{r}^\Struct(\bar{a}) = 
\setcondition{\denotation{s}^\Struct(\bar{a}b)}{b \in \denotation{t}^\Struct(\bar{a}),
(\bar{a}b) \in \denotation{\phi}^\Struct}$,
where $\bar{a}b$ denotes the tuple $(a_1, \dots , a_k, b)$.
For an iteration term $s^*$ we define a sequence of sets
by setting $a_0 := \emptyset$ and $a_{i+1} := \denotation{s}^\Struct(a_i)$.
Let $\ell := \ell(s^*,\Struct)$
be the least number satisfying $a_{i+1} = a_{i}$.
We set $\denotation{s^*}^\Struct := a_i$ if such an~$\ell$ exists and set
$\denotation{s^*}^\Struct := \emptyset$ otherwise.

A CPT term (or formula respectively) is a tuple $(t,p)$ (or $(\phi, p)$ respectively) of a BGS term (or formula) and a polynomial $p(n)$.
CPT has the same semantics as BGS by replacing~$t$ with $(t,p)$
everywhere (or $\phi$ with $(\phi, p)$)
with an exception for iteration terms:
We set $\denotation{(s^*,p)}^\Struct := \denotation{s^*}^\Struct$
if  $\ell(s^*, \Struct) \leq p(|\StructP|)$
and $|\TC{a_i}| \leq p(|\StructP|)$ for all $i$,
where the $a_i$ are defined as above.
Otherwise, we set $\denotation{(s^*,p)}^\Struct := \emptyset$.
We use $|\TC{a_i}|$ as a measure of the size of $a_i$,
because by transitivity of $\TC{a_i}$,
whenever there is a set $b_k \in  \dots \in b_1 \in a_i$,
then also $b_k \in \TC{a_i}$ and thus $\TC{a_i}$
counts all sets occurring somewhere in the structure of $a_i$.
Putting polynomial bounds on iteration terms suffices,
because all other terms can increase the size of the constructed sets only polynomially.

In this paper we frequently build CPT terms
that compute certain sets, by which we mean that the term's denotation is the output set.
For readability, we use the usual set notation and common abbreviations,
e.g.~$x \cup y$ for $\Union(\Pair(x,y))$
and $(x,y)$ for $\Pair(x,\Pair(x,y))$,
or describe on a higher level how the desired CPT formula is obtained.

\section{Classification of $2$-Injective Subdirect Products of Dihedral Groups}
\label{sec:classification-2-inj-products}

In this section we focus on dihedral groups
and classify all $2$-injective $3$-factor subdirect products
of dihedral and cyclic groups.
We first recall some existing definitions and
introduce new ones, which are then used to state our classification concisely.
The rest of the section is used to prove this classification.

\subparagraph{$2$-Injective Subdirect Products}A group $\groupA \leq G_1 \direct G_2 \direct G_3$ is called a \defining{($3$-factor) subdirect product}
if the projection to each factor is surjective, that is $\proj^\groupA_i(\groupA) = G_i$ for all ${i \in [3]}$.
It is called \defining{$2$-injective} if $\kernel{\outproj^\groupA_i} = \{1\}$ for all $i \in [3]$, that is, each projection onto two components is an injective map. Another way of looking at this is that two components of an element of $\groupA$ determine the third one uniquely.
Simple examples of $2$-injective subdirect products are diagonal subgroups:
for a group $G$, the group $\setcondition{(g,g,g)}{g\in G} \leq G^3$
is called the \defining{diagonal subgroup}.

\subparagraph{Dihedral Groups} For $n\geq 3$, the \defining{dihedral group~$\DihedralGroup{n}$} of order~$2n$ is the automorphism group of a regular~$n$-gon in the plane.
It consists of~$n$ rotations and~$n$ reflections and acts naturally on the set of~$n$ vertices of the polygon. 
In the degenerate cases $\DihedralGroup{1}$ and $\DihedralGroup{2}$ of orders~$2$ and~$4$ respectively, the dihedral group is abelian. For~$n\geq 2$ is it non-abelian.
We write $\CyclicGroup{n}$ for the cyclic group of order $n$.
It holds that $\DihedralGroup{1} \iso \CyclicGroup{2}$ and $\DihedralGroup{2} \iso \CyclicGroup{2}^2$.
For~$n\geq 3$, we can define rotations and reflections of dihedral groups independently of the action of the group.
\begin{definition}[Rotation and Reflection]

	Let $n \geq 3$.
	We call an element $\rot \in \DihedralGroup{n}$ a \defining{rotation},
	if $\ord{\rot} \geq 3$ or $\rot$ commutes with
	all rotations of order at least $3$. 
	An element $\refl \in \DihedralGroup{n}$ is called a \defining{reflection}
	if it is not a rotation.
	We extend the notation to tuples $g=(g_1, \dots, g_k) \in \DihedralGroup{n_1} \direct \cdots \direct \DihedralGroup{n_k}$:
	If $g_i$ is a rotation (respectively reflection) for all $i \in [k]$,
	then $g$ is called a rotation (respectively reflection).
\end{definition}
Note that we regard the identity $1$ as a rotation.
We use the letters $\rotA, \rotB$ for rotations
and $\reflA, \reflB$ for reflections.
Suppose for $n>2$ that $\rotA, \rotB \in \DihedralGroup{n}$ are rotations and~$\reflA, \reflB \in \DihedralGroup{n}$ are reflections. Then
$\rotA\rotB = \rotB\rotA$ and $\reflA\reflB$ are rotations
and $\rotA\reflA=\reflA\inv{\rotA}$ is a reflection.

For~$n_i\geq 3$ the direct product~$\DihedralGroup{n_1} \direct \cdots \direct \DihedralGroup{n_k}$ contains mixed elements that are neither a reflection nor a rotation. Subgroups of such a group may or may not contain such mixed elements.

\begin{definition} [Rotate-or-Reflect Group]
	Let $\group \leq \DihedralGroup{n_1} \direct \cdots \direct \DihedralGroup{n_k}$
	and $n_i > 2$ for all $i \in [k]$.
	We call $\group$ a \defining{rotate-or-reflect} group if 
	every $g \in \group$
	is a rotation or a reflection.
\end{definition}

\begin{definition}[Rotation Subgroup]
	We define the rotation subgroup of $\DihedralGroup{n}$ for $n > 2$
	as
	\[\rotsubgroup{\DihedralGroup{n}} := \setcondition{g \in \DihedralGroup{n}}{g \text{ is a rotation}} \iso \CyclicGroup{n}.\]
	For a group $\group \leq \DihedralGroup{n_1} \direct \dots \direct \DihedralGroup{n_k}$ with $n_i > 2$ for all $i\in [k]$ we define
	\[\rotsubgroup{\group} := \group \intersect \left( \rotsubgroup{\DihedralGroup{n_1}} \direct \cdots \direct \rotsubgroup{\DihedralGroup{n_k}} \right).\]
\end{definition}

\begin{corollary}
	Let $\group \leq \DihedralGroup{n_1} \direct \cdots \direct \DihedralGroup{n_k}$
	be a reflect-or-rotate group.
	Then $\rotsubgroup{\group} = \group \cap \left( \rotsubgroup{\DihedralGroup{n_i}} \direct \DihedralGroup{n_2} \direct \cdots \direct \DihedralGroup{n_k} \right)$.
	By symmetry restricting any other factor to its rotation subgroup yields the same group.
\end{corollary}

Now, we are prepared to state the classification of
$2$-injective subdirect products of dihedral groups and cyclic groups.
For this we analyse how reflections and rotations of
dihedral groups in $2$-injective subdirect products can be combined.
We will see that if no factor is isomorphic to $\DihedralGroup{1}$, $\DihedralGroup{2}$, and $\DihedralGroup{4}$
we always obtain a rotate-or-reflect group.
Specifically, the goal of this section is to prove the following two theorems
(the mentioned double CFI group is defined in Section~\ref{sec:CFI-and-double-CFI}, see Figure~\ref{fig:double-cfi}):

\begin{theorem}	
	\label{thm:classification-2inj-sub-dihedral}
	Let $\group \leq  \DihedralGroup{n_1} \direct \DihedralGroup{n_2} \direct \DihedralGroup{n_3}$
	be a $2$-injective subdirect product.
	Then exactly one of following holds:
	\begin{enumerate}
		\item $n_i > 2$ for all $i \in [3]$ and $\group$ is a rotate-or-reflect group.
		\item $n_i = 4$ for all $i \in [3]$  and $\group$ is isomorphic to the double CFI group~$\DoubleCFIgroup$.
		\item $n_i \leq 2$, $n_j = n_k > 2$ for $\set{i,j,k} = [3]$, and $\outproj_i(\group)$ is a rotate-or-reflect group.
		\item $n_i  \leq 2$ for all $i \in [3]$   and $\group$ is abelian.
	\end{enumerate}
\end{theorem}

\begin{theorem}
	\label{thm:classification-2inj-sub-dihedral-cyclic}
	Let $\group \leq \CyclicGroup{n_1} \direct \DihedralGroup{n_2} \direct \DihedralGroup{n_3}$ be a $2$-injective subdirect product.
	Then exactly one of the following holds:
	\begin{enumerate}
		\item $n_1 \leq2$, $n_2, n_3 > 2$, and $\outproj_1(\group)$ is a rotate-or-reflect group.
		\item $n_1 = n_2 = n_3 = 4$ and $\group \iso \DoubleCFIgroup \cap (\rotsubgroup{\DihedralGroup{4}} \direct \DihedralGroup{4} \direct \DihedralGroup{4})$.
		\item $n_1,n_2, n_3 \leq 2$ and $\group$ is abelian.
	\end{enumerate}
	Furthermore, there are no $2$-injective subdirect products of 
	$\DihedralGroup{n} \direct G_2 \direct G_3$ 
	for $n>2$ if $G_2$ and~$G_3$ are abelian groups.
\end{theorem}

We also checked the classification with a computer program written in the computer algebra system GAP~\cite{GAP} up to $n_i = 20$.

\subsection{Dihedral Groups not of Order $2$, $4$, or $8$ }
We first consider the more general case
where all dihedral groups are $\DihedralGroup{i}$ for $i \notin \set{1,2,4}$.
For this, we collect some basic facts of $2$-injective subdirect products.
\begin{lemma}
	\label{lem:2-inj-sub-basics}
	Let $\groupA \leq G_1 \direct G_2 \direct G_3$
	be a $2$-injective group and $g_i \in G_i$ for $i \in [3]$. 
	If~$(g_1,g_2,g_3)\in \groupA$ then~$\ord{g_i}$ divides the least common multiple~$\lcm\{\ord{g_{i+1}},\ord{g_{i+2}}\}$ (indices modulo $3$).  In particular:
	\begin{enumerate}[label=\alph*)]
		\item If $g_1 \neq 1$ then $(g_1,1,1) \notin \groupA$.\label{part:nontrivial:order:lemma}
		\item If $\ord{g_1} > \ord{(g_2, g_3)}$ then $(g_1, g_2, g_3) \notin
		\groupA$.\label{part:greater:order:two:nontrivial:lemma}
	\end{enumerate}
\end{lemma}
\begin{proof}
	By symmetry we only consider~$i=1$.
	Suppose for~$(g_1,g_2,g_3)\in \groupA$ that~$\ord{g_1}$ does not divide~$\ell :=  \lcm\{\ord{g_{2}},\ord{g_{3}}\}$. Then~$(g_1^\ell ,g_2^\ell,g_3^\ell)= (g_1^\ell ,1,1)\in \groupA\setminus\{1\}$, which contradicts 2-injectivity. Items a) and b) follow immediately.
\end{proof}

For a $2$-injective subdirect product $\groupA$,
we define  $H^\group_i := \kernel{\proj^\group_i} = \setcondition{(g_1, g_2, g_3)}{g_i = 1}$.
Note that $H^\group_i$ defines an isomorphism between $\proj_{i+1}(H^\group_i)$ and $\proj_{i+2}(H^\group_i)$ (indices again modulo~3) since the entries for~$g_{i+1}$ and~$g_{i+2}$ are in one to one correspondence when~$g_i=1$.
We finally set $H^\group := \langle H^\group_1, H^\group_2, H^\group_3 \rangle$.
In the following we just write $H_i$ (respectively $H$) for $H^\group_i$ (respectively $H^\group$),
if the group $\group$ is clear from the context.

\begin{theorem}[\cite{NeuenSchweitzer2016}]
	\label{thm:two-inj-sub-kernels}
	Let $\groupA \leq G_1 \direct G_2 \direct G_3$ be a $2$-injective subdirect product.
	Then $[G_i : \proj_i(H)] = [\groupA : H]$.
\end{theorem}

Now, we turn to dihedral groups.

\begin{lemma}
	\label{lem:one_refl_two_rot-implies-two_refl_one_rot}
	Let $\group \leq \DihedralGroup{n_1} \direct \DihedralGroup{n_2} \direct \DihedralGroup{n_3}$ be a 2-injective subdirect product
	and $n_i \notin \set{1,2}$ for all $i \in [3]$.
	If $(\reflA_1,\rotA_2,\rotA_3) \in \group$ for a rotation~$\reflA_1$ and reflections~$\rotA_2$ and $\rotA_3$ then 
	there is an element $(\rotA_1',\reflA_2', \reflA_3') \in \group$
	where~$\rotA_1'$ is a rotation
	and~$\reflA_2'$ and~$\reflA_3'$ are reflections.
\end{lemma}
\begin{proof}
	For each $i \in \set{2,3}$ we pick a reflection $\reflA_i$ in $\DihedralGroup{n_i}$.
	Then there are two elements $h_i = (g_{i,1}, g_{i,2}, g_{i,3}) \in \group$
	where $i \in \set{2,3}$ and $g_{i,i} = \refl_i$.
	We make the following case distinction:
	\begin{enumerate}
		\item One of the $h_i$ already has the desired form of having a reflection in the first component and reflections elsewhere. In this case we are done.
		\item If one $h_i = (\reflB_1,\reflB_2,\reflB_3) $ only consists of reflections, then 
		\[h_i(\reflA_1, \rotA_2, \rotA_3) = (\reflB_1\reflA_1,\reflB_2\rotA_2,\reflB_3\rotA_3) = (\rotA_1', \reflA_2', \reflA_3') \in \group\]
		has the desired form.
		\item If $h_2 = (\rotB_1, \reflB_2, \rotB_3)$, we make another case distinction on $h_3$.
		If $h_3 = (\rotB_1', \rotB_2', \reflB_3')$ then
		\[h_2h_3 = (\rotB_1\rotB_1',\reflB_2\rotB_2',\rotB_3\reflB_3') = 
		(\rotA_1', \reflA_2', \reflA_3') \in \group\]
		has the desired form.
		Otherwise $h_3 = (\reflB_1', \rotB_2',\reflB_3')$ and
		\[h_2h_3 = (\rotB_1\reflB_1',\reflB_2\rotB_2',\rotB_3\reflB_3')\]
		only consists of reflections. Thus, we reduced to Case~2.
		
		\item Otherwise $h_2 = (\reflB_1, \reflB_2, \rotB_3)$ and we perform the same case distinction on $h_3$. 
		If $h_3 = (\rotB_1', \rotB_2', \reflB_3')$
		then $h_2h_3$ only consists of reflections
		and we reduced to Case~2.
		If otherwise $h_3 = (\reflB_1', \rotB_2',\reflB_3')$,
		then $h_2h_3$ has the desired form.
		\qedhere
	\end{enumerate}
\end{proof}

\begin{lemma}
	\label{lem:rot-refl-from-kernels}
	Let $\group \leq \DihedralGroup{n_1} \times \DihedralGroup{n_2} \times \DihedralGroup{n_3}$ be a 2-injective subdirect product
	and $n_i \notin \set{1,2}$ for all $i \in [3]$.
	If $\outproj_1(H_1)$ contains no reflections for all $i \in [3]$,
	then $\group$ is a reflect-or-rotate group.
\end{lemma}
\begin{proof}	
	If there is a violating element in $\group$,
	we can assume by Lemma~\ref{lem:one_refl_two_rot-implies-two_refl_one_rot}
	that (up to reordering of the factors) $(\rotA_1,\reflA_2,\reflA_3) \in \group$.
	In particular $\ord{\rotA_1} \leq 2$ by Lemma~\ref{lem:2-inj-sub-basics}\ref{part:greater:order:two:nontrivial:lemma}.
	If $\rotA_1 = 1$ is one, there is a reflection in $\outproj_1(H_1)$.
	Otherwise $\ord{\rotA_1}= 2$ and in particular $n_1 > 3$
	($\DihedralGroup{3}$ contains no rotation of order $2$).
	Consider a root $\rotB_1$ of $\rotA_1$ with $\rotB_1^k = \rotA_1$
	(it exists because $n_1 > 3$), $k>1$,
	and an element $(\rotB_1, g_2, g_3) \in \group$.
	Now $\ord{\rotB_1} > 2$ and thus one of $g_2$ and $g_3$
	is a rotation of order $>2$, say $g_2 = \rotB_2$.
	Then $(\rotA_1, \reflA_2, \reflA_2)(\rotB_1, \rotB_2, g_3)^k
	= (1, \reflA_2 \rotB_2^k, \reflA_3 g_3^k)$.
	Clearly $\reflA_2 \rotB_2^k$ is a reflection, hence $\reflA_3 g_3^k$ is a reflection and we found a reflection in $\outproj_1(H_1)$.
\end{proof}

\begin{lemma}
	\label{lem:kernel-no-reflection-124}
	Let $\group \leq \DihedralGroup{n_1} \times \DihedralGroup{n_2} \times \DihedralGroup{n_3}$ be a 2-injective subdirect product
	and $n_i \notin \set{1,2,4}$ for some $i \in [3]$
	and $n_j > 2$ for all $j \in [3]$.
	Then $H_i$ contains only rotations.
\end{lemma}
\begin{proof}
	By symmetry we fix w.l.o.g.~$i=1$.
	For sake of contradiction
	let $(1, g_2,g_3) \in H_1$ be an element that is not a rotation. Then~$g_2$ or~$g_3$ is not a rotation.
	Because $H_1$ defines an isomorphism between $\proj_2(H_1)$ and $\proj_3(H_i)$,
	$g_2 = \reflA_2$ and $g_3 = \reflA_3$ must both be reflections.
	
 	Let $\rot_1 \in \DihedralGroup{n_1}$ be a rotation 
	such that
	$\ord{\rot_1} > 2$ and~$\ord{\rot_1^2} > 2$,
	which is possible because $n_1 \notin\set{1,2,4}$.
	Then there is an element $(\rot_1, h_2, h_3) \in \group$.
	The elements $h_2$ and $h_3$ cannot be both of order at most $2$
	by Lemma~\ref{lem:2-inj-sub-basics}\ref{part:greater:order:two:nontrivial:lemma}.
	So assume w.l.o.g.~that $\ord{h_2} > 2$ and
	hence $h_2 = \rot_2$ is a rotation.
	If $h_3 = \rotA_3$ is a rotation, too, then
	\[(\rot_1,\rot_2, \rot_3)(1, \reflA_2, \reflA_3) = (\rot_1, \rot_2\reflA_2, \rot_3 \reflA_3) \in \group\]
	yields a contradiction to Lemma~\ref{lem:2-inj-sub-basics}\ref{part:greater:order:two:nontrivial:lemma},
	because $\ord{\rotA_1} > 2$ and
	$\rotA_i\reflA_i$ are reflections of order~$2$ (for $i \in \set{2,3}$).

	So, finally $h_3 = \reflB_3$ must be a reflection.
	But then consider
	\[(\rot_1, \rot_2, \reflB_3)^2 (1, \reflA_2, \reflA_3) = ({\rot_1}^2, {\rot_2}^2\reflA_2, \reflA_3 )\in \group\]
	and note the last two components are reflections of order $2$
	but the first component has order $>2$.
	This again contradicts Lemma~\ref{lem:2-inj-sub-basics}\ref{part:greater:order:two:nontrivial:lemma}.
\end{proof}

\begin{corollary}
	\label{cor:reflect-or-rotate-not-124}
	Let $\group \leq \DihedralGroup{n_1} \times \DihedralGroup{n_2} \times \DihedralGroup{n_3}$ be a 2-injective subdirect product
	and $n_i \notin \set{1,2,4}$ for all $i \in [3]$.
	Then $\group$ is a reflect-or-rotate group.
\end{corollary}
\begin{proof}
	This follows from Lemmas~\ref{lem:rot-refl-from-kernels} and~\ref{lem:kernel-no-reflection-124}.
\end{proof}

The previous lemma classifies all $2$-injective subdirect products of dihedral groups 
as reflect-or-rotate groups if $\DihedralGroup{1}$, $\DihedralGroup{2}$ and
$\DihedralGroup{4}$ are not involved as one of the factors.
We now analyse how rotations and reflections of the single factors
can be combined if these groups are involved.

\subsection{The CFI and the Double CFI Group}
\label{sec:CFI-and-double-CFI}
One exception where rotations can be combined with reflections
will be of particular interest.

\begin{definition}[CFI Groups]
	We call the following group $\CFIgroup < \DihedralGroup{1}^3$ the
	\defining{CFI group}:
	\[\CFIgroup := \setcondition{(g_1,g_2,g_3) \in \DihedralGroup{1}^3}{g_1g_2g_3 = 1} \]
	that is $\CFIgroup$ is the group consisting of those triples that contain an even number
	of the non-trivial elements of $\DihedralGroup{1}$.
	We call the wreath product $\DoubleCFIgroup := \CFIgroup \wreath \CyclicGroup{2}$
	the \defining{double CFI group}.
\end{definition}

Figure~\ref{fig:double-cfi} depicts two colored graphs
whose automorphism groups are isomorphic to~$\CFIgroup$ and~$\DoubleCFIgroup$, respectively.
It is straightforward to verify that the CFI group is a $2$-injective subdirect product. 
We now show how the double CFI group can be realized as a $2$-injective subdirect product $\DoubleCFIgroup < \DihedralGroup{4}^3$:
Let $\rot \in \DihedralGroup{4}$ be a rotation of order $4$ and $\reflA,\reflA' \in \DihedralGroup{4}$ be reflections such that $\reflA\reflA' = \rot^2$
(that is $\set{\reflA,\reflA'}$ is a conjugacy class), and $\reflB \in \DihedralGroup{4}$ be a reflection with $\reflB \notin\set{\reflA, \reflA'}$. Then
\begin{align*}
	\DoubleCFIgroup &\iso \langle (1, \refl, \refl), (\refl, 1, \refl), (1, \refl', \refl'), (\refl', 1, \refl'), (\reflB, \reflB, \reflB) \rangle\\
	& =\langle (1, \refl, \refl), (\refl, 1, \refl), (\rot, \rot, \beta) \rangle.
\end{align*}
The elements with the two reflections $\reflA$ and $\reflA'$ generate
the two independent CFI groups, respectively, 
and $(\reflB, \reflB, \reflB)$ exchanges the two CFI groups (swaps top with bottom in the figure).

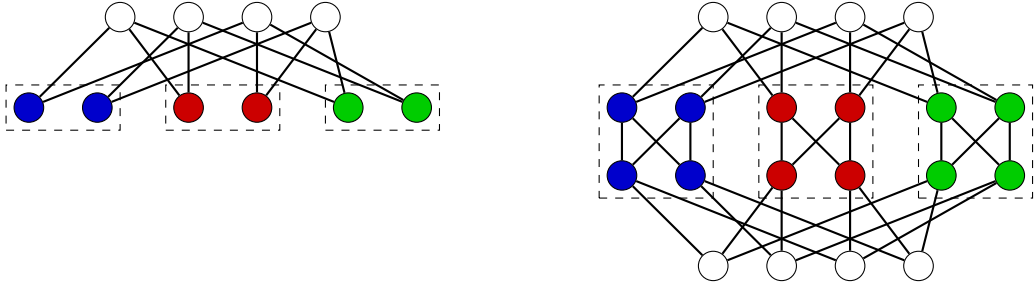
\begin{figure}
\centering	
\tikzstyle{normalvertex}=[circle, draw]

\begin{tikzpicture}[scale =0.6]
\begin{scope}[shift={(-13,0)}]

  \node[style=normalvertex,fill=red!80!black,label=below:{$$}] (a0) at (1,5) {};
  \node[style=normalvertex,fill=red!80!black,label=below:{$$}] (b0) at (-0.5,5) {};
  
  \node[style=normalvertex,fill=blue!80!black,label=above:{$$}] (a2) at (-2.5,5) {};
  \node[style=normalvertex,fill=blue!80!black,label=above:{$$}] (b2) at (-4,5) {};
  
  \node[style=normalvertex,fill=green!80!black,label=below:{$$}] (a1) at (4.5,5) {};
  \node[style=normalvertex,fill=green!80!black,label=below:{$$}] (b1) at (3,5) {};
  
  \node[style=normalvertex,label=180:{{\small $$}}] (v1) at (-2,7) {};
  \node[style=normalvertex,label=180:{{\small $$}}] (v2) at (-0.5,7) {};
  \node[style=normalvertex,label=180:{{\small $$}}] (v3) at (1,7) {};
  \node[style=normalvertex,label=180:{{\small $$}}] (v4) at (2.5,7) {};

   \path[thick]
   (v1) edge (b0)    (v1) edge (b2)    (v1) edge (b1)    (v2) edge (b0)
   (v2) edge (a2)    (v2) edge (a1)    (v3) edge (a0)    (v3) edge (b2)
   (v3) edge (a1)    (v4) edge (a0)    (v4) edge (a2)    (v4) edge (b1);

 \draw[dashed]  (-1,5.5) rectangle (1.5,4.5);
 \draw[dashed]  (-4.5,5.5) rectangle (-2,4.5);
 \draw[dashed]  (5,5.5) rectangle (2.5,4.5);
\end{scope}
  \node[style=normalvertex,fill=red!80!black,label=below:{$$}] (a0) at (1,5) {};
  \node[style=normalvertex,fill=red!80!black,label=below:{$$}] (b0) at (-0.5,5) {};
  
  \node[style=normalvertex,fill=blue!80!black,label=above:{$$}] (a2) at (-2.5,5) {};
  \node[style=normalvertex,fill=blue!80!black,label=above:{$$}] (b2) at (-4,5) {};
  
  \node[style=normalvertex,fill=green!80!black,label=below:{$$}] (a1) at (4.5,5) {};
  \node[style=normalvertex,fill=green!80!black,label=below:{$$}] (b1) at (3,5) {};
  
  \node[style=normalvertex,label=180:{{\small $$}}] (v1) at (-2,7) {};
  \node[style=normalvertex,label=180:{{\small $$}}] (v2) at (-0.5,7) {};
  \node[style=normalvertex,label=180:{{\small $$}}] (v3) at (1,7) {};
  \node[style=normalvertex,label=180:{{\small $$}}] (v4) at (2.5,7) {};
  
  \node[style=normalvertex,fill=red!80!black,label=below:{$$}] (as0) at (1,3.5) {};
  \node[style=normalvertex,fill=red!80!black,label=below:{$$}] (bs0) at (-0.5,3.5) {};  
  
  \node[style=normalvertex,fill=blue!80!black,label=above:{$$}] (as2) at (-2.5,3.5) {};
  \node[style=normalvertex,fill=blue!80!black,label=above:{$$}] (bs2) at (-4,3.5) {};

    \node[style=normalvertex,fill=green!80!black,label=below:{$$}] (as1) at (4.5,3.5) {};
  \node[style=normalvertex,fill=green!80!black,label=below:{$$}] (bs1) at (3,3.5) {};

    \node[style=normalvertex,label=180:{{\small $$}}] (vs1) at (-2,1.5) {};
  \node[style=normalvertex,label=180:{{\small $$}}] (vs2) at (-0.5,1.5) {};
  \node[style=normalvertex,label=180:{{\small $$}}] (vs3) at (1,1.5) {};
  \node[style=normalvertex,label=180:{{\small $$}}] (vs4) at (2.5,1.5) {};
  
   \path[thick]
   (v1) edge (b0)    (v1) edge (b2)    (v1) edge (b1)    (v2) edge (b0)
   (v2) edge (a2)    (v2) edge (a1)    (v3) edge (a0)    (v3) edge (b2)
   (v3) edge (a1)    (v4) edge (a0)    (v4) edge (a2)    (v4) edge (b1);
   
   \path[thick]
   (vs1) edge (bs0)   (vs1) edge (bs2)   (vs1) edge (bs1)   (vs2) edge (bs0)
   (vs2) edge (as2)   (vs2) edge (as1)   (vs3) edge (as0)   (vs3) edge (bs2)
   (vs3) edge (as1)   (vs4) edge (as0)   (vs4) edge (as2)   (vs4) edge (bs1);
   
   \foreach \x in {0,1,2}
   {
   \path[thick]
   (a\x) edge (as\x)
   (a\x) edge (bs\x)
   (b\x) edge (as\x) 
   (b\x) edge (bs\x);
   }
   
 \draw[dashed]  (-1,5.5) rectangle (1.5,3);
 \draw[dashed]  (-4.5,5.5) rectangle (-2,3);
 \draw[dashed]  (5,5.5) rectangle (2.5,3);
\end{tikzpicture}
\caption{Two vertex-colored graphs whose automorphism groups are the CFI group (left) and the double CFI group (right).}\label{fig:double-cfi}
\end{figure}

\begin{lemma}
	\label{lem:2inj-sub-D4-D4-D4}
	There is (up to isomorphism) exactly one $2$-injective subdirect product
	of $\DihedralGroup{4}^3$, which is not a rotate-or-reflect group,
	namely the double CFI group.
\end{lemma}
\begin{proof}
	The group $\DihedralGroup{4}$ has the following normal subgroups:
	$\DihedralGroup{4}, \DihedralGroup{2}, \CyclicGroup{4}, \CyclicGroup{2}, \CyclicGroup{1}$.
	Note that the normal subgroup $\CyclicGroup{2}$ of $\DihedralGroup{4}$
	is generated by the 180 degree rotation (and not by a reflection, because reflections do not generate normal subgroups).
	
	Now the $H_i$ must be isomorphic to one of these normal subgroups and $|\proj_i(H)| = |\proj_j(H)|$ by Theorem~\ref{thm:two-inj-sub-kernels}. 
	If all $H_i$ are cyclic groups, then by Lemma~\ref{lem:rot-refl-from-kernels}, $\group$ is a rotate-or-reflect group.
	
	We first show, that $H_1 \iso \DihedralGroup{4}$ is contradictory.
	Let
	$\rotA \in \DihedralGroup{4}$ be a rotation of order~$4$.
	Now $g = (1, \rotA, \rotB) \in H_1 \subseteq \group$
	where $\ord{\rotB} = 4$ because $H_1$ defines an isomorphism between its projections to second and third component. 
	Since $\proj_1(H) = \DihedralGroup{4}$
	there is another element $h = (\reflA, h_2, h_3) \in H_2 \cup H_3 \subseteq \group$.
	W.l.o.g.~assume $h = (\reflA, \reflB, 1) \in H_2$.
	Then $gh = (\reflA, \reflB \rotA, \rotB) \in \group$
	which is impossible by Lemma~\ref{lem:2-inj-sub-basics}\ref{part:greater:order:two:nontrivial:lemma}.
	
	Second, assume that $H_1\iso \DihedralGroup{2}$ but both $H_2$ and $H_3$
	are not isomorphic to~$\DihedralGroup{2}$.
	If $H_2 \iso \CyclicGroup{2}$ (respectively $\CyclicGroup{1}$), then
	$\proj_1(\langle H_1, H_2 \rangle) = \CyclicGroup{2}$
	(respectively $\CyclicGroup{1}$)
	and $\DihedralGroup{2} \leq \proj_i(H)$ 
	and in particular $|\proj_i(H)| \geq 4$ for $i \in \set{2,3}$.
	Then $|\proj_1(H)| \geq 4$ by Theorem~\ref{thm:two-inj-sub-kernels}
	and so $H_3 \iso \CyclicGroup{4}$
	($H_3 \not\iso \DihedralGroup{2}$ by assumption
	and $H_3 \iso \DihedralGroup{4}$ was already proved inconsistent).
	Then $|\proj_1(H)| = 4$ but
	$|\proj_2(H)| = |\DihedralGroup{4}| = 8$
	contradicting Theorem~\ref{thm:two-inj-sub-kernels}.
	By symmetry the case $H_2 \iso H_3 \iso \CyclicGroup{4}$ remains.
	Then $\proj_1(H) \iso \CyclicGroup{4}$ but $\proj_2(H) \iso \DihedralGroup{4}$,
	which contradicts Theorem~\ref{thm:two-inj-sub-kernels} again.
	
	Lastly, we assume that $H_1 \iso H_2 \iso \DihedralGroup{2}$.
	Then $\proj_3(\langle H_1, H_2 \rangle)$ is either $\DihedralGroup{2}$ or $\DihedralGroup{4}$ (depending on whether $H_1$ and $H_2$ contain reflections
	from the same conjugacy class or not).
	In the case that $\proj_3(H_1) \neq \proj_3(H_2)$, $H_1$ and $H_2$ use
	reflections from different conjugacy classes of~$\DihedralGroup{4}$.
	Let $\reflA$ and $\reflA'$ be these reflections such that
	$(1,\reflB,\reflA) \in H_1$ and $(\reflB', 1, \reflA') \in H_2$.
	Then $(1,\reflB,\reflA)(\reflB', 1, \reflA') = (\reflB', \reflB, \rotA) \in \group$ where $\rotA$ is the rotation of order $4$
	and thus contradicts Lemma~\ref{lem:2-inj-sub-basics}\ref{part:greater:order:two:nontrivial:lemma}.
	
	Finally, let $\proj_3(H_1) = \proj_3(H_2)$, that is both kernels use the same reflections of $\DihedralGroup{4}$.
	Let $\reflA \in \DihedralGroup{4}$ be one of these reflections.
	Then $(1,\reflB, \reflA) \in H_1$ and $(\reflB', 1, \reflA) \in H_2$
	and $(1,\reflB,\reflA)(\reflB', 1, \reflA)=(\reflB', \reflB, 1) \in H_3$.
	So $H_3$ contains reflections and
	cannot be isomorphic to $\DihedralGroup{4}$ by the prior reasoning,
	i.e.,~$H_3 \iso \DihedralGroup{2}$.

	So we have shown that unless $\group$ is a rotate-or-reflect group,
	we have $H_i \iso \DihedralGroup{2}$ and
	$\proj_i(H) = \DihedralGroup{2}$ for all $i \in [3]$.
	That is, all $H_i$ use reflections of the same conjugacy class of $\DihedralGroup{4}$ for each component
	(there are two embeddings of $\DihedralGroup{2}$ in $\DihedralGroup{4}$).
	We apply an isomorphism to $\group$,
	such that the embedding is the same for each component.
	
	Let $\reflA, \reflA' \in \DihedralGroup{4}$
	be the two reflections of this embedding of 
	$\DihedralGroup{2}$ into $\DihedralGroup{4}$
	and $\rotA \in \DihedralGroup{4}$ be a rotation of order $4$.
	We assume that $(1,\reflA, \reflA), (1,\reflA', \reflA') \in H_1$ and
	$(\reflA,1,\reflA), (\reflA', 1, \reflA') \in H_2$.
	In the case that $\reflA$ gets combined with $\reflA'$ we apply isomorphisms to the first and/or second factor exchanging $\reflA$ and $\reflA'$.
	Then $H_3$ also combines $\reflA$ with $\reflA$.
	
	Now $(\rotA, g, h) \in \group$, one of $g$ and $h$ has to be $\rotA$ or $\inv{\rotA}$, too, by Lemma~\ref{lem:2-inj-sub-basics}\ref{part:greater:order:two:nontrivial:lemma}
	and the other one cannot be $1$
	because $H_2, H_3 \not\iso \DihedralGroup{4}$. 
	Assume, w.l.o.g.,~that $g = \rotA$ (if it is $\inv{\rotA}$, we apply another
	isomorphism to the second factor exchanging $\rotA$ and $\inv{\rotA}$
	but which is constant on the reflections).
	Let $\reflB$ and $\reflB'$ be the reflections not in the conjugacy class of $\reflA$. Then $\reflA \reflB, \reflA\reflB', \reflA'\reflB, \reflA'\reflB' \in \set{\rotA, \inv{\rotA}}$.
	Assume that $h=\rotB$ is a rotation.
	Then $(\rotA, \rotA, \rotB)(1,\reflA, \reflA) = (\rotA, \rotA\reflA, \rotB\reflA) \in \group$ which contradicts
	Lemma~\ref{lem:2-inj-sub-basics}\ref{part:greater:order:two:nontrivial:lemma}.
	
	So $h$ is a reflection. If $h \in \{\reflA, \reflA'\}$, then 
	$(\rotA, \rotA, h)(1,h,h) = (\rotA, \rotA h, 1) \in H_3$.
	This is a contradiction
	because $H_3$ defines an isomorphism between $\proj_1(H_3)$ and $\proj_2(H_3)$ and no isomorphism can
	map a rotation $\rotA$ to a reflection $\rotA h$.
	Hence $h \in \set{\reflB, \reflB'}$, say $h = \reflB$,
	and finally $\groupB := \langle (\rotA, \rotA, \reflB), (\reflA, 1, \reflA), (1, \reflA, \reflA) \rangle$ is isomorphic to the double CFI group.
	Note that $|H| = |\langle H_1, H_2, H_3 \rangle| = 16$,
	that $|\groupB| = 32$, and that
	$[\group : H] = [\DihedralGroup{4} : \proj_1(H)] = [\DihedralGroup{4} : \DihedralGroup{2} ] = 2 = [\groupB : H]$
	by Theorem~\ref{thm:two-inj-sub-kernels}.
	Hence $|\groupB| = |\group|$ and $\group \iso \groupB$ is isomorphic to the double CFI group.
\end{proof}

\subsection{Products Involving $\DihedralGroup{4}$, $\DihedralGroup{2}$, or $\DihedralGroup{1}$}

\begin{lemma}
	\label{lem:2inj-sub-D1-or-D2-other}
	Let $\group \leq \DihedralGroup{i} \direct \DihedralGroup{j} \direct \DihedralGroup{k}$
	with $i \in [2]$ and $j, k > 2$ be a $2$-injective subdirect product.
	Then $\outproj_1(\group)$ is a rotate-or-reflect group.
\end{lemma}
\begin{proof}
	Suppose $(g,\rotA,\reflA) \in \group$.
	By Lemma~\ref{lem:2-inj-sub-basics}\ref{part:greater:order:two:nontrivial:lemma}, $\rotA$ is of order~$2$
	because $g$ and $\reflA$ are of order at most~$2$.
	Let $\rotB \in \DihedralGroup{j}$ be a rotation not of order~$2$
	and $(h, \rotB, \rotB') \in \group$ an element whose second component is~$s$
	(the last component must be a rotation of order larger than~$2$).
	Then $(gh, \rotA\rotB, \reflA\rotB') \in \group$,
	$gh$ and $\reflA\rotB'$ are of order at most~$2$,
	and $\rotA\rotB$ is a rotation with $\ord{\rotA\rotB}>2$.
	This is a contradiction to Lemma~\ref{lem:2-inj-sub-basics}\ref{part:greater:order:two:nontrivial:lemma}.
\end{proof}

\begin{lemma}
	\label{lem:no-2inj-sub-DCC-not2}
	There are no $2$-injective subdirect products of
	$\DihedralGroup{i} \direct G_2 \direct G_3$
	for $i > 2$ if $G_2$ and $G_3$ are abelian.
\end{lemma}
\begin{proof}
	Let $\group \leq \DihedralGroup{i} \direct G_2 \direct G_3$ be a $2$-injective subdirect product and
	$\refl, \rot \in \DihedralGroup{i}$ such that $\refl\rot \neq \rot\refl$.
	Two such elements exist because $i > 2$ and $\DihedralGroup{i}$ is non-abelian.
	Consider two elements
	$g =(\refl, g_2, g_3),h= (\rot, h_2, h_3) \in \group$ for some $g_2,h_2 \in G_2$ and $g_3,h_3 \in G_3$.
	Now $gh = (\refl\rot, g_2h_2,g_3h_3)$ and
	$hg = (\rot\refl, g_2h_2,g_3h_3) \neq gh$ contradicting $2$-injectivity of $\group$.
\end{proof}

\begin{lemma}
	\label{lem:no-2-injs-D4-D4-Di}
	There are no $2$-injective subdirect products of
	$\DihedralGroup{4} \direct \DihedralGroup{4} \direct \DihedralGroup{i}$ for $i \notin \set{1,2,4}$.
\end{lemma}
\begin{proof}
	Let $\group \leq \DihedralGroup{4} \direct \DihedralGroup{4} \direct \DihedralGroup{i}$ 
	be a $2$-injective subdirect product and $i \notin \set{1,2,4}$.
	Let $\rot \in \DihedralGroup{i}$
	be a rotation with $\ord{\rot} = i$
	and $(g,h,\rot) \in \group$ for some $g,h \in \DihedralGroup{4}$.
	For all $g,h \in \DihedralGroup{4}$ it holds that $\ord{g},\ord{h} \in \set{1,2,4}$
	contradicting Lemma~\ref{lem:2-inj-sub-basics}.
\end{proof}

\begin{lemma}
	\label{lem:2inj-sub-D4-other}
	Let $\group \leq \DihedralGroup{4} \direct \DihedralGroup{n_2} \direct \DihedralGroup{n_3}$ be a $2$-injective subdirect product and $n_2,n_3 \notin \{1,2,4\}$.
	Then $\group$ is a rotate-or-reflect group.
\end{lemma}
\begin{proof}
	We apply Lemma~\ref{lem:rot-refl-from-kernels}.
	We have to show 
	that the $\outproj_i(H_i)$ contain no reflection.
	By Lemma~\ref{lem:kernel-no-reflection-124} this is the case for $H_2$ and $H_3$.
	Assume that $(1, \refl_2, \refl_3) \in H_1$
	and note that $\proj_i(H_1) \normal \DihedralGroup{n_i}$
	for $i \in \set{2,3}$.
	So $\proj_i(H_1) \in \set{\DihedralGroup{n_i}, \DihedralGroup{n_i/2}}$
	for $i \in \set{2,3}$
	(if $n_i$ is odd, $\DihedralGroup{n_i/2}$ of course does not exist).
	Recall that by Theorem~\ref{thm:two-inj-sub-kernels} we have
	$[\DihedralGroup{n_i} : \proj_i(H)] = [\group : H]$.
	If $\proj_2(H) = \DihedralGroup{n_2}$
	(which in particular holds if $\proj_2(H_1) = \DihedralGroup{n_2}$)
	then
	\[1 = [\DihedralGroup{n_2} : \proj_2(H)] = [\DihedralGroup{n_1} : \proj_1(H)]\]
	and so $\proj_1(H) = \DihedralGroup{4}$,
	i.e.,~one of $H_2$ and $H_3$ contains reflections,
	which is a contradiction.
	
	Otherwise $\proj_2(H) = \proj_2(H_1) = \DihedralGroup{n_2/2}$,
	by symmetry $\proj_3(H) = \proj_3(H_1) = \DihedralGroup{n_3/2}$,
	and $2 = [\DihedralGroup{n_1} : \proj_1 (H)]$.
	So $\proj_1(H) = \langle \proj_1(H_2) , \proj_1(H_3) \rangle \in \{\DihedralGroup{2}, \CyclicGroup{4}\}$.
	If $\proj_1(H) = \DihedralGroup{2}$ then
	$H_2$ or $H_3$ again contain reflections.
	Thus, $\proj_1(H) = \CyclicGroup{4}$.
	Let $\rot \in \CyclicGroup{4} < \DihedralGroup{4}$ be a rotation of order $4$.
	Now, there is an element $(\rot, \rot', 1) \in H_3$.
	But then $(\rot, \rot',1)(1, \refl_2, \refl_3) = (\rot, \rot'\refl_2,\refl_3) \in \group$,
	which is a contradiction to Lemma~\ref{lem:2-inj-sub-basics}\ref{part:greater:order:two:nontrivial:lemma}.
\end{proof}

\begin{proof}[Proof of Theorem~\ref{thm:classification-2inj-sub-dihedral}]
	Let $\group \leq \DihedralGroup{n_1} \direct \DihedralGroup{n_2} \direct \DihedralGroup{n_3}$ be a $2$-injective subdirect product.
	
	If $n_i \leq 2$ for all $i \in [3]$, then $\group$ is abelian.
	If $n_i \leq 2$ for exactly one $i \in [3]$,
	then $\outproj_i(\group)$ is a rotate-or-reflect group
	by Lemma~\ref{lem:2inj-sub-D1-or-D2-other}.
	The case that $n_i \leq 2$ for exactly two $i$
	is a contradiction to Lemma~\ref{lem:no-2inj-sub-DCC-not2}
	because $\DihedralGroup{1}$ and $\DihedralGroup{2}$
	are the only abelian dihedral groups.
	
	Lastly, consider the case that $n_i > 2$ for all $i \in [3]$.
	If $n_i = 4$ for all $i \in [3]$,
	then $\group$ is a rotate-or-reflect group or the double CFI group
	by Lemma~\ref{lem:2inj-sub-D4-D4-D4}.
	The case $n_i = n_j = 4$ and $n_k \notin\set{1,2,4}$ for $\set{i,j,k} =[3]$
	is impossible due to Lemma~\ref{lem:no-2-injs-D4-D4-Di}.
	If $n_i = 4$ for exactly one $i \in [3]$,
	then $n_j \notin\set{1,2,4}$ for every $j \neq i$.
	Consequently, $\group$ is a rotate-or-reflect group by
	Lemma~\ref{lem:2inj-sub-D4-other}.
	If $n_i \neq 4$ for all $i \in [3]$,
	then $\group$ is a rotate-or-reflect group by Corollary~\ref{cor:reflect-or-rotate-not-124}.
\end{proof}

\subsection{Combinations with Cyclic Groups}
As last step, we consider $2$-injective subdirect products
of a mixture of dihedral and cyclic groups.

\begin{lemma}
	\label{lem:no-2-inj-sub-CDD-not-124}
	There are no $2$-injective subdirect products of 
	$\CyclicGroup{i} \direct \DihedralGroup{j} \direct \DihedralGroup{k}$
	for $i \notin \set{1,2,4}$.
\end{lemma}
\begin{proof}
	The proof is essentially a simpler version of the proof of Lemma~\ref{lem:kernel-no-reflection-124}.
	We show now, that the reflections in the dihedral group cannot be combined with the rotations in the cyclic group.
	
	We argue first that there is an element $(\rot, \refl_2, \refl_3) \in \group$.
	Let $\refl \in \DihedralGroup{j}$ be a reflection
	and $(\rot, \refl, g) \in \group$ for some $\rot \in \CyclicGroup{i}$ and $g \in \DihedralGroup{k}$.
	If $g$ is a reflection, we are done.
	Otherwise $g = \rot_3$
	and consider a reflection $\refl' \in \DihedralGroup{k}$
	and some element $(\rot', g', \refl') \in \group$.
	Again, if $g'$ is a reflection, we are done.
	Otherwise $g' = \rot_2$
	and $(\rot, \refl, \rot_3)(\rot', \rot_2, \refl') = (\rot\rot', \refl\rot_2, \rot_3\refl')$ is the desired group element.
	By Lemma~\ref{lem:2-inj-sub-basics}\ref{part:greater:order:two:nontrivial:lemma},
	$\ord{\rot} \leq 2$.
	
	Let $\rotB \in \CyclicGroup{i}$ be a rotation with $\ord{\rotB} \notin \set{1,2,4}$
	and $(\rotB, g, h) \in \group$ be some group element.
	By Lemma~\ref{lem:2-inj-sub-basics}\ref{part:greater:order:two:nontrivial:lemma}
	one of $g$ and $h$ must have order $> 2$, say w.l.o.g.~$g= \rot_2$.
	If $h = \rot_3$ is a rotation, too, then
	$(\rot, \refl_2, \refl_3) (\rotB, \rot_2, \rot_3) = (\rot\rotB, \refl_2\rot_2, \refl_3\rot_3) \in \group$.
	This contradicts
	Lemma~\ref{lem:2-inj-sub-basics}\ref{part:greater:order:two:nontrivial:lemma},
	 because $\ord{\rot\rotB} > 2$.
	
	If $h = \reflB_3$ is another reflection,
	then $(\rot, \refl_2, \refl_3)(\rotB, \rot_2, \reflB_3)^2 =
	(\rot\rotB^2, \refl_2\rot_2^2, \refl_3) \in \group$.
	As before, $\ord{\rot\rotB^2} > 2$ but the other components are reflections contradicting
	Lemma~\ref{lem:2-inj-sub-basics}\ref{part:greater:order:two:nontrivial:lemma}.
\end{proof}

\begin{lemma}
	\label{lem:2-inj-sub-C2DD-rot-refl-group}
	Let $\group \leq \CyclicGroup{i} \direct \DihedralGroup{j} \direct \DihedralGroup{k}$ be a $2$-injective subdirect product, $i \leq 2$ and $j, k > 2$.
	Then $\outproj_1(\group)$ is a rotate-or-reflect group.
\end{lemma}
\begin{proof}
	The case $i = 2$ follows immediately from Lemma~\ref{lem:2inj-sub-D1-or-D2-other}
	because $\CyclicGroup{2} = \DihedralGroup{1}$.
	The case $i = 1$ implies that $H = H_1 = \group$ (Theorem~\ref{thm:two-inj-sub-kernels}).
	Because $H_1$ defines an isomorphism
	between $\proj_2(H_1)$ and $\proj_3(H_1)$,
	$\outproj_1(\group)$ is a rotate-or-reflect-group.
\end{proof}

\begin{lemma}
	\label{lem:2-in-sub-C4D4D4-is-DCFI}
	Let $\group \leq \CyclicGroup{4} \direct \DihedralGroup{n_2} \direct \DihedralGroup{n_3}$ be a $2$-injective subdirect product and $n_i > 2$ for $i \in \set{2,3}$.
	Then $n_2 = n_3 = 4 $ and
	$\group \iso \DoubleCFIgroup \cap (\rotsubgroup{\DihedralGroup{4}} \direct \DihedralGroup{4} \direct \DihedralGroup{4})$.
\end{lemma}

\begin{proof}
	We first show that $[\group : H] > 1$.
	Assume that $[\group: H] = 1$.
	By Theorem~\ref{thm:two-inj-sub-kernels}
	it must hold that $\proj_i(H) = \DihedralGroup{n_i}$ for $i \in \set{2,3}$
	and hence that $\proj_i(H_1) = \DihedralGroup{n_i}$
	because $H_2$ and $H_3$ can only use rotations
	(they are normal subgroups and define isomorphisms to subgroups of~$\CyclicGroup{4}$).
	Furthermore, $\proj_1(H_i) = \CyclicGroup{4}$ for some $i \in \set{2,3}$,
	say w.l.o.g.~for $i = 2$.
	So there is an element $g = (1, \refl_2, \refl_3) \in H_1$
	and	an element $h = (\rot_1, \rot_2, 1) \in H_3$
	with $\ord{\rot_1} > 2$.
	But then $gh = (\rot_1, \refl_2\rot_2, \refl_3) \in \group$
	contradicting
	Lemma~\ref{lem:2-inj-sub-basics}\ref{part:greater:order:two:nontrivial:lemma}. 
	So $[\group : H] > 1$ and in particular $\proj_1(H) \in \set{\CyclicGroup{2}, \CyclicGroup{1}}$.
	
	Next, we want to show that there is an element $(1,\reflA_2', \reflA_3') \in \group$ and $n_2 \geq 4$.
	We know there is an element $g = (\rotA_1, \reflA_2, g_3) \in \group$.
	\begin{itemize}
		\item 
	Assume $g_3 = \reflA_3$ and then by
	Lemma~\ref{lem:2-inj-sub-basics}\ref{part:greater:order:two:nontrivial:lemma}
	$\ord{\rotA_1} = 2$.
	Let $\rotB_1 \in \CyclicGroup{4}$ be of order $4$
	and consider the element $h = (\rotB_1,h_2, h_3)$.
	Then $gh^2 = (\rotA_1\rotB_1^2, \reflA_2h_2^2, \reflA_3h_3^2) = (1, \reflA_2', \reflA_3').$
	
	One of $h_2$ and $h_3$ must be of order $>2$, so a rotation.
	Up to reordering of the factors $h_2 = \rotB_2$.
	Now $gh = (\rotA_1\rotB_1, \reflA_2\rotB_2, \reflA_3h_3) = (\inv{\rotB}, \reflA_2\rotB_2, \reflA_3h_3)$.
	Because $\ord{\inv{\rotB_1}} = 4$ and $\ord{\reflA_2\rotB_2} = 2$,
	we have $\ord{\reflA_3h_3} = 4$,
	in particular $h_3 = \reflA'_3$ is a reflection
	and hence $\rotB_2$ is of order $4$.
	So we have $n_2 \geq 4$.
	
	\item Assume $g_3 = \rotA_3$.
	Let $\rotB_2 \in \DihedralGroup{n_2}$ be a rotation of maximal order
	$\ord{\rotB_2} = n_2$.
	Then  $\rotB_2\reflA_2 \neq \reflA_2 \rotB_2$
	and there is an element $h = (\rotB_1, \rotB_2, h_3) \in \group$.
	If $h_3 = \rotB_3$, then
	$gh = (\rotA_1\rotB_1, \reflA_1\rotB_2, \rotA_3\rotB_3)$ and
	$hg = (\rotA_1\rotB_1, \rotB_2\reflA_1, \rotA_3\rotB_3)$
	which is a contradiction to $2$-injectivity.
	
	So $h_3 = \reflB_3$ must be a reflection.
	By Lemma~\ref{lem:2-inj-sub-basics}\ref{part:greater:order:two:nontrivial:lemma}
	$\ord{\rotB_1} = 4$ 
	because $\ord{\rotB_2 } > 2$ (and~$\CyclicGroup{4}$ contains only rotations of order $1,2,4$).
	Then also $\ord{\rotB_2} = 4$ implying $n_2 \geq 4$.
	Considering $gh = (\rotA_1\rotB_1, \reflA_1\rotB_2, \rotA_3\reflB_3)$
	and Lemma~\ref{lem:2-inj-sub-basics}\ref{part:greater:order:two:nontrivial:lemma} we obtain
	$\ord{\rotA_1\rotB_1}  \in \set{1,2}$.
	So either $\rotA_1\rotB_1 = 1$ or $\rotA_1\inv{\rotB_1} = 1$
	and one of $gh$ and $g\inv{h}$ is equal to $(1,\reflA_2', \reflA_3') \in \group$.
\end{itemize}
	It follows that $\proj_2(H_1) = \proj_2(H) = \DihedralGroup{n_2/2}$
	(because it contains reflections)
	and by Theorem~\ref{thm:two-inj-sub-kernels}
	then also $\proj_3(H_1) = \DihedralGroup{n_3/2}$
	and so $n_2 = n_3$ ($H_1$ defines an isomorphism)
	and $\proj_1(H) = \CyclicGroup{2}$. 
	
	We now show that $n_2 = n_3 = 4$.
	Assume $n_2 > 4$
	and let $h = (\rotB_1, \rotB_2, \reflB_3)$
	be the element as in the case $g_2 = \rotA_3$.
	Recall that $\ord{\rotB_2} = n_2$.
	There is a rotation $\rotB_2' \in \DihedralGroup{n_2}$
	such that $\ord{\rotB_2\rotB_2'} \notin\set{1,2,4}$.
	If $n_2 \neq 8$, we pick $\rotB_2' := \rotB_2$.
	Then $\ord{\rotB_2\rotB_2'} \in \set{n_2, n/2}$.
	If $n_2 = 8$, we pick $\rotB_2' := \rotB_2^2$
	and then $\ord{\rotB_2\rotB_2'} = 8$.
	We consider an element $g' = (\rotB_1', \rotB_2', g_3') \in \group$
	and $g_2' = \rotB_3'$ by 
	Lemma~\ref{lem:2-inj-sub-basics}
	because $\ord{\rotB_1'} \in \set{1,2,4}$.

	Finally, $hg' =  (\rotB_1\rotB_1', \rotB_2\rotB_2', \reflB_3\rotB_3')$,
	$\ord{\rotB_2\rotB_2'} \notin \set{1,2,4}$,
	and $\ord{\rotB_1\rotB_1'},\ord{\reflB_3\rotB_3'} \in \set{1,2,4}$
	contradicting
	Lemma~\ref{lem:2-inj-sub-basics}.
	The case $n_3 > 4$ follows by symmetry.

	Now, $\group$ contains elements generating a group isomorphic to
	$\DoubleCFIgroup \cap (\rotsubgroup{\DihedralGroup{4}} \direct \DihedralGroup{4} \direct \DihedralGroup{4})$.
	One can show this using the same argument
	with conjugacy classes of reflections
	as in the proof of Lemma~\ref{lem:2inj-sub-D4-D4-D4}.
	The order of this group is 16, the order of $H$ is $8$, so $\group$ cannot contain any other elements.
\end{proof}

\begin{proof}[Proof of Theorem~\ref{thm:classification-2inj-sub-dihedral-cyclic}]
	From Lemma~\ref{lem:no-2-inj-sub-CDD-not-124}  it follows that $n_1 \in \set{1,2,4}$.
	Assume $n_1 \leq 2$.
	If $n_2, n_3 >2$, then $\group$ is a rotate-or-reflect group
	due to Lemma~\ref{lem:2-inj-sub-C2DD-rot-refl-group}.
	If $n_2 \leq 2$, then $n_3 \leq 2$ by Lemma~\ref{lem:no-2inj-sub-DCC-not2} 
	and $\group$ is abelian.
	
	Assume otherwise $n_1 = 4$.
	If $n_2,n_3 \leq2$, then there is an element $(g_1,g_2,g_3) \in \group$
	and $\ord{g_1} = 4$, but $\ord{g_2}, \ord{g_3} \leq 2$
	contradicting Lemma~\ref{lem:2-inj-sub-basics}\ref{part:greater:order:two:nontrivial:lemma}.
	Otherwise $n_2,n_3 > 2$ by Lemma~\ref{lem:no-2inj-sub-DCC-not2}.
	Then $n_2=n_3=4$ and $\group$ is the claimed group by Lemma~\ref{lem:2-in-sub-C4D4D4-is-DCFI}.
	
	$2$-injective subdirect products of a non-abelian dihedral groups
	and two abelian groups do not exist
	by Lemma~\ref{lem:no-2inj-sub-DCC-not2}.
\end{proof}

\section{Normal Forms for Structures}
\label{sec:normal-forms}

In this section we describe a normal form for relational structures.
This normal form may alter the signature of the structure.
However, important for us is that the normal form is definable in choiceless polynomial time (CPT).
Also important for us is
that we have means within CPT to translate a canonical form of the normal form
back into a canonical form of the original structure. 

\begin{definition}
A \defining{canonization-preserving CPT-reduction} 
from a class of structures~$\mathcal{A}$ to a class of structures~$\mathcal{B}$
is a pair of CPT-interpretations $(\Phi, \Psi)$ with the following properties:
\begin{itemize}
	\item 
	$\Phi$ is a CPT-interpretation from $\mathcal{A}$-structures to $\mathcal{B}$-structures.
	\item 
	$\Psi$ is a CPT-interpretation from pairs of an $\mathcal{A}$-structure 
	and an ordered $\mathcal{B}$-structure to ordered $\mathcal{A}$-structures.
	\item Given a CPT-interpretation $\Theta$ from
	$\mathcal{B}$-structures to ordered $\mathcal{B}$ structures,
	i.e.,~a CPT-definable canonization procedure,
	then $\Psi((\Struct,\Theta(\Phi(\Struct))))$ is an ordered copy of $\Struct$
	for every $\mathcal{A}$-structure $\Struct$. 
\end{itemize}
We also say that $\mathcal{A}$ can be reduced canonization preservingly in CPT to $\mathcal{B}$
if there is a canonization-preserving CPT-reduction
from $\mathcal{A}$ to $\mathcal{B}$.
\end{definition}

That is, given a canonization-preserving CPT-reduction
from $\mathcal{A}$-structures to $\mathcal{B}$-structures
and a CPT-definable canonization procedure for $\mathcal{B}$-structures,
we obtain a CPT-definable canonization procedure for $\mathcal{A}$-structures.
Note that the reduction $\Psi$ not only takes the canonized $\mathcal{B}$-structure
but also the original $\mathcal{A}$-structure as input. This in particular allows for the possibility that $\Phi$ is not injective.
For example,
$\Phi$ can remove some parts of the input which can be canonized by $\Psi$
and thus produce the same $\mathcal{B}$-structure for different $\mathcal{A}$-structures.
For a definition of interpretation see e.g.~\cite{Grohe2017}
(interpretations are called transductions there).
We omit its definition here, because it is not relevant for this paper.\footnote{For the reader familiar with the concept, we note that all the reductions in this paper will be equivalence-free.}

The concept of a canonization-preserving reduction is akin to the concept of Levin reductions between problems in NP that allow us to pull certificates back. These reductions are thus certificate preserving.

We need to add a note on the signatures of the $\mathcal{A}$ and $\mathcal{B}$-structures:
in what follows, we fill fix a signature $\tau$ of the $\mathcal{A}$ structures.
However, we would like that the signature of $\Phi(\Struct)$ not only depends on $\tau$ but also on $\Struct$.
For example, assume we want to insert a relation $\rel_\colclass$ for every color class of $\Struct$ yielding $|\colorclasses{\Struct}|$ many new relations.
But this would make a definition of $\Phi$ impossible, because a CPT interpretation has a fixed signature for its input and output structure.
A solution is to define a single relation $\rel$
containing the tuples of all desired relations $\rel_\colclass$,
but whose arity is one larger than of the $\rel_\colclass$.
This additional entry is used to color the tuples,
where the colors are encoded by a directed path of length $|\colorclasses{\Struct}|$,
which is added to the structure.
To avoid that the arity increases with each further reduction,
we have to reuse this additional entry and refine the colors by increasing the length of the path.
Executing this approach formally would lead to unnecessarily complicated definitions and proofs.
Hence, we just define the relations $\rel_\colclass$.

\subsection{Preprocessing: Transitivity and Typed Relations}

As a first goal, we want to define a normal form
such that the induced substructure induced by $s$ arbitrarily chosen color classes 
is transitive one every color class.
Additionally, the relations should ``respect the color classes'',
that is, tuples
containing vertices of different color classes at the same position
cannot be in the same relation.
We then show that this normal form
is CPT-definable in a canonization preserving way.
The reduction is split into different steps.

\begin{definition}
	We call a relational structure $\Struct$
	\defining{transitive on $s$ color classes}
	if for every $I \subseteq \colorclasses{\Struct}$ satisfying $|I| \leq s$
	the group $\restrictGroup{\autgroup{\Struct[I]}}{\colclass}$
	is transitive for every $\colclass \in I$.
\end{definition}

\begin{lemma}
	\label{lem:compute-and-order-orbits}
	For every constant $q \in \nat$
	there is a CPT term that for 
	every relational structure $\Struct$
	of size $|\Struct| \leq q$ and $k\leq q$
	defines and orders the $k$-orbits of $\autgroup{\Struct}$.
\end{lemma}
\begin{proof}
The $k$-orbits are exactly the
equivalence classes on $k$-tuples
induced by first order logic with $|\Struct|$ many variables.
These classes can be ordered (see e.g.~\cite{Otto1997})
and said logic is captured by CPT because $|\Struct|$ is bounded.
\end{proof}

In many cases it is important that
a group is given as the automorphism group of an explicitly defined structure
and not just as an abstract group.
For example, we can of course always define the trivial group
acting on an arbitrary vertex set,
but we cannot necessarily define a structure with the trivial automorphism group. Indeed, if we could always define such a structure, then we could always order 
the vertices according to
Lemma~\ref{lem:compute-and-order-orbits} and thus define a canonical order on them.
However, if the automorphism group of the original structure was non-trivial on the vertices then defining such an order is impossible.
With the following lemma we can preserve the automorphism groups of
small structures when only considering some of its orbit:
\begin{samepage}
\begin{lemma}
	\label{lem:define-aut-group-restriction}
	For every $q\in\nat$ there is a CPT term 
	that 
	given a relational $\sig$-structure $\Struct = (\StructP, \rel_1^\Struct, \dots, \rel_k^\Struct, \spleq)$
	of size $|\Struct| \leq q$ 
	and a set $V \subseteq \StructP$,
	which is a union of orbits of $\autgroup{\Struct}$,
	defines relations $\rel_{k+1}^\Struct, \dots, \rel_j^\Struct$ over $V$
	of arity at most $|V|$,
	such that for the structure
	$\Struct' = (\StructP, \rel_1^\Struct, \dots, \rel_j^\Struct, \spleq)$
	the following holds:
	$\autgroup{\Struct} = \autgroup{\Struct'}$
	and $\autgroup{\Struct'[V]} = \restrictGroup{\autgroup{\Struct}}{V}$.
\end{lemma}
\end{samepage}
\begin{proof}
	We use Lemma~\ref{lem:compute-and-order-orbits}
	to compute the $|V|$-orbits of $\Struct$
	(note that $|V| \leq |\Struct| \leq q$).
	For every orbit, we define a separate relation
	containing the tuples of the orbit.
	This can be done, because the orbits can be ordered.
\end{proof}
Note that the restriction that $V$ is a union of orbits
is no restriction for CPT terms at all,
because no CPT term can split an orbit of $\autgroup{\Struct}$
and thus every CPT definable set is a union of orbits.

\begin{lemma}
	\label{lem:reduce-to-transitive-on-s-color-classes}
	For every constant $s\in \nat$
	$q$-bounded $\sig$-structures of arity~$r$
	can be reduced canonization preservingly in CPT
	to $q'$-bounded $\sig'$-structures of arity~$r$
	that are transitive on~$s$ color classes
	and have the following additional property:
	Let $\Struct$ be the input structure and
	$\Struct'$ the output structure of the reduction.
	Then for every color class $\colclass' \in \colorclasses{\Struct'}$ 
	there is a color class $\colclass \in \colorclasses{\Struct}$ such that
	the group $\autgroup{\colstruct'}$ is 
	a section (i.e., a quotient group of a subgroup) of $\autgroup{\colstruct}$.
\end{lemma}
\begin{proof}
	Let $\Struct = (\StructP, \rel_1^\Struct, \dots , \rel_k^\Struct, \spleq)$
	be a $q$-bounded $\sig$-structures of arity $r$.
	Let $J := \setcondition{I \subseteq \colorclasses{\Struct}}{|I| \leq s}$
	be the set of all possible choices of $s$ color classes.
	The total order of the color classes
	extends to a total order on $J = \set{I_1, \dots , I_m}$.
	For a set $I \in J$ we denote with $\orbitpart{\Struct[I]}$
	the partition of $\Struct[I]$ into orbits.
	For a vertex $v \in \Struct$
	we write $\orbitGroup{I}{v}$ for the orbit of $v$
	in $\Struct[I]$
	if $v \in \Struct[I]$ and $\emptyset$ otherwise.
	Finally, we define $\orbit{v} = (\orbitGroup{I_1}{v}, \dots ,\orbitGroup{I_m}{v})$.
	Note that there is a lexicographical order $\preceq_\text{lex}$ on $\bigotimes_{i \in [m]} (\orbitpart{\Struct[I_i]} \cup \set{\emptyset})$
	given by the order on $J$ and the orders on the $\orbitpart{\Struct[I_i]}$
	from Lemma~\ref{lem:compute-and-order-orbits}.
	We refine the preorder on $\StructP$ by the preorder
	$u \spleq' v := \orbit{u} \preceq_\text{lex} \orbit{v}$.
	
	Finally, we exploit Lemma~\ref{lem:define-aut-group-restriction}
	to define new relations $\rel_{k+1}^\Struct, \dots, \rel_j^\Struct$
	for all color classes that are split.
	Then the structure $\Struct' := (\StructP, \rel_1^\Struct, \dots , \rel_j^\Struct, \spleq')$ is transitive on $s$ color classes
	by construction.
	We turn a canonical copy of $\Struct'$
	into a canonical copy of $\Struct$
	by replacing the order with the coarser one
	and removing the additional relations.
	This is clearly CPT-definable.

	Let $\colclass'$ be an $\Struct'$-color class.
	By construction there is an $\Struct$-color class $\colclass$
	with $\colclass' \subseteq \colclass$.
	Because of the application of Lemma~\ref{lem:define-aut-group-restriction},
	it holds that 
	$\autgroup{\Struct'[\colclass']} = \restrictGroup{\autgroup{\Struct[\colclass]}}{\colclass'}$. Thus $\autgroup{\Struct'[\colclass']}$
	is a quotient of the setwise stabilizer  $\stab{\autgroup{\Struct'[\colclass']}}{\colclass'} \leq \autgroup{\Struct[\colclass]}$.
\end{proof}

In the previous lemma we added homogeneous relations of possibly larger arity
than present in the original structure
(with Lemma~\ref{lem:define-aut-group-restriction}).
This was necessary to ensure that color classes induce automorphism groups 
that are sections of automorphism groups originally induced by the color classes.
Recall however, that our notion of arity of a structure only takes 
heterogeneous relations into account, and for these the arity does not increase in the reduction.

In contrast to our technique, in the canonization procedure for abelian color classes
\cite[Theorem~6.8]{pakusa2015}
it is not necessary to augment the structure to maintain the automorphism groups.
This is the case because abelian automorphism groups can be ordered,
the order transfers to quotient groups,
and this is sufficient for the rest of the procedure.

The next step is to modify a structure
so that relations respect the color classes. More precisely, we want that within a relation
all tuples
come from the same color classes, i.e.,~%
$\rel_i^\Struct \subseteq \colclass_{i,1} \times \dots \times \colclass_{i, \ell_i}$
for all relations.
Additionally, we want to require that
a tuple in a heterogeneous relation contains at most one vertex of each color class,
so $\colclass_{i,j} \neq \colclass_{i,k}$ for $j \neq k$.
We formalize the notion as follows.
\begin{definition} [Typed Relations]
	Let $\Struct = (\StructP, \rel_1^\Struct, \dots , \rel_k^\Struct, \spleq)$
	be a structure of arity $r$
	and $t=(u_1, \dots, u_\ell) \in \StructP$.
	We call $\type = (\colclass_1, \dots, \colclass_\ell)$
	the \defining{type} of $t$ if $u_i \in \colclass_i$ for all $i \in [\ell]$
	and $\colclass_i \neq \colclass_j$ for all $i \neq j$.
	We denote the set of all occuring types with $\types_\Struct$,
	that is
	\[\types_\Struct = \setcondition{(\colclass_1, \dots, \colclass_\ell)}{
		(\colclass_1, \dots, \colclass_\ell) \text{ is the type of some }(u_1, \dots, u_\ell) \in \rel_j^\Struct}.\]
	We omit the subscript if not ambiguous.
	Let $\type$
	be a type.
	We say that a non-empty relation $\rel_j^\Struct$ \defining{has type $T$},
	if all $t \in \rel_j^\Struct$ have type $T$.
	We say that $\Struct$ \defining{has typed relations}
	if every relation is either homogeneous or has a type.
\end{definition}

Tuples of different types can never be mapped onto each other by
automorphisms
and we can easily separate them.

\begin{lemma}
	\label{lem:reduce-to-typed-relations}
	Relational $q$-bounded $\sig$-structures of arity $r$
	can be reduced canonization preservingly in CPT
	to $q$-bounded $\sig'$-structures of arity $r$ 	with typed relations.
	The reduction preserves transitivity on $s$ color classes
	and the automorphism groups of the color classes
	in the following sense:
	Let $\Struct$ be the input and $\Struct'$
	be the output structure of the reduction.
	Then for every $\colclass' \in \colorclasses{\Struct'}$
	there is a $\colclass \in \colorclasses{\Struct}$
	such that $\Struct[\colclass] \iso \Struct'[\colclass']$.
\end{lemma}
\begin{proof}
	Let $\Struct = (\StructP, \rel_1^\Struct, \dots , \rel_k^\Struct, \spleq)$
	be a $q$-bounded $\sig$-structure of arity $r$.
	We define for every color class $\colclass \in \colorclasses{\Struct}$
	and every $p \in [r]$
	a copy of $\colclass$ 
	\[\colclass_{p} = \setcondition{u_p}{u \in \colclass}\]
	where $u_p$ denotes a new atom.
	The intention is that $\colclass_{p}$ only occurs at position $p$ in 
	tuples of the heterogeneous relations.
	We set \[\StructP' := \bigcup_{\colclass\in \colorclasses{\Struct}, p \in [r]}{\colclass_{p}}\]
	and define a new order $\spleq'$ such that
	$\colclass_{p} \spless' \colclass_{p'}'$
	if and only if $(C, p)$ is lexicographically smaller than $(C',p')$.
	
	We next define relations connecting the copies of the color classes 
	with a perfect matching.
	For every $\colclass \in \colorclasses{\Struct}$ and $p \in [r-1]$ we set
	\begin{align*}
	\rel_{=, \colclass, p}^{\Struct'} :=  \setcondition{(u_p, u_{p+1}) \in \colclass_p \times \colclass_{p+1}}
	{u \in \colclass}.
	\end{align*}
	Clearly, $\rel_{=,\colclass,p}^{\Struct'}$ has type $(\colclass_p, \colclass_{p+1})$.
	To ensure that the color classes $\colclass_{p}$ are isomorphic to 
	$\colclass$,
	we not only have to copy the vertices,
	but the whole substructure $\colstruct$.
	We set $M_{\colclass,j,\ell} := \setcondition{t|_{\colclass} \subseteq \colclass^\ell}{t \in \rel_j^\Struct}$
	to be the restrictions of all tuples in $\rel_j^\Struct$ to $\colclass$
	of length $\ell$.
	
	We define the copies of the relations restricted to the color classes:
	\begin{align*}
	\rel_{\colclass, p, j,\ell}^{\Struct'} := \setcondition{(u^1_p, \dots , u^\ell_p) \in \colclass_p^\ell}{(u^1, \dots, u^\ell) \in M_{\colclass, j, \ell}}
	\end{align*}
	for all $\colclass \in \colorclasses{\Struct}, p \in [r], j \in [k]$ and $\ell \in [r]$.
	The relation $\rel_{\colclass, p, j,\ell}^{\Struct'}$ is homogeneous.
	
	Finally, we split every relation $\rel_j^\Struct$ by defining for every 
	possible combination of color classes $ (\colclass_1, \dots, \colclass_\ell) \in \colorclasses{\Struct}^{\leq r}$
	\begin{align*}
	\rel_{j, (\colclass_1, \dots, \colclass_\ell)}^{\Struct'} := \setcondition{(u^1_1, \dots , u^\ell_\ell)}{
	(u^1, \dots, u^\ell) \in \rel_j^\Struct, u^i \in \colclass_i \text { for all } i \in [\ell]}.
	\end{align*}
	As intended, $\rel_{j,(\colclass_1, \dots, \colclass_\ell)}^{\Struct'}$ has type $(\colclass_{1,1}, \dots, \colclass_{\ell, \ell})$ if it is non-empty.
	Then the structure~$\Struct'$
	obtained by all the defined non-empty relations
	(which can easily be ordered)
	has the desired property.
	
	The reduction is canonization preserving
	because by contracting the perfect matching relations~$\rel_{=, \colclass, p}^{\Struct'}$
	and undoing the splitting by type
	we recover $\Struct$.
	This operation can be defined in CPT.
	Assume that $\Struct$ is transitive on $s$-color classes
	and let $I' \subseteq \colorclasses{\Struct'}$ be a set of $|I'| \leq s$
	color classes.
	We consider the substructure $\Struct'[I']$
	and contract all perfect matching relations.
	For the obtained structure $\Struct''$
	there is an $I \subseteq \colorclasses{\Struct}$
	with $|I| \leq s$ such that $\Struct'' \subseteq \Struct[I]$.
	$\Struct''$ is not necessarily induced, because some tuples in relations may be missing.
	But all missing relations are of different type,
	hence an automorphism of $\Struct[I]$ is always an automorphism of~$\Struct''$. 
	So $\Struct''$ is transitive on $s$-color classes
	because $\Struct[I]$ is so
	and thus also $\Struct'[I']$ and $\Struct'$.
\end{proof}

\subsection{$2$-Injective Subdirect Products and Quotients}

In the end, want to achieve that the automorphism group
of three color classes is always a $2$-injective subdirect product.
Before we can proceed to that,
we need to modify the color classes to allow further operations on them.

A color class is called \defining{regular},
if its automorphism group is regular.
We now show that we can replace in
a structure
every color class $\colclass$ by a regular color class $\colclass'$
satisfying $\autgroup{\colstruct} \iso \autgroup{\colstruct'}$.
Note that a permutation group $\group \leq \SymSetGroup{\Omega}$ is regular
if $\group$ is transitive and $|\Omega| = |\group|$.

\begin{lemma}
	\label{lem:regular-and-faithful-k-orbit}
	Let $\group \leq \SymSetGroup{\Omega}$ be a group
	with domain $\Omega$.
	Then there is an $\ell \leq |\group|$
	and an $\ell$-orbit of $\group$
	on which $\group$ acts regularly and faithfully.
\end{lemma}
\begin{proof}
Let~$(g_1,\ldots,g_\ell)$ be a  base of minimum order of $\group$,
that is the only element of~$\group$ fixing all~$g_i$ is~$1$ and~$\ell$ is minimal.
Then~$\ell \leq \log |\group|$.
The group acts regularly and faithfully on the orbit of~$(g_1,\ldots,g_\ell)$.
\end{proof}

Note that a construction taking $\group$ as new domain
cannot succeed in CPT,
because the elements in $\group$ are distinguishable,
e.g.~$1\in \group$ is different from all others.

\begin{lemma}
	\label{lem:reduce-to-regular-color-classes}
	Relational $q$-bounded $\sig$-structures of arity $r$
	with typed relations
	can be reduced canonization preservingly in CPT
	to $q'$-bounded $\sig'$-structures of arity $r$
	with typed relations and regular color classes.
	The reduction preserves transitivity on $s$ color classes
	and the automorphism groups
	(in the sense of Lemma~\ref{lem:reduce-to-typed-relations}).	
\end{lemma}
\begin{proof}
	
	Let $\Struct = (\StructP, \rel_1^\Struct, \dots, \rel_k^\Struct, \spleq)$ be a $q$-bounded structure of arity $r$ with typed relations
	and $\colclass \in \colorclasses{\Struct}$.
	We show that there is a CPT term which on input $\Struct$ and $\colclass$  replaces~$\colclass$ by a regular color class~$\colclass'$
	such that the structure 
	$\Struct' = ((\StructP\setminus \colclass) \cup \colclass', \rel_1^{\Struct'}, \dots, \rel_{k'}^{\Struct'}, \spleq')$ 
	has typed relations,
	$\Struct[\StructP \setminus\colclass] = \Struct'[\StructP'\setminus \colclass']$,
	$\colclass'$ is regular, 
	$\autgroup{\Struct[\colclass]} \iso \autgroup{\Struct'[\colclass']}$,
	and transitivity on $s$ color classes is preserved.
	By applying this reduction subsequently to all color classes of $\Struct$
	we obtain the claimed reduction.

	We set $\group := \autgroup{\Struct[\colclass]}$.
	Let $\ell \leq |\group|$
	be the smallest number such that
	there is an $\ell$\nobreakdash-orbit $O$ of $\group$
	on which $\group$ acts regularly and faithfully
	(we pick the minimal orbit $O$ with Lemma~\ref{lem:compute-and-order-orbits}).
	Such an $\ell$ exists by Lemma~\ref{lem:regular-and-faithful-k-orbit}.
	We now replace $\colclass$ with $O$.
	The new color class $\colstruct'$ is given by
	\[\colclass' := \setcondition{u_t}{t \in O},\]
	where the $u_t$ are fresh atoms.
	As a first step, we need to ensure that the automorphism group of $\colclass'$ is isomorphic to the one of $\colclass$.
	This is done as in Lemma~\ref{lem:define-aut-group-restriction}.
	We define relations $\rel^{\Struct'}_{k+1}, \dots, \rel^{\Struct'}_{k'-1}$
	for the $|O|$-orbits of the permutation group on $\colclass'$
	induced by $\group$.
	Next, we add another relation $\rel^{\Struct'}_{k'}$
	that identifies vertices of $\colclass'$ with a single vertex of $\colclass$:
	\[\rel^{\Struct'}_{k'} := \setcondition{(u_t, u_{t'}) \in \colclass'^2}{t(1) = t'(1)}.\]
	We group vertices of $\colclass'$
	by the first position of the tuples $t$,
	with the intention,
	that we use all vertices $u_t$ with first position $t(1) = u \in \colclass$
	as replacement for $u$ in $\Struct'$.
	The relation $\rel^{\Struct'}_{k'}$ is a disjoint union of cliques
	and the cliques correspond one-to-one to the elements of $\colclass$.
	
	Last, we need to adapt the existing relations.
	Let $\rel_j^\Struct$ be a relation of type $(\colclass_1, \dots, \colclass_\ell)$ such that $\colclass_i = \colclass$.
	We replace it with the following relation:
	\[\rel_j^{\Struct'} := \setcondition{(v_1,\dots, v_{i -1}, u_t,v_{i+1},\dots, v_\ell)}{t(1) = v_i, (v_1,\dots,v_\ell) \in \rel_j^\Struct}.\]
	That is, we relate vertices in other color classes
	to all new vertices with first position $v_i$.
	The relation $\rel_j^{\Struct'}$ has the same type as $\rel_j^{\Struct}$
	except that $\colclass$ is replaced with $\colclass'$.

	Assume that $\rel_j^\Struct \subseteq \colclass^m$ is a homogeneous relation.
	We replace the vertices in the relation in the same manner as before:
	\[\rel_j^{\Struct'} := \setcondition{(u_{t_1} ,\dots, u_{t_m}) \in \colclass'^m}{(u_1,\dots, u_m) \in \rel_j^{\Struct}, t_i(1) = u_i, \text{ for all } i \in [m]}.\]
	All remaining relations (those whose type does not involve $\colclass$) are just copied.
	We obtain the structure $\Struct'$
	by the definitions above 
	and the preorder $\spleq'$ defined by
	replacing~$\colclass$ with~$\colclass'$ in $\spleq$.
	By construction we have 
	$\autgroup{\Struct'[\colclass']} \iso \group = \autgroup{\Struct[\colclass]}$
	and that~$\colclass'$ is regular.
 	One easily checks that the reduction is CPT-definable
 	and preserves transitivity on $s$ color classes.
 	It is canonization preserving,
 	because we can contract the cliques in $\rel_{k'}^{\Struct'}$
 	in a canonical copy of $\Struct'$
 	to obtain a canonical copy of $\Struct$.
\end{proof}
Next, we exploit regularity of the color classes to construct quotient groups.
They will be important to reduce to $2$-injective subdirect products.

\begin{definition}[Quotient Color Class]
	\label{def:quotient-construction}
	Let  $\Struct = (\StructP, \rel_1^\Struct, \dots, \rel_k^\Struct, \spleq)$
	be a structure
	and the automorphism group of $\colclass \in \colorclasses{\Struct}$ be regular.
	Let $N \normal \autgroup{\colstruct}$.
	We say that another color class~$\colclass'$
	is an \defining{$N$-quotient of $\colclass$}
	if $\autgroup{\colstruct'} \iso \autgroup{\colstruct}/N$
	and there is a function 
	$\rel_j^\Struct \subseteq \colclass \times \colclass'$
	determining the orbit partition of $N$ acting on $\colclass$,
	i.e., a vertex in $\colclass'$ corresponds to an $N$ orbit
	and the vertices of $\colclass$ are adjacent to its orbit vertices
	via $\rel_j^\Struct$.
	The relation $\rel_j^\Struct$ is called the 
	\defining{orbit-map} (of $\colclass$).
\end{definition}

The following lemma states that quotient groups
of regular permutation groups can be defined using the
orbits of the normal subgroup.

\begin{lemma}
	\label{lem:regular-perm-group-quotient}
	Let $\group \leq \SymSetGroup{\Omega}$ be a regular permutation group and
	$N \normal \group$.
	Then $\group$
	acting on the set $\orbitpart{N}$
	forms a regular permutation group isomorphic to $\group / N$.
\end{lemma}
\begin{proof}
	Since~$N$ is normal in~$\Gamma$, every permutation $\perm \in \group$
		induces a permutation $\perm_N$ of the $N$-orbits.	
	Thus $\group$ induces a permutation group $\group'$ on
	$\orbitpart{N}$.
	One easily checks that the map $\psi \colon \group / N \to \group'$
	defined by $\perm N \mapsto \perm_N$
	is an isomorphism.
	The group $\group'$ is transitive because $\group$ is transitive.
	Let there be an $N$-orbit of size $k$.
	Then in particular $|N| = k$,
	because otherwise $\group$ was not regular.
	But so all $N$-orbits have size~$k$
	and $|\orbitpart{N}| = |\Omega|/k = |\group| / k = |\group / N| = |\group'|$. Hence, $\group'$ is regular.
\end{proof}

We exploit the previous lemma to construct quotient color classes.

\begin{lemma}
	\label{lem:quotient-color-classes}
	There is a canonization-preserving CPT-reduction
	that given a
	$q$\nobreakdash-bounded $\sig$\nobreakdash-structure
	$\Struct = (\StructP, \rel_1^\Struct, \dots, \rel_k^\Struct)$ of arity $r$,
	a regular color class $\colclass \in \colorclasses{\Struct}$,
	and $N \normal \autgroup{\colstruct}$,
	outputs a $q$-bounded $\sig'$-structure
	$\Struct' = (\StructP \disunion  \colclass', \rel_1^\Struct, \dots, \rel_k^\Struct, \rel_\orbsym, \rel'_1,
	\dots, \rel'_\ell, \spleq')$ for some~$\ell$
	such that $\Struct'[\StructP] = \Struct$,%
	~$\colclass'$ is
	an $N$-quotient of~$\colclass$,~$\colclass'$ is only connected to~$\colclass$ and only
	via~$\rel_\orbsym$,
	and all other new relations $\rel'_i$ are homogeneous (and thus in $\colclass'$)%
	\footnote{So formally this reduction is from triples consisting of a structure, a color class, and a normal subgroup to structures, each satisfying the condition stated in the lemma.}.
\end{lemma}
\begin{proof}
	Let $\Struct$, $\colclass$, and $N$ be as stated in the lemma.
	We set $\group := \autgroup{\colstruct}$.
	We define new atoms $\colclass' := \setcondition{u_O}{O \in \orbitpart{N}}$
	and realize the orbit-map by the relation
	\[\rel_{\orbsym} := \setcondition{(u,u_O) \in \colclass \times \colclass'}{u \in O}.\]
	
	Let $\Struct_\colclass := (\colclass\cup \colclass', \rel_{\orbsym}, \spleq_\colclass) \cup \colstruct$ be the structure 
	consisting of $\colclass$ and the attached new vertices
	and $\spleq_\colclass$ be defined such that $\colclass \spless_\colclass \colclass'$.
	Let $\groupB := \autgroup{\Struct_\colclass} \leq \group \times \SymSetGroup{\colclass'}$.
	Let $(\phi, \psi) \in \groupB$.
	By Lemma~\ref{lem:regular-perm-group-quotient}
	the automorphism $\phi$ permutes the $N$-orbits.
	Because every orbit vertex $u_O \in \colclass'$
	is adjacent to all vertices contained in $O$,
	$\psi$ has to be the permutation of $N$-orbits
	corresponding to $\phi$.
	So $\group' := \restrictGroup{\groupB}{\colclass'}$
	is the permutation group of $N$-orbits given by~$\group$
	and from Lemma~\ref{lem:regular-perm-group-quotient}
	follows that $\group' \iso  \group / N$.
	
	We apply Lemma~\ref{lem:define-aut-group-restriction}
	and define homogeneous relations $\rel'_1, \dots, \rel'_\ell$ on $\colclass'$
	such that $\autgroup{(\colclass', \rel'_1, \dots, \rel'_\ell)} \iso \group'$.
	Then $\colclass'$ is an $N$-quotient of $\colclass$.
	Finally, $\Struct'$ is obtained
	by defining $\spleq'$ to be $\spleq$
	with $\colclass'$
	as immediate successor of $\colclass$.
	The reduction is canonization preserving,
	because we just need to remove $\colclass'$.
\end{proof}

Now we are prepared to turn to structures with $2$-injective subdirect products
as local automorphism groups:

\begin{definition}
	Let $\Struct = (\groupvertices \disunion \extensionvertices, \rel_1^\Struct, \dots , \rel_k^\Struct, \spleq)$
	be a structure,
	where $\groupvertices$ and $\extensionvertices$ are unions of color classes.
	We call a color class $\colclass \subseteq \groupvertices$
	(respectively $\colclass \subseteq \extensionvertices$)
	a \defining{group color class} (respectively an \defining{extension color class}).
	We define the \defining{group types} $\grouptypes^\Struct \subseteq \types^\Struct$ to be the set of all types only consisting of group color classes.
	For $\type = (\colclass_1, \dots, \colclass_j) \in \grouptypes^\Struct$.
	we set $\typegroupStruct{\Struct}{\type} := \autgroup{\Struct[\bigcup_{i \in [j]}\colclass_i]}\leq \bigotimes_{i \in [j]} \autgroup{\colstruct_i}$.
	
	Finally, we call~$\Struct$
	an \defining{$(r-1)$-injective quotient structure}
	if it satisfies the following:
	\begin{itemize}
		\item $\Struct$ is of arity $r$, has typed relations, and all color classes are regular.
		\item Every group color class $\colclass \subseteq \groupvertices$
		is an $N$-quotient of exactly one extension color class
		$\colclass' \subseteq \extensionvertices$,
		where $N \normal \autgroup{\colstruct'}$,
		and not related to any other extension color class apart from $\colclass'$.
		Moreover, $\colclass$ is only related by the orbit-map to $\colclass'$.
		\item 
		All relations only between group color classes are of arity exactly $r$
		and every group color class occurs in only one group type.
		\item
		For every $\type \in \grouptypes^\Struct$
		the group $\typegroupStruct{\Struct}{\type}$
		is an $(r-1)$-injective subdirect product.
	\end{itemize}
	
\end{definition}
We leave out the superscripts if the structure $\Struct$ is clear from the context.

\tikzstyle{colorclass} = [circle, draw=black, minimum width=8mm]
\tikzstyle{quotient} = [circle, draw=black, inner sep=0.2mm, minimum width=5mm]
\tikzstyle{qvertex} = [circle, fill=black, inner sep=0.3mm]
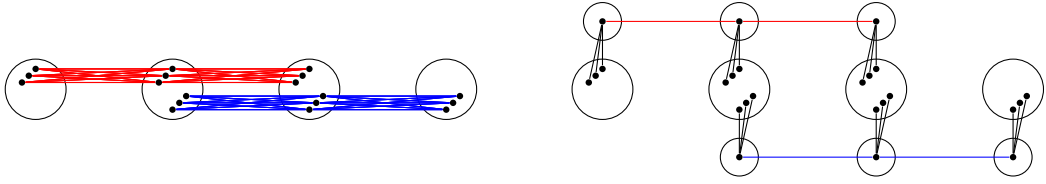
\begin{figure}
	\centering
	
		\begin{tikzpicture}[>=stealth, scale=0.9]

	\def\quotientY {1}
	\def\groupY{0}
	
	\draw[opacity = 0, use as bounding box] (-0.5,{\quotientY +0.5}) rectangle (6.5,{-\quotientY -0.5});

	\node(C1)[colorclass] at (0,\groupY) {};
	\node(C2)[colorclass] at (2,\groupY) {};
	\node(C3)[colorclass] at (4,\groupY) {};
	\node(C4)[colorclass] at (6,\groupY) {};
	
	\foreach \i in {1,...,3} {
		\node (R1\i)[qvertex] at ($(C1.center) + \i*(0.1,0.1) - (0.3,0)$) {};
		\node (R2\i)[qvertex] at ($(C2.center) + \i*(0.1,0.1) - (0.3,0)$) {};
		\node (R3\i)[qvertex] at ($(C3.center) + \i*(0.1,0.1) - (0.3,0)$) {};	
	}
	
	\foreach \x in {1,...,3} {
		\foreach \y in {1,...,3} {
			\foreach \z in {1,...,3} {
				\path[red,-]
				(R1\x) edge 
				(R2\y)  
				(R2\y) edge
				(R3\z)  ;
			}
		}
	}
	
	\foreach \i in {1,...,3} {
		\node (B2\i)[qvertex] at ($(C2.center) - \i*(0.1,0.1)+ (0.3,0)$) {};
		\node (B3\i)[qvertex] at ($(C3.center) - \i*(0.1,0.1)+ (0.3,0)$) {};
		\node (B4\i)[qvertex] at ($(C4.center) - \i*(0.1,0.1)+ (0.3,0)$) {};	
	}
	
	\foreach \x in {1,...,3} {
		\foreach \y in {1,...,3} {
			\foreach \z in {1,...,3} {
				\path[blue,-]
				(B2\x) edge 
				(B3\y)  
				(B3\y) edge
				(B4\z)  ;
			}
		}
	}

	\end{tikzpicture} %
	\hspace{1cm}%
	\begin{tikzpicture}[>=stealth, scale=0.9]
	
	\def\quotientY {1}
	\def\groupY{0}
	
	\draw[opacity = 0, use as bounding box] (-0.5,{\quotientY +0.5}) rectangle (6.5,{-\quotientY -0.5});
	
	\node(C1)[colorclass] at (0,\groupY) {};
	\node(C2)[colorclass] at (2,\groupY) {};
	\node(C3)[colorclass] at (4,\groupY) {};
	\node(C4)[colorclass] at (6,\groupY) {};

	\node(Q11)[quotient] at (0,\quotientY) {};
	\node(Q12)[quotient] at (2,\quotientY) {};
	\node(Q13)[quotient] at (4,\quotientY) {};
	\node(Q22)[quotient] at (2,-\quotientY) {};
	\node(Q23)[quotient] at (4,-\quotientY) {};
	\node(Q24)[quotient] at (6,-\quotientY) {};
	
	\node(Q11c)[qvertex] at (Q11.center) {};
	\node(Q12c)[qvertex] at (Q12.center) {};
	\node(Q13c)[qvertex] at (Q13.center) {};
	\node(Q22c)[qvertex] at (Q22.center) {};
	\node(Q23c)[qvertex] at (Q23.center) {};
	\node(Q24c)[qvertex] at (Q24.center) {};

	\foreach \i in {1,...,3} {
		\node (QR1\i)[qvertex] at ($(C1.center) + \i*(0.1,0.1) - (0.3,0)$) {};
		\node (QR2\i)[qvertex] at ($(C2.center) + \i*(0.1,0.1) - (0.3,0)$) {};
		\node (QR3\i)[qvertex] at ($(C3.center) + \i*(0.1,0.1) - (0.3,0)$) {};
		\node (QB2\i)[qvertex] at ($(C2.center) - \i*(0.1,0.1) + (0.3,0)$) {};
		\node (QB3\i)[qvertex] at ($(C3.center) - \i*(0.1,0.1)+ (0.3,0)$) {};
		\node (QB4\i)[qvertex] at ($(C4.center) - \i*(0.1,0.1)+ (0.3,0)$) {};		
	}

	\foreach \x in {1,...,3} {
		\path[black,-]
		(Q11c) edge (QR1\x)
		(Q12c) edge (QR2\x)
		(Q13c) edge (QR3\x)  ;
		\path[black,-]
		(Q22c) edge (QB2\x)
		(Q23c) edge (QB3\x)
		(Q24c) edge (QB4\x)  ;
	}	
	
	\path[red]
	(Q11c) edge (Q12c)
	(Q12c) edge (Q13c);
	\draw[blue]
	(Q22c) edge (Q23c)
	(Q23c) edge (Q24c);
	
	\end{tikzpicture} 
	\caption{The situation of Lemma~\ref{lem:reduce-2-inj}:
	On the left the input structure.
	Each circle represents one color class
	with drawn tuples of two relations (red and blue)
	between two orbits of each color class
	(there can of course be more orbits and tuples of the red and blue relations).
	On the right the altered structure:
	For the types of the red and blue relations 
	there are new group color classes (on the top for red and on the bottom for blue),
	where the orbits are contracted to a single vertex. 
	The ``old'' color classes became extension color classes.
	}
	\label{fig:reduce-2-inj}
\end{figure}

\begin{lemma}
	\label{lem:reduce-2-inj}
	Relational $q$-bounded $\sig$-structures of arity at most $3$,
	which are transitive on $3$ color classes
	and have typed relations,
	can be canonization preservingly reduced in CPT
	to $2$-injective quotient structures. 
	Moreover, the reduction satisfies the following:
	If the input structure~$\Struct$ was of arity $2$
	and $\colclass_1, \colclass_2, \colclass_3$ are related group color classes of the output structure $\Struct'$,
	then (up to permutation of the $\colclass_i$) $\outproj_{\colclass_1}(\autgroup{\Struct'[\bigcup_{i \in [3]}\colclass_i]} \subseteq \autgroup{\colstruct_2} \direct \autgroup{\colstruct_3}$ is a diagonal subgroup
	via a canonical isomorphism $\autgroup{\colstruct_2} \to \autgroup{\colstruct_3}$.
\end{lemma}
\begin{proof}
	Let $\Struct = (\Struct, \rel_1^{\Struct},\dots, \rel_k^{\Struct}, \spleq)$
	be a $q$-bounded $\sig$-structure as required by the lemma.
	We call a type $\type \in \types$
	order-compatible, if $\type = (\colclass_1, \dots, \colclass_\ell)$
	and $\colclass_1 \spleq \cdots \spleq \colclass_\ell$.
	We first assume that the type of every relation is order-compatible.
	If this is not the case, we just reorder the positions of the tuples,
	which is of course canonization preserving.
	This assumption is only required to simplify formulas.
	We second assume that all heterogeneous relations are of arity~$3$
	and discuss the arity~$2$ case in the end.

	Let
	$\type = (\colclass_1, \colclass_2, \colclass_3) \in \types$.
	We set $\typegroup{\type}$ to be the automorphism group
	of the substructure given by all relations of type $\type$.
	The group $\typegroup{\type}$ is subdirect,
	because $\Struct$ is transitive on $3$ color classes.
	To make $\typegroup{\type}$ a $2$-injective subdirect product,
	we want to factor out the kernels
	$N^\type_i := \kernel{\outproj_i^{\typegroup{\type}}}$
	for all $i \in [3]$
	and use quotient color classes for that.
	Note that naturally $N^\type_i \normal \autgroup{\colstruct_i}$.
	Now, a fixed color class $\colclass \in \colorclasses{\Struct}$
	can be contained in different types
	yielding different kernels and thus different quotient color classes.
	Hence, we construct for every type containing $\colclass$
	a new quotient color class of $\colclass$.
	
	We use Lemma~\ref{lem:quotient-color-classes}
	to define $N^\type_i$-quotient color classes $Q^\type_i$
	and orbits maps $\rel^{\Struct'}_{\orbsym,\type,i}$
	for every $\type \in \types$
	and $i \in [3]$.
	We set
	\[\groupvertices' := \bigcup_{
		\substack{\type \in \types, \\i \in [3]}}
		Q^T_i\]
	and
	\[\extensionvertices' := \StructP.\]
	Next, we have to ``move'' the relations from $\StructP$ to $\groupvertices$. We set
	\[\rel^{\Struct'}_{\type,j} := \setcondition{(v_{O_1}, v_{O_2} , v_{O_3})}{(u_1, u_2, u_\ell) \in \rel_j^\Struct, O_i = \orbitGroup{{N^\type_i}}{u_i}\text { for all } i \in [3]},\]
	that is, we just relate the orbits instead of the vertices
	(see Figure~\ref{fig:reduce-2-inj}).

	We obtain $\Struct'$ by using
	all homogeneous relations of $\Struct$,
	the additional homogeneous relations for the quotient color classes,
	as well as the orbit-maps given by Lemma~\ref{lem:quotient-color-classes},
	and the~$\rel^{\Struct'}_{\type,j}$
	for all heterogeneous relations $\rel_j^\Struct$.
	The order $\spleq'$ is given by $\spleq$,
	the quotient color classes ordered behind the original ones,
	and
	$Q^\type_i \spleq' Q^{\type'}_{i'}$
	if and only if $(\type,i)$ is lexicographically smaller than $(\type',i')$.
	The reduction is clearly CPT-definable.
	
	Let $Q$ be a set of three related group color classes.
	By construction, there is a type $\type=(\colclass_1, \colclass_2, \colclass_3)$
	such that $Q = \set{Q^\type_1, Q^\type_2, Q^\type_3}$
	and $Q^\type_i$ is a quotient color class of $\colclass_i$.
	Let $\group_Q := \autgroup{\Struct'[Q]}$.
	We argue that $\group_Q \iso \typegroup{\group} / N^\type_1 / N^\type_2 / N^\type_3$
	and hence $\group_Q$ is a $2$-injective subdirect product.
	
	First note that $\typegroup{\type}/N^\type_1/N^\type_2/N^\type_3 = \typegroup{\type} / (N^\type_1 N^\type_2 N^\type_3)$
	and that $N := N^\type_1 N^\type_2 N^\type_3 \normal \typegroup{\type}$.
	Hence, $\typegroup{\type}/N$ defines a permutation group on the $N$-orbits of $\typegroup{\type}$,
	which are precisely the~$N^\type_i$ orbits on $\colclass_i$ for all $i \in [3]$.
	We now study these orbits:
	Let $\rel_j^\Struct$ be a relation of type $\type$,
	$ (u_1, u_2, u_3) \in \rel_j^\Struct$,
	and $\perm_i \in N^\type_i$ for every $i \in [3]$.
	Then $\perm_1\perm_2\perm_3((u_1, u_2, u_3)) \in \rel_j^\Struct$
	(because $N^\type_i \leq \group_T$).
	Note that for $\ell \neq i$
	we have $\perm_i(u_{\ell}) = u_{\ell}$.
	It follows that ${(w_1,w_2,w_3) \in \rel_j^\Struct}$
	whenever $w_i \in \orbitGroup{N^\type_i}{u_i}$ for all $i \in [3]$,
	that is
	$\orbitGroup{N^\type_1}{u_1} \times \orbitGroup{N^\type_2}{u_2} \times \orbitGroup{N^\type_3}{u_3} \subseteq \rel_j^\Struct$.
	
	We define a map $\phi \colon \typegroup{\type}/N \to \group_Q$ as follows:
	Let $\perm N \in \typegroup{\type}/N$.
	This defines a permutation on the $N$-orbits $\perm_N$,
	which translates to a permutation of the orbit vertices of the quotient color classes $Q^T_i$.
	Then $\perm_N \in \group_Q$
	because whenever $(v_{O_1}, v_{O_2}, v_{O_3}) \in \rel^{\Struct'}_{\type,j}$,
	there is a tuple $(u_1,u_2,u_3) \in \rel^\Struct_j$ such that $O_i = \orbitGroup{N^\type_i}{u_i}$ for all $i \in [3]$
	and clearly $\perm((u_1, u_2, u_3)) \in \rel^\Struct_j$
	and $\perm_N(O_i) = \orbitGroup{N^\type_i}{u_i}$ for all $i \in [3]$.
	By construction of the $\rel^{\Struct'}_{\type,j}$
	we also have $(v_{\perm_N(O_1)}, v_{\perm_N(O_2)}, v_{\perm_N(O_3)})
	\in \rel^{\Struct'}_{\type,j}$
	and so $\perm_N$ is an isomorphism.
	In particular, 
	the prior reasoning also holds in the other direction
	and thus $\phi$ is an isomorphism
	and thus $\group_Q$ is a $2$-injective subdirect product.	
	
	Lastly, we show that the reduction is canonization preserving.
	Given a canonization of~$\Struct'$,
	we delete the quotient color classes
	and recover the heterogeneous relations as follows:
	By the argument above,
	we have that
	$(u_1,u_2,u_3) \in \rel_j^{\Struct}$
	if and only if
	$(v_{O_1},v_{O_2},v_{O_3}) \in \rel^{\Struct'} _{\type,j}$
	where~$T$ is the type of $\rel_j^\Struct$
	and $O_i = \orbitGroup{N^\type_i}{u_i}$.
	By the orbit maps, we have that
	$O_i = \orbitGroup{N^\type_i}{u_i}$
	if and only if
	$(u_i, v_{O_i}) \in \rel^{\Struct'}_{\orbsym, \type,i}$.
	Both conditions can easily be checked in CPT.
	
	It remains to discuss relations of arity $2$.
	Let $\rel^\Struct_j$ be of type $(\colclass_1, \colclass_2)$.
	We extend $\rel^\Struct_j$ to arity $3$
	by creating a copy $\colclass_2'$ of $\colclass_2$,
	set $\rel_{j,1}^\Struct := \setcondition{(u_1,u_2,u_2')}{(u_1,u_2) \in \rel^\Struct_j}$
	where~$u_2'$ is the copy of $u_2$ in $\colclass_2'$,
	and introduce another relation that connects $\colclass_2$ and $\colclass_2'$ by a perfect matching
	$\rel^\Struct_{j,2} := \setcondition{(u_1,u_2,u_2') }{u_1\in\colclass_1, u_2 \in \colclass_2}$.
	The first component is just to have arity $3$.
	By this matching relation, we have that
	$\group := \autgroup{\Struct[\colclass_1\cup \colclass_2 \cup \colclass_2']}
	= \setcondition{(\perm_1,\perm_2,\perm_2')}{(\perm1,\perm_2) \in \autgroup{\Struct[\colclass_1\cup \colclass_2}}$
	where $\perm_2'$ is the action of $\perm_2$ on the copy $\colclass_2'$,
	that is $\outproj^\group_{\colclass_1} \subseteq \autgroup{\colstruct_2} \direct \autgroup{\colstruct_2'}$
	is a diagonal subgroup via the canonical isomorphism
	$\autgroup{\colstruct_2} \to \autgroup{\colstruct_2'}$
	defined by $\perm_2 \mapsto \perm_2'$. 
\end{proof}

The construction extends to structures of arity $r$
and then yields $(r-1)$ injective quotient structures of arity $r$
for a straightforward generalization of $2$-injective subdirect products.
Figure~\ref{fig:2-inj-dih-structure} in Section~\ref{sec:canonization-dihedral}
shows (apart additional properties from this section)
a sketch of a $2$-injective quotient structure.
We conclude the section on normal forms:
\begin{theorem}
	\label{thm:convert-to-normal-form}
	Relational $q$-bounded $\sig$-structures of arity $r$
	can canonization preservingly in CPT be reduced to
	\begin{enumerate}[label=\alph*)]
		\item $q'$-bounded $2$-injective
		quotient $\sig'$-structures and to
		\item $q'$-bounded $(r-1)$-injective
		quotient $\sig'$-structures
		such that 
		for every group color class $\colclass' \in \colorclasses{\Struct'}$
		of the output structure $\Struct'$
		the group $\autgroup{\Struct'[\colclass']}$
		is a section of 
		$\autgroup{\Struct[\colclass]}$ for some
		color class
		$\colclass \in \colorclasses{\Struct}$
		of the input structure $\Struct$.
	\end{enumerate}
\end{theorem}
\begin{proof}
	Case b) follows by Lemmas~\ref{lem:reduce-to-transitive-on-s-color-classes},~%
	\ref{lem:reduce-to-typed-relations},~
	\ref{lem:reduce-to-regular-color-classes}, and~
	\ref{lem:reduce-2-inj}.
	For a), we first need to reduce the arity of the structure to~$3$,
	and then apply b).
	We sketch how one can restrict the arity:
	For simplicity, we first apply Lemma~\ref{lem:reduce-to-typed-relations}
	to obtain typed relations.
	Let $\rel_j^\Struct$ be a relation of arity $r>3$ of type $(\colclass_1,\dots, \colclass_r)$.
	We add for every pair $(u,v) \in \restrictVect{\rel_j^\Struct}{\colclass_{r-1}\cup \colclass_{r}}$
	a new vertex and split $\rel_j^\Struct$
	into two relations $\rel_{j,1}^\Struct$ and $\rel_{j,2}^\Struct$ of arity $r-1$ and $3$ as follows.
	We define them such that 
	$(u_1, \dots, u_{r-2},u_t) \in \rel_{j,1}^\Struct$ and
	$(u_t, u_{r-1}, u_r) \in \rel_{j,2}^\Struct$
	if and only if $(u_1,\dots,u_r) \in \rel_j^\Struct$,
	where $u_t$ is the new vertex for the tuple $(u_{r-1}, u_r)$.
	The new vertices form a new color classes.
	By iteratively applying this reduction,
	one can reduce the arity down to $3$.
	For $r=3$ the relation $\rel_{j,1}^\Struct$ has arity~$2$,
	but $\rel_{j,2}^\Struct$ as still arity $3$. 
\end{proof}
Note that Theorem~\ref{thm:convert-to-normal-form-light}
is just case~a) of the previous theorem.
Of course, one could reduce the arity of the structures to~$2$
by encoding every tuple in a relation by a new vertex adjacent to
the vertices in the tuple.
Obviously, the automorphism group of these new vertices
must be the automorphism group of the tuples in the relations,
so we have indeed encoded the structure in a graph,
but are still faced with exactly the same groups.
The new color classes obtained by our construction can be much simpler,
depending on the actual group.
In the following, we will consider structures of arity at most $3$.

\section{Structures with Dihedral Colors}
\label{sec:dihedral-colors}

We now consider structures whose color classes
have dihedral automorphism groups.
\begin{definition}[Dihedral Colors]
	A structure $\Struct = (\StructP, \rel_1^\Struct, \dots, \rel_k^\Struct, \spleq)$ \defining{has dihedral colors}
	if for every $\Struct$-color class $\colclass$
	the group $\autgroup{\colstruct}$
	is a dihedral or cyclic group.
\end{definition}
Here, we include cyclic automorphisms groups in the definition
of dihedral colors. We do so because we want that the class of groups admissible for the color classes 
is closed under taking subgroups and under taking quotient groups.
We now show that the automorphism group of a regular and dihedral color class
can be made explicit in the following sense:

\begin{definition} [Color Class in Standard Form]
	Let  $\Struct$ be a structure and $\colclass \in \colorclasses{\Struct}$.
	We say that 
	\defining{a color class $\colclass \in \colorclasses{\Struct}$
	is in standard form} if the following holds: 
	\begin{itemize}
		\item If $\autgroup{\colstruct} \iso \CyclicGroup{|C|}$
			then there are relations
			$\rel_i^\Struct,\rel_j^\Struct \subseteq \colclass^2$ of arity $2$ 
			each forming a directed cycle of length $|C|$ on $\colclass$.
		\item Otherwise $\autgroup{\colstruct} \iso \DihedralGroup{|C|/2}$
			and there are two relations
			$\rel_i^\Struct, \rel_j^\Struct \subseteq \colstruct^2$
			such that $\rel_j^\Struct$ defines two directed and disjoint cycles
			of length $|C|/2$
			and $\rel_i^\Struct$ connects them by a perfect matching
			such that the two cycles are directed into opposite directions
			(cf.~Figure~\ref{fig:regular-drawing}).\footnote{The two relations form the Cayley graph generated by a reflection and a rotation of maximum order, but we will not need this fact.}
	\end{itemize}
	We say that the relations $\rel_i^\Struct$ and $\rel_j^\Struct$
	\defining{induce the standard form of $\colstruct$}.
	The \defining{color classes of $\Struct$ are in standard form},
	if every color class is in standard form.
\end{definition} 
\begin{figure}
	\centering
	\begin{tikzpicture}[>=stealth,every node/.style={vertex}, thick,
	]
		\def \n {5}
		\def \nminusone {4} 
				
		\def \radius {2cm}
		\def \innerradius {1cm}
		\def \degoffset {90}
		\foreach \s in {0,...,\nminusone}
		{
			\node(a\s) at ({cos(360*\s/\n + \degoffset)*\radius},{sin(360/\n*\s+\degoffset)*\radius}) {};
			
			\node(b\s) at ({cos(360/\n*\s+\degoffset)*\innerradius},{sin(360/\n*\s+\degoffset)*\innerradius}) {};
		}
		\foreach \s [evaluate=\s as \t using {int(mod(\s+1,\n))}] in {0,...,\nminusone} 
		{
			\path[->, black, draw] (a\s) edge(a\t) ;
			\path[->, black, draw] (b\t) edge (b\s);
			\path[black, draw, dashed] (a\s) edge (b\s);
		}
	
		\begin{scope}[xshift = 4cm, yshift = -\innerradius ]
			\node (c1) at (0,0) {};
			\node (c2) at (0, \radius){};
			\node (c3) at (\radius,\radius){};
			\node (c4) at (\radius, 0){};
			\path[-, black, draw]
			(c1) edge (c2)
			(c3) edge (c4);
			\path[black, draw, dashed]
			(c1) edge(c4)
			(c2) edge(c3);
		\end{scope}
	\end{tikzpicture}
	\caption{The dihedral group $\DihedralGroup{5}$ on $10$ vertices
		and the $\DihedralGroup{2}$ on $4$ vertices in standard form.
	The two relations are drawn in different line styles.}\label{fig:regular-drawing}
\end{figure}
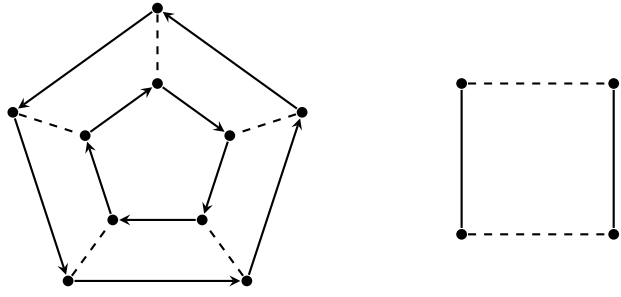
For cyclic groups we require two relations only for technical reasons: we do so
to allow for a uniform treatment of the cyclic and the dihedral case.
One can always pick $\rel_i^\Struct = \rel_j^\Struct$
as we see in the following.
For dihedral groups we really need two relations for the group $\DihedralGroup{2}$,
because for $\DihedralGroup{2}$ the directed cycles are just two undirected edges
and they cannot not be distinguished from the perfect matching.

\begin{corollary}
	Let  $\Struct = (\StructP, \rel_1^\Struct, \dots, \rel_k^\Struct, \spleq)$
	be a structure with dihedral colors.
	Suppose $\colclass \in \colorclasses{\Struct}$ is in standard form
	induced by the relations $\rel_i^\Struct$ and $\rel_j^\Struct$.
	Then $\autgroup{\colstruct} = \autgroup{(\colclass, \rel_i^\Struct, \rel_j^\Struct)}$.
\end{corollary}

We show that every regular and dihedral color class can be converted to standard form.

\begin{lemma}
	\label{lem:dihedral-transitive-colors-draw-automorphism-group}
	There is a CPT term that given a $q$-bounded structure $\Struct$ and
	a regular and dihedral color class $\colclass \in \colorclasses{\Struct}$
	defines homogeneous relations $\rel_1'$ and $\rel_2'$ on $\colclass$
	that induce the standard form of $\colclass$.
\end{lemma}
\begin{proof}
	Let $\Struct = (\StructP, \rel_1^\Struct, \dots \rel_k^\Struct, \spleq)$
	be a $q$-bounded structure,
	$\colclass \in \colorclasses{\Struct}$
	be a regular and dihedral color class,
	$\group := \autgroup{\colstruct}$,
	and $n := |\colclass|$.
	We compute the $2$-orbits of $\group$.
	There are only two cases because $\colclass$ is regular:
	\begin{enumerate}
		\item Assume that $\group \iso \CyclicGroup{n}$.
		Let $\rot$ be a rotation of order $n$.
		Let $O$ be an orbit such that $(u,\rot(u)) \in O$ for some
		$u \in \colclass$.
		By iterating $\rot$, $(\rot^k(u), \rot^{k+1}(u)) \in O$
		and thus~$O$ contains a directed cycle of length $n$
		because $\group$ is transitive. 
		If $|O| > n$, then $n < |\group| = n$.
		So $O$ forms a directed cycle.
		To define such an orbit in CPT,
		we just pick the smallest with the mentioned properties
		(Lemma~\ref{lem:compute-and-order-orbits}).
		We define 
		$\rel_1' = \rel_2' := O$.
		
		\item Assume $\group \iso \DihedralGroup{n / 2}$.
		By the reasoning for cyclic groups above,
		there is a $2$-orbit $O_1$ containing a directed cycle $O' \subseteq O_1$
		of length $n/2$.
		Because $\colclass$ is regular
		all rotations $\rot \in \group$ map 
		a vertex contained in $O'$ to a vertex contained in $O'$
		and all reflections map vertices in $O'$ to vertices in
		$O_1\setminus O'$.
		Hence $O_1$ contains two disjoint directed cycles of length $n/2$.
		Because $|\group| = n$, $O_1$ is the union of these two cycles.
		
		Let $u$ and $v$ be vertices not in the same cycle.
		Then (by the argument above),
		there is a reflection $\refl \in \group$
		with $\refl(u) = v$ and $\refl(v) = u$.
		We pick the orbit $O_2$ with $(u,v), (v,u) \in O_2$.
		Let $\rot \in \group$ be a rotation of order $n/2$.
		Then $(\rot^k(u), \rot^k(v)), (\rot^k(v), \rot^k(u)) \in O_2$.
		So $|O_2| \geq n$,
		that is, $O_2$ exactly contains these elements
		and thus is a perfect matching between the two directed cycles in $O_1$.
		
		In CPT we first pick the minimal orbit $O_1$
		satisfying the required conditions,
		and then the minimal orbit $O_2$
		satisfying the conditions with respect to~$O_1$.
		We then set $\rel_i' := O_i$ for $i \in [2]$.
		\qedhere
	\end{enumerate}
\end{proof}

By iteratively applying the previous lemma,
we can assume that up to CPT-definable preprocessing
the color classes of a structure are in standard form.
The preprocessing is clearly canonization preserving,
because we just have to remove the added relations.

\newcommand{\coeffc}{a}
\newcommand{\coeffmat}{A}

\newcommand{\constL}{b}
\newcommand{\constR}{d}

\newcommand{\cvectL}{y}
\newcommand{\cvectR}{z}
\newcommand{\vvect}{x}

\section{Cyclic Linear Equation Systems in CPT}
\label{sec:cyclic-linear-equations-systems}

Before we begin to canonize structures with dihedral colors,
we need to discuss a special class of linear equation systems.
These linear equation systems are later used in the canonization procedure
to encode the canonical labellings.
We start with the definition of cyclic linear equations systems
from~\cite{AbuZaidGraedelGrohePakusa2014}.
\begin{definition} [Cyclic Linear Equation System]
	Let $\Vars$ be a set of variables.
	A \defining{cyclic constraint} on $W \subseteq \Vars$ is a consistent
	set of linear equations
	containing for each pair of variables $\varA,\varB \in W$
	an equation of the form $\varA - \varB = \constR$ for $\constR \in \ZZ_q$.
	A \defining{cyclic linear equation system (CES)} over $\ZZ_q$
	(for $q$ a prime power)
	is a triple $(\Vars,S,\spleq)$
	where $\spleq$ is a total preorder and
	$S$ is a linear equation system over $\Vars$
	that contains a cyclic constraint on each
	$\spleq$-equivalence class.
\end{definition}
Like for structures, the total preorder $\spleq$
induces a total order on the $\spleq$-equivalence classes,
which we call \defining{variable classes}.
In CPT, the linear equation system $S$ is represented by a set of constraints.
Constraints itself are encoded as sets, too.
The usual encoding using a matrix is not possible in general,
because a matrix induces a total order on both,
the variables and the constraints.

\begin{theorem}[\cite{AbuZaidGraedelGrohePakusa2014}]
	Solvability of CESs over $\ZZ_q$ can be defined in CPT.
\end{theorem}

\newcommand{\assignment}{\xi}
\newcommand{\hyperterm}{T}
\newcommand{\coeff}[2]{\mathsf{c}_{#1}(#2)}

We relax the requirement on the order on the variables being total:

\begin{definition}[Tree-Like Cyclic Linear Equation System]
	A \defining{tree-like cyclic linear equation system (TCES)}
	over $\ZZ_q$ ($q$ a prime power)
	is a tuple $(\Vars, S, \spleq)$ with the following properties:
	\begin{itemize}
		\item The variable classes form a rooted tree
				with respect to being a direct successor in $\spleq$. (I.e., 
		$\spleq$ is a preorder such that for all 
		$\varA,\varB,\varC \in \Vars$ with $\varA \spless \varB$, $\varA \spless \varC$,
		and $\varB$ and $\varC$ incomparable ($\varB \not\spleq \varC$ and $\varC \not\spleq \varB$)
		there is no $\varA' \in \Vars$ with $\varB \spleq \varA'$ and $\varC \spleq \varA'$.)
		
		\item $S$ is a linear equation system on $\Vars$
		containing for every variable class a cyclic constraint.
		\item For every constraint $\sum_i \coeffc_i\var_i = \constR$ 
		with $\var_i \in \Vars$ and $\coeffc_i,\constR \in \ZZ_q$
		in $S$
		every pair of variables $\var_i,\var_j$ is $\spleq$-comparable. 
		(That is, a constraint can only use the variables from the classes
		on a root-to-leaf path of the tree.)
	\end{itemize}
\end{definition}
Note that TCESs are a strict generalization of CESs.
We show that solvability of a certain subclass of TCESs can be defined in CPT.
In principle, we follow the same strategy of \cite{pakusa2015, AbuZaidGraedelGrohePakusa2014}
to solve CESs.
Thus, before we turn to solving TCESs, we sketch the method to solve CESs.
As a first step, we preprocess a (T)CES,
such that every variable class
contains exactly $q$ variables~\cite[Lemma~5.3]{pakusa2015}
and that $\varA-\varB = \constR \neq 0$ for all equations with $\varA \neq \varB$
in the cyclic contraints.
This step is necessary to define hyperterms.

\subsection{Hyperterms}
\newcommand{\hyper}[1]{\text{Hyp}(#1)}
\newcommand{\assigns}{\mathbb{L}}

Let $\ces = (\Vars, S, \spleq)$ be a CES 
with variable classes $\colclass_1 \spless \cdots \spless\colclass_n$.
For every two variables $\varA,\varB\in\colclass_i$ in the same variable class
there is a constraint $\varA-\varB = \constR$ for some $\constR \in \ZZ_q$.
Hence, we can pick any variable $\varA \in \colclass_i$ and
substitute all other variables of the class $\colclass_i$ with a term of the form $\varA -\constR$.
Of course, we cannot do this in CPT, because we cannot choose~$\varA$ canonically.
Anyway, assume for the moment that we picked for every variable class one variable
and substituted all other variables in that fashion.
Then we obtain a linear system of equations $\ces'$
with a total order on the variables,
which can be extended to a total order on the equations.
Hence we can write the system in matrix form $\coeffmat\vvect + \cvectL = \cvectR$,
where $\coeffmat$ is the $m\times n$ coefficient matrix and $\cvectL, \cvectR \in \ZZ_q^m$. Here~$\cvectR$ is the original right hand side of the equation system $S$ and~$\cvectL$ collects the coefficients that arose from substituting the variables.
Now, consider another choice when picking the variables, yielding another system $\ces''$.
Then we can write $\ces''$ as $\coeffmat\vvect + \cvectL' = \cvectR$
with the same coefficient matrix $\coeffmat$ and the same right-hand side $\cvectR$.
The only change is the vector $\cvectL'$, the difference coming from the different constants arising when substituting two variables of the same class.

Hyperterms are hereditarily finite sets built from variables and constants.
They constitute a succinct, simultaneous encoding of all possible ways to pick the variables
and the resulting vectors $\cvectL'$. They are invariant under permuting variables of a class.
The encoding is performed in a way that allows us to use hyperterms to mimic algebraic operations involving
linear terms over the variables in polynomial time.
We summarize properties and capabilities of hyperterms,
without explicitly describing how hyperterms are constructed and without providing proofs (see \cite{pakusa2015, AbuZaidGraedelGrohePakusa2014}).

Clearly, only assignments $\Vars \to \ZZ_q$ satisfying the cyclic constraints of a (T)CES
are reasonable choices when checking whether a (T)CES is consistent.
We call these assignment \textbf{reasonable assignments},
write $\assigns_\ces$ for the reasonable assignments of $\ces$,
and leave out the subscript if $\ces$ is clear from the context.
Let $\hyperterm, \hyperterm_1,$ and  $\hyperterm_2$ be hyperterms.
\begin{enumerate}
	\item For each variable class $\colclass_i$ there is a CPT term defining the coefficient $\coeff{i}{\hyperterm}$
	of $\colclass_i$ in~$\hyperterm$.
	\item Under an assignment $\assignment \in \assigns$
	hyperterms can be evaluated to $\hyperterm[\assignment] \in \ZZ_q$.
	Given~$\assignment$ this can be done in CPT (but we usually do not have access to a single assigment~$\assignment$).
	\item Given a choice of a variable $\varA_i \in \colclass_i$ for each $i \in [n]$
	there is an equivalent linear term $t := \sum_{i \in [n]} \coeffc_i \varA_i + \constL$
	with $\coeffc_i := \coeff{i}{\hyperterm}$
	for $\hyperterm$, that is
	$T[\assignment] = t[\assignment] := \sum_{i \in [n]} \coeffc_i \assignment(\varA_i) + \constL$ for all assignments $\assignment$.
	\item Every constant $\constL \in \ZZ_q$ and every variable $\varA \in \Vars$
	is a hyperterm.
	\item There are CPT terms realizing
	addition and scalar multiplication on hyperterms
	that behave as expected with respect to evaluation and coefficients.
	That is,
	$\hyperterm_1[\assignment] + \hyperterm_2[\assignment] = (\hyperterm_1 + \hyperterm_2)[\assignment]$,
	$\constL \cdot \hyperterm[\assignment] = (\constL \cdot \hyperterm)[\assignment]$,
	in particular
	$\coeff{i}{\hyperterm_1 + \hyperterm_2} = \coeff{i}{\hyperterm_1} + \coeff{i}{\hyperterm_2}$ and
	$\coeff{i}{\constL \cdot \hyperterm} = \constL  \cdot \coeff{i}{\hyperterm}$
	for all $i \in [n]$, $\constL \in \ZZ_q$, and $\assignment \in \assigns$.
	Addition on hyperterms is \emph{not} associative and we stipulate evaluation to be from left to right.
	\item A hyperterm $\hyperterm$ is called \defining{constant}
	if $\hyperterm[\assignment] = \constR$ for some constant $\constR \in \ZZ_q$
	and all reasonable assignments $\assignment \in \assigns$.
	We call $\constR$ the \defining{value} of $\hyperterm$.
	There is a CPT term that defines the value of constant hyperterms.
	Necessarily $\coeff{i}{\hyperterm} = 0$ for all $i \in [n]$
	if $\hyperterm$ is constant.
	\item There is a CPT term that extends the preorder $\spleq$
	to all linear equations $t = \constR \in S$ and
	translates a linear term $t$ to an equivalent hyperterm~$\hyperterm_t$
	(and hence a linear equation $t = \constR \in S$ into
	an equivalent hyperequation $\hyperterm_t = \constR$)
	such that $\spleq$\nobreakdash-equivalent linear equations
	are translated into the same hyperequation,
	in particular the linear equations have the same solutions.
	So we obtain
	an equivalent system of hyperequations $\hyper{\ces}$
	with a total order on the hyperequations.
	\item The operations preserves set-theoretical membership.
	If the transitive closure of $\hyperterm_1 + \hyperterm_2$ (respectively of $\constR \cdot \hyperterm$)
	contains a variable $\varA$,
	then the transitive closure of $\hyperterm_1$ or $\hyperterm_2$ (respectively of $\hyperterm$)
	contains $\varA$.
\end{enumerate} 
Only the points~6.\ and~7.\ depend on the totality of the preorder $\spleq$.
Because $\hyper{\ces}$ is equivalent to $\ces$ (i.e.,~ they have the same satisfying assignments)
we can focus on checking ordered systems of hyperequations for consistency. 

Because $\hyper{\ces}$ is ordered
(both the variable classes and the hyperequation are),
we can define the coefficient matrix of $\ces$ as 
$\coeffmat:= \left(\coeff{i}{T_j}\right)_{j \in [m],i \in [n]}$.
We now use $m$ for the number of hyperequations in $H(\ces)$
and denote with $T_j$ the hyperterm of the $j$-th hyperequation $T_j = \constR_j$.

\subsection{Gaussian Elimination for Rings $\ZZ_q$}
\label{sec:gaussian-for-rings}

We want to apply a variant of Gaussian elimination
adapted for the ring $\ZZ_q$ with $q$ a prime power
(which is in particular only a field if $q$ is a prime).
Recall that Gaussian elimination uses elementary row operations
to convert the coefficient matrix to an upper triangular matrix.
Then consistency can be tested by
checking whether all \defining{atomic} equations,
i.e., equations with all coefficients $0$,
are consistent.
From a CPT perspective, it is important to have a total order on the variable classes
to pick the next class whose coefficient should be eliminated.
We also need the total order on the constraints
to pick the unique constraint, in which the coefficient is not to be eliminated.
We now sketch 
how the consistency check needs to be adapted for rings $\ZZ_q$,
in particular which additional constraints must be satisfied by the coefficient matrix such that we can check consistency similarly to fields.
For simplicity,
consider an ordinary linear system of equations in matrix form $\coeffmat\vvect + \cvectL= \cvectR$
with upper triangular coefficient matrix 
\[\coeffmat = \left( \begin{matrix}
\coeffc_{1,1}  & & &&\cdots &\coeffc_{1,n} \\
0 & \coeffc_{2,2} &&&\cdots&\coeffc_{2,n}  \\
\vdots &\ddots & \ddots &&& \vdots\\
0  & \cdots & 0 & \coeffc_{k,k} & \cdots & \coeffc_{k,n}\\
0 & \cdots &  & 0 &\cdots  & 0 \\
\vdots & &&&\ddots& \vdots \\
0 &\cdots &&&& 0
\end{matrix}   \right)\]
such that $\coeffc_{i,i} \neq 0$ for $i \leq k$
and variables $\vvect = (\var_1, \dots , \var_n)$.
Assume we found an assignment $\assignment$
for the variables $\var_{j+1}, \dots, \var_n$ ($j+1 \leq k$)
satisfying all equations that use only these variables.
We consider the $j$-th equations
$\sum_{j \leq i \leq n} \coeffc_{j,i} \var_i + \cvectL_j = \cvectR_{j}$.
We solve for the variable $\var_{j}$
and obtain
$\coeffc_{j,j}\var_{j} = \cvectR_{j} - \cvectL_j - \sum_{j < i \leq n} \coeffc_{j,i} \assignment(\var_i)$.
In a field we can divide both sides by $\coeffc_{j,j}$,
but in the ring $\ZZ_q$ this is not possible in general.
To overcome this problem,~\cite{pakusa2015}~makes use of the fact
that divisibility in $\ZZ_q$ for prime powers $q$ is a total preorder
and rearranges the order of the variables as follows:
$\varA$ is considered smaller than~$\varB$, 
if there is some coefficient of $\varA$ in some equation
that divides every coefficient of $\varB$.
If that is the case for both $\varA$ and $\varB$,
they are ordered according to the given order.
Then the variables are eliminated in this order.
When eliminating variable $\varA$,
the equation containing this minimal coefficient
(with respect to divisibility) is picked to remain.
In all other equations $\varA$ is eliminated
(if there are multiple such equations,
we pick one using the total order on the equations).
Consequently, it holds that $\coeffc_{j,j} | \coeffc_{j,i}$ for all $i \in [n]$
and that $\coeffc_{j,j} | \coeffc_{j+1, j+1}$.
Hence, once we ensure that $\coeffc_{j,j}$ also divides $(\cvectR_j- \cvectL_j)$,
we can divide by $\coeffc_{j,j}$.
Coefficient matrices satisfying this property
are said to be in Hermite normal form:

\begin{definition}(Hermite Normal Form, \cite{pakusa2015})
	A system of hyperequations is in \defining{Hermite normal form},
	if for some total ordering of its variable classes
	and some order of its equations
	its coefficient matrix $\coeffmat$ satisfies the following conditions:
	\begin{itemize}
		\item $\coeffmat$ is upper triangular.
		\item $\coeffc_{j,j} | \coeffc_{i,i}$ for all $j < i \leq k$,
		where $k$ is maximal such that $\coeffc_{k,k} \neq 0$.
		\item $\coeffc_{j,j} | \coeffc_{i,j}$ for all $j \in [k]$ and $i \in [n]$.
	\end{itemize}
\end{definition}

To check a system of hyperequations for consistency, it is not sufficient to check the atomic equations.
Recall from the case of linear equations above,
that we have to ensure that $\coeffc_{j,j} | (\cvectR_j - \cvectL_j)$
and that for a hyperequation $T_j = \cvectR_j$,
there is always an equivalent linear equation
$\sum_{i \in [n]} \coeffc_{j,i}\var_i + \cvectL_j = \cvectR_j$.
Also recall that we do not have access to $\cvectL_j$ directly,
because it depends on the choice of the $\var_i$.

It turns out, that a system of hyperequation in Hermite normal form
is consistent if and only if the following two conditions hold \cite[Lemma~5.21]{pakusa2015}:
\begin{enumerate}
	\item Every atomic hyperequation is consistent.
	\item For every non-atomic hyperequation $T = \constR$ the following holds:
	Let $q$ be a power of $p$ and $p^\ell$ be the smallest power of $p$
	annihilating all coefficients of $T$, that is $p^\ell T$ is constant.
	Then $p^\ell T = p^\ell \constR$ is consistent.
\end{enumerate}
From the second condition
we can deduce for non-atomic hyperequations that
$0 = p^\ell(\cvectR_j - \cvectL_j)$
and that $\coeffc_{j,j} | (\cvectR_j - \cvectL_j)$.
We show this in more detail later in our extension of the algorithm for TCESs.
It is clear that the former conditions can be checked in CPT,
because we can define the value of constant hyperterms.

\subsection{Solving TCESs in CPT}

Let $\tces = (\Vars, S, \spleq)$ be a TCES over $\ZZ_q$.
We now write $\colorclasses{\tces}$ for the set of all variable classes
and $\coeff{\colclass}{\hyperterm}$
for the coefficient of the variable class $\colclass$
in a hyperterm $\hyperterm$.
We call the directed tree on the variable classes induced by $\spleq$
the \defining{variable tree},
where there is an edge $(\colclass, \colclass')$
if and only if  $\colclass \spless \colclass'$
and there is no other variable class with
$\colclass \spless \colclass'' \spless \colclass'$.
After converting $\tces$ to a system of hyperequations,
we can transform its coefficient matrix to be upper triangular in CPT.
This can be achieved by first eliminating the variable classes
that are leaves in the variable tree, as explained below.
However, as argued in the previous section,
to obtain an upper triangular coefficient matrix it not sufficient,
since~$\ZZ_q$ may not be a field.
We could try to eliminate first the variable class with minimal coefficient
with respect to divisibility as it was done for CESs.
But then it may occur that a variable class, which is to be eliminated,
is not a leaf.
For non-leaf variable classes, we are not able to pick a unique equation to retain (while all others are eliminated).
In other words, reordering variables with regard to divisibility of the coefficients may not be compatible with the tree-like structure of the equation system.

Overall, we are unsure how to check consistency of a TCES in CPT.
Hence, in the following, we will only look at a certain restricted class of TCES,
where divisibility of the coefficients of ``critical'' variable classes
is not important.
As we argue later, this is sufficient for our overall goal of canonizing structures with dihedral colors.
Before defining this class of TCESs,
we need to introduce some terminology for the structure of the variable tree.

\begin{definition}
Let $\tces = (\Vars, S, \spleq)$ be a TCES.
A set $L \subseteq \colorclasses{\tces}$ of variable classes is called
a \defining{local component} if it satisfies the following
(cf.~Figure~\ref{fig:local-components}):
\begin{itemize}
	\item The induced graph of $L$ in the variable tree is a path.
	\item If $\colclass \in L$ is of degree at least $2$ in the variable tree,
	then the children of $\colclass$ are not in $L$.
	\item $L$ is maximal with respect to set inclusion.
\end{itemize}
A variable occurs in an equation if its coefficient in the equation is non-zero.
A variable is \defining{local} if 
in every equation in which it occurs,
only variables of the same local component occur, too.
Other variables are called \defining{global}.
An equation is \defining{local} if it contains at least one local variable and \defining{global} otherwise.
\end{definition}

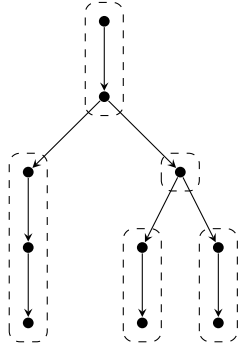
\begin{figure}
	\centering
	\begin{tikzpicture}[>=stealth,every node/.style={vertex},
	]
		\node (R) at  (0,0) {};
		\node (V1) at (0,-1) {};
		\node (V21) at (-1, -2) {};
		\node (V22) at (1 , -2) {};
		\node (V31) at (-1, -3) {};
		\node (V32) at (0.5,-3){};
		\node (V33) at (1.5,-3) {};
		\node (V41) at (-1, -4) {};
		\node (V42) at (0.5,-4){};
		\node (V43) at (1.5,-4) {};
		
		\path[->]
			(R) edge (V1)
			(V1) edge (V21)
			(V1) edge (V22)
			(V21) edge (V31)
			(V22) edge (V32)
			(V22) edge (V33)
			(V31) edge (V41)
			(V32) edge (V42)
			(V33) edge (V43);
			
		\draw[rounded corners, dashed] (-0.25, -1.25) rectangle (0.25, 0.25) {};
		\draw[rounded corners, dashed] (-1.25, -4.25) rectangle (-0.75, -1.75) {};
		\draw[rounded corners, dashed] (0.75, -2.25) rectangle (1.25, -1.75) {};
		\draw[rounded corners, dashed] (0.25, -4.25) rectangle (0.75, -2.75) {};
		\draw[rounded corners, dashed] (1.25, -4.25) rectangle (1.75, -2.75) {};
	\end{tikzpicture}
	\caption{The variable tree (a vertex represents a variable class)
		of a TCES with all local components.}
	\label{fig:local-components}
\end{figure}

On the local components the preorder $\spleq$ induces a tree
in which every local component has degree at least $2$ or is a leaf.
Note that local equations only use variables of the same local component,
but a global equation may also be restricted to just one local component
(if all occurring variables are global).
Furthermore, in a local equation
the coefficient of global variables of the same local component
can be non-zero.

We extend the definition to variable classes:
A variable class is called global, if it contains some global variable.
Likewise, the other variable classes are called local.
Note that a global variable class can contain both,
global and local variables
and that we cannot simply split the class
because we need to maintain that every class contains $q$ variables.
However, when we are working with hyperterms,
it will not be important, whether there is a local variable in a global variable class,
but only that the class occurs in global equations.

\begin{definition}
	A TCES $\tces = (\Vars, S, \spleq)$ over $\ZZ_q$ is called
	\defining{weakly global}, if
	\begin{itemize}
		\item $q$ is a power of an odd prime and
		every equation (equivalently every variable) is local  or
		\item $q=2^\ell$ is a power of $2$
		and for every global variable $\var \in \Vars$
		there is an equation $2\var = 0 \in S$.
	\end{itemize}
\end{definition}
Note that every CES is an weakly global TCES, because a CES only has one local component.
We first note that weakly global TCESs for odd prime powers are easy:
\begin{corollary}
	\label{cor:solve-tces-odd-prime-powers}
	Solvability of weakly global TCESs over $\ZZ_q$ for odd prime powers $q$ is CPT definable.
\end{corollary}
\begin{proof}
	A weakly global TCESs for odd prime powers consists only of local equations,
	so is just a disjoint union of CESs,
	one for each local component.
	Hence we can solve each CES
	by the methods of~\cite{AbuZaidGraedelGrohePakusa2014, pakusa2015}
	in parallel.
\end{proof}

In the case that $q = 2^\ell$ is a power of $2$,
the constraints connecting these CESs are formed by variables
that can be either $0$ or $2^{\ell-1}$.
Hence, these variables are actually over~$\ZZ_2$, but embedded in $\ZZ_{2^\ell}$.
In $\ZZ_2$, divisibility of coefficients does not become an issue,
as there is only one non-zero coefficient.

We use the notion of local and global equations for hyperequations, too,
and adapt the definitions for
tree-like and weakly global systems of hyperequations:
Let $\tces = (\Vars, S, \spleq)$ be a system of hyperequations.
We call $\tces$ tree-like, if $\tces$ satisfies the following:
The preorder~$\spleq$ induces a tree on the variable classes
and
for every hypereqution $\hyperterm = \constR$
the transitive closure of $\hyperterm$ (set-theoretically) only contains variables which are $\spleq$\nobreakdash-comparable\footnote{
	This is a stronger condition than requiring that all variable classes
	with non-zero coefficient are $\spleq$\nobreakdash-comparable,
	because some hyperterm operations may set the coefficient of a variable class to zero,
	but the hyperterm as a set only grows,
	so its transitive closure still contains the variables of that class.}.
Additionally,
there is a CPT-definable preorder on the hyperequations,
such that for every path from the root to some leaf in the variable tree,
the preorder is a total order on the subset $S' \subseteq S$ of hyperequations
in which only variable classes contained in said path have non-zero coefficients.
Finally, $\tces$ is called weakly global,
if it satisfies the same conditions a weakly global TCES does
but for variable classes rather than for variables.
Note that we can define what the global variables are for $\tces$:
Let $T=0$ be a hyperequation with coefficient $2$ for a global variable class
and coefficient $0$ otherwise.
Then there is only one equivalent linear term $2\var = 0$
(all others are $2\var + \constL = 0$ with $\constL\neq 0$)
and the variable $\var$ is global.

When converting a weakly global TCES to a system of hyperequations,
then the resulting system is tree-like  and weakly global.
The order on the hyperequations
for a path from the root to a leaf in the variable tree
is given by the same construction
as the order on all hyperequations in the case of a CES
because the restriction of $\tces$ to these equations is indeed a CES.

\subsection{Systems of Hyperequations over $\ZZ_{2^\ell}$}

As explained before, we can solve weakly global TCESs over $\ZZ_q$
if $q$ is an odd prime power.
So fix $q := 2^\ell$ to be a power of $2$.
We first make the following simple observation:
\begin{lemma}
	\label{lem:zz2-in-zz-d}
	Let $\tces$ be a weakly global TCES.
	If $\tces$ is consistent, then the following conditions hold
	for all global variable classes $\colclass$:
	\begin{enumerate}[label=\alph*)]
		\item There are at most two global variables $\varA,\varB \in \colclass$.
		\item Let $\varA,\varB \in \colclass$  be global variables with $\varA \neq \varB$.
		Then the equation
		$\varA-\varB=2^{\ell-1}$
		is contained in the cyclic constraint for $\colclass$.
	\end{enumerate}
\end{lemma}
\begin{proof}
	For a), assume there are $3$ global variables $\varA_1, \varA_2, \varA_3 \in \colclass$
	and hence there are equations $2\varA_i = 0$ for all $i \in [3]$,
	so $\assignment(\varA_i) \in \set{0, 2^{\ell-1}}$ for all satisfying assignments $\assignment$ and $i \in [3]$.
	But by the cyclic constraints we have that $\assignment(\varA_i) \neq \assignment(\varA_j)$ for all $i \neq j$.
	
	For b), assume that the equation $\varA - \varB = \constR$ for $ \constR \neq 2^{\ell-1}$
	is contained in the cyclic constraints.
	Then all reasonable assignments necessarily violate
	$\assignment(\varA), \assignment(\varB) \in \set{0, 2^{\ell-1}}$.
\end{proof}
The conditions express formally that we really embed $\ZZ_2$ in $\ZZ_q$
and they can surely be checked in CPT.
So we assume in the following that all global variable classes satisfy the conditions.
We adapt the required normal form for the coefficient matrix:
\begin{definition}[Local Hermite normal form]
	Let $\tces = (\Vars, S, \spleq)$ be a tree-like system of hyperequations
	over $\ZZ_{2^\ell}$.
	We say that $\tces$ is in \defining{local Hermite normal form}
	if there are (not necessarily CPT-definable) total orders
	on the variable classes
	$\colclass_1 \spless \cdots \spless \colclass_n$
	and on the equations
	such that the coefficient matrix $\coeffmat$ satisfies the following:
	
	$\coeffmat$ is upper triangular
	(whereby the coefficients of global variables are taken modulo~$2$)
	and for every local component $L$
	the submatrix $\coeffmat_L$
	of all local equations of $L$ has the following properites:
	\begin{itemize}
		\item All global variable classes of $L$ are at the end of the linear order.
		\item $\coeffmat_L$ has the following shape:
		\[\coeffmat_L = \left( \begin{matrix}
		\coeffmat_L^\text{local} & \coeffmat^\text{global}_L
		\end{matrix}   \right)\]
		where $\coeffmat_L^\text{local}$ is the submatrix
		of all columns of the local variable classes.
		\item $\coeffmat_L^\text{local}$ is in Hermite normal form.
	\end{itemize}
	In that case we also say that $L$ is in local Hermite normal form.
\end{definition}
Note that $\coeffmat_L^\text{local}$ does not contain zero rows,
because the equation encoded by a zero row
would be either global or atomic -- both not local equations of $L$.
We are only interested in the coefficients for global variable classes
modulo $2$,
because the values of these variables are~$0$ or of order~$2$
and hence only the pairity of coefficients is important.
Note that we cannot simply replace the coefficients
by the pairity,
because we do not have such an operation on hyperterms available.

Before we address solvability of weakly global TCESs in local Hermite normal form,
we revisit the problem of defining the value of a hyperterm.
Now, we are not interested anymore in constant hyperterms $\hyperterm$,
which are those for which 
$\hyperterm[\assignment]$ is constant for all reasonable assignments $\assignment \in \assigns$.
Rather we will consider an even more restricted set of assignments.
In the case of global variables only assignments are of interest
that assign global variables values of order at most $2$,
because otherwise the $2\varA =0$ constraints are violated.
We refine our definition of reasonable assignments for weakly global TCESs.

\begin{definition}
	For a weakly global TCES $\tces= (\Vars, S, \spleq)$ over $\ZZ^{2^\ell}$
	an assignment $\Vars \to \ZZ_{2^\ell}$ is called \defining{reasonable},
	if it satisfies all cyclic constraints
	and all $2\varA = 0$ constraints for all global variables $\varA$.
\end{definition}
Now, we can keep our definitions of constant hyperterms.
Regarding the coefficient of constant hyperterms,
we now obtain the following characterization:
\begin{lemma}
	\label{lem:order-2-constant-hyperterm}
	Let $\tces$ be a weakly global TCES.
	A hyperterm $T$ 
	is constant
	if and only if
	$\coeff{\colclass}{T} = 0$ for all local variable classes $\colclass$
	and
	$\coeff{\colclass}{T} \equiv 0~(\text{mod}\  2)$
	for all global color classes~$\colclass$.
\end{lemma}
\begin{proof}
	The condition on the local variable classes is the same
	as for CES
	(recall that, restricted on a local component, a TCES is just a CES).
	So we can assume that $\coeff{\colclass}{T} = 0$ for all local variable classes $\colclass$.
	
	Consider an equivalent linear term $t = \sum_{\colclass \in \colorclasses{\Vars}} \coeffc_\colclass\var_\colclass + \constL$ of $T$.
	Assume $\coeffc_\colclass = \coeff{\colclass}{T} \equiv 1~(\text{mod}\  2)$
	for some global color class $\colclass$.
	Let $\assignment$ be a reasonable assignment
	(which always exists due to the assumption
	that the conditions of Lemma~\ref{lem:zz2-in-zz-d} hold).
	Let $\assignment'$ be the assignment that assigns the same values as $\assignment$
	but $\assignment'(\var) = \assignment(\var) + 2^{\ell-1}$
	for all variables $\var\in\colclass$.
	Then~$\assignment'$ is still reasonable,
	but evaluates differently for $t$:
	$t[\assignment] = t[\assignment'] + 2^{\ell-1}$.
	
	For the other direction, assume that
	$\coeff{\colclass}{T} \equiv 0~(\text{mod}\  2)$ for all global color classes~$\colclass$.
	Let~$\assignment$ be a reasonable assignment.
	Then $t[\assignment] = \sum_{\colclass \in \colorclasses{\Vars}} \coeffc_\colclass\assignment(\var_\colclass) + \constL = \constL$
	because the coefficients of all local variable classes are~$0$
	and $\coeff{\colclass}{T}$ annihilates values of order at most~$2$
	for all global color classes.
\end{proof}

\begin{lemma}
	\label{lem:values-for-order-k-constant-hyperterms}
	Let $\tces$ be a weakly global system of hyperequations over $\ZZ_{2^\ell}$
	and $\hyperterm$ be constant hyperterm of $\tces$
	such that its transitive closure (set-theoretically) contains only variables from a single root-to-leaf path
	in the variable tree.
	There is a CPT term that on given such a hyperterm $\hyperterm$ defines the value of $\hyperterm$.
\end{lemma}
\begin{proof}
	Recall that a hyperterm is a nested set of variables and constants,
	in particular it can set-theoretically contain (at some level)
	a variable, even if its variable class has coefficient $0$.
	
	A TCES restricted to one root-to-leaf path is a CES.
	We apply the same procedure to define the value of $\hyperterm$
	as in~\cite[Lemma~5.11]{pakusa2015},
	but we have to argue that it works with the global variable classes
	(which only have coefficient $0$ modulo $2$).

	The procedure follows inductively the order of the CES.
	Let $\colclass$ be the smallest variable class (with respect to the order)
	of which a variable is set-theoretically contained in $\hyperterm$.
	We consider all possible reasonable assignments of all variables in these class.
	For each such assignment,
	we syntactically replace all occurrences of the variables
	according to the assignments in $\hyperterm$.
	Because $\hyperterm$ was constant,
	the result is the same for all assignments:
	By Lemma~\ref{lem:order-2-constant-hyperterm}
	the values for global variables are always annihilated,
	because the assignment only uses value of order at most $2$
	and the coefficient is $0$ modulo $2$.
	We then proceed with the next class.
	In the end, we are left with a variable free hyperterm
	whose value can be defined easily as shown in \cite[Lemma~5.11]{pakusa2015}.
\end{proof}

We are now ready to characterize consistent systems of tree-like
hyperequations in local Hermite normal form:

\begin{lemma}
	A weakly global tree-like system of hyperequations $\tces = (\Vars, S, \spleq)$ over $\ZZ_{2^\ell}$
	in local Hermite normal form
	is consistent if and only if the following holds:
	\begin{enumerate}[label=\alph*)]
		\item Every atomic equation is consistent.
		\item For every non-atomic equation $\hyperterm = \constR \in S$ the following holds:
		Let $2^k$ be the minimal power of $2$ that annihilates
		the coefficients of every local variable class in $\hyperterm$.
		If all local variables classes have coefficient zero, we set $k = 1$.
		Then  $2^k \hyperterm$ is constant and
		$2^k \hyperterm = 2^k \constR$ is consistent.
	\end{enumerate}
\end{lemma}
\begin{proof}
	If $\tces$ is consistent, then surely Condition~a) is satisfied.
	Condition b) follows from Lemma~\ref{lem:order-2-constant-hyperterm}
	because  $k\geq 1$ and $\tces$ is weakly global and consistent.
	
	For the other direction,
	assume that Conditions~a) and~b) hold.
	Recall that changing the order of the variable classes or equations
	does not affect consistency.
	We order the variable classes
	$\colclass_1 \spless \cdots \spless \colclass_n$
	as given by local Hermite normality
	and pick for every variable class~$\colclass_i$ one variable $\var_i \in \colclass_i$.
	Let $\coeffmat\vvect + \cvectL = \cvectR$ be the equivalent linear system,
	where $\coeffmat$ is the $m \times n$ coefficient matrix
	as granted by local Hermite normality,
	$\cvectR$ contains the same values as the right hand side of $\tces$,
	and $\vvect = (\var_1, \dots, \var_n)$.
	Such a system exists by the properties of hyperterms.
	Note that we do not want to construct a satisfying assignment in CPT
	but only show that it exists
	and hence we can order the variable classes and pick the variables.
	
	Let $\assignment \in \assigns$
	be a reasonable assignment
	such that $\sum_{i \in [n]}\coeffc_{k,i}\assignment(\var_k) + \cvectL_k = \cvectR_k$
	for all $j < k \leq n$ but not for $k = j$
	(that is, all rows not among the the first $j$ rows in $\coeffmat$ are satisfied)
	and $j$ is minimal.
	If $j = 0$, then $\tces$ is consistent.
	So assume that $j \geq 1$.
	We show that this contradicts minimality of $j$.

	By Condition~a) all atomic equations are consistent
	and so the $j$-th equation is not atomic.	
	Because $\coeffmat$ is upper triangular,
	$\assignment$ satisfies all equations
	in which only $\var_{j+1}, \dots, \var_n$ occur.
	We consider the $j$-th equation $\sum_{i \in [n]} \coeffc_{j,i} \var_i + \cvectL_j = \cvectR_j$,
	which additionally uses the variable $\var_j$.
	Note that we can change $\assignment(\var_j)$ without affecting
	the equations with index~$>j$ because $\coeffmat$ is upper triangular.
	We show that we can find a value $\assignment(\var_j)$
	of appropriate order 
	also satisfying the $j$-th equation.
	We make the following case distinction:
	
	\begin{itemize}
		\item The $j$-th equation is global.
		Then in particular $2T_j = 2 \cvectL_j =  2\cvectR_j$ is consistent
		for the corresponding hyperterm $T_j$ by Condition~b)
		and $2(\cvectR_j-\cvectL_j) = 0$.
		So $\cvectR_j-\cvectL_j$ and
		$\constL := (\cvectR_j-\cvectL_j) - \sum_{i \in [n], i \neq j} \coeffc_{j,i}\assignment(\var_i)$
		are of order at most $2$
		because $\assignment$ is reasonable by assumption
		and all $\varA_i$ with non-zero coefficient must be global.
		Because $\coeffmat$ is in local Hermite normal form,
		$\coeffc_{j,j}\equiv 1\ (\text{mod}\ 2)$
		and we set $\assignment(\var_i) := \constL$.
		
		\item The $j$-th equation is local.
		Then by local Hermite normality the variable $\var_j$ is local
		and $\coeffc_{j,j} | \coeffc_{j,i}$
		for all local variable classes $\colclass_i$.
		
		Let $2^k$ be the smallest power annihilating all these $\coeffc_{j,i}$.
		In particular $k \geq 1$ and $2^k T_j = 2^k\cvectL_j$.
		So $2^k(\cvectR_j - \cvectL_j) = 0$ by Condition~b).
		We show $\coeffc_{j,j} | (\cvectR_j - \cvectL_j)$.
		Every $\constL \in \ZZ_{2^\ell}$ can be written as
		$\constL = 2^{k'} \cdot \constL'$,
		for a unit $\constL' \in \ZZ_{2^\ell}$ of order $2^\ell$
		and $k'$ is independent of the choice of $\constL'$.
		The minimal power of $2$ annihilating $\constL$
		is $2^{\ell -k'}$.
		In particular $\coeffc_{j,j} = 2^{\ell-k} \constR$
		for some unit $\constR$ and $k \geq 0$
		because $\coeffc_{j,j}|\coeffc_{j,i}$
		and $2^\ell$ is the smallest power of $2$ annihilating the~$\coeffc_{j,i}$.
		Next,
		$\cvectR_j-\cvectL_j = 2^{\ell-k'} \constR'$ for $0 \leq k' \leq k$
		for a unit $\constR'$
		because $2^k$ annihilates $\cvectR_j-\cvectL_j$.
		This means that $\coeffc_{j,j} | (\cvectR_j-\cvectL_j)$
		because $2^{\ell-k} | 2^{\ell-k'}$ and $\constR|\constR'$
		(both are units).
		
		Lastly, $\coeffc_{j,j} | \coeffc_{j,i} \assignment(v_i)$
		for all global variables $\var_i$
		because $\coeffc_{j,j} \neq 0 $ and
		$\assignment(\var_i)$ and thus $\coeffc_{j,i} \assignment(\var_i)$
		are of order at most $2$.
		We conclude that 
		$\coeffc_{j,j} | (\cvectR_j - \cvectL_j) - \sum_{i \in [n],i \neq j} \coeffc_{j,i} \assignment(\var_i) =: \constL$
		and thus there is a value $\assignment(\var_j)$
		satisfying $\coeffc_{j,j} \assignment(\var_j) = \constL$.
		\qedhere
	\end{itemize}
\end{proof}
\begin{corollary}
	\label{cor:solvability-hermite-normal-form}
	Solvability of weakly global tree-like systems of hyperequations 
	in local Hermite normal form is CPT-definable.
\end{corollary}
\begin{proof}
	Recall that we first check the conditions of Lemma~\ref{lem:zz2-in-zz-d}.
	If that check is successfull,
	we know that the terms $2^k T$ are constant
	by Lemma~\ref{lem:order-2-constant-hyperterm},
	easiness and $k\geq 1$.
	So we can use
	Lemma~\ref{lem:values-for-order-k-constant-hyperterms}
	to define the value of $2^k T$
	and check for consistency.
\end{proof}

To define solvability of weakly global TCESs in CPT,
it only remains to convert a weakly global tree-like system of hyperequations
into local Hermite normal form.
We treat local and global hyperequations differently.
For the local ones we use Gaussian elimination for rings
as described in Section~\ref{sec:gaussian-for-rings}.
For the global ones Gaussian elimination for fields suffices.

\begin{lemma}
	\label{lem:tces-to-hermite-normal-form}
	There is a CPT term
	that converts a 
	weakly global tree-like system of hyperequations over $\ZZ_{2^\ell}$
	into an equivalent one in local Hermite normal form.
\end{lemma}
\begin{proof}
	We first define the global variable classes
	and rearrange the orders of every local component
	such that the global variables are at the end.
	
	Second, we transform the local variables
	and local hyperequations of each
	local component into Hermite normal form.
	This can be done as described in~\cite{pakusa2015}
	and outlined above
	where we ignore the coefficients of global variables
	(of course, they are manipulated when adding equations,
	but we do not care about divisibility for them and may do so
	as stated by local Hermite normality).
	The local components can be processed in parallel
	because they are independent of each other.

	Third, we need to process the remaining global equations.
	The local components are processed inductively following
	the tree on the local components induced by $\spleq$.
	We now write $\colclass \spleq L$
	if  $\colclass \spleq \colclass'$ for some
	$\colclass' \in L$
	and  $\colclass \spless \colclass'$
	if additionally $\colclass \notin L$.
	Let~$L$ be a local component.
	Assume that for every local component~$L'$
	that is a direct successor of~$L$ 
	the equation system of all equations
	with non-zero coefficient of at least one variable class
	$\colclass \spleq L'$
	is in local Hermite normal form.
	
	We can simply combine these systems (i.e., form the union of all equations)
	for all direct successors of $L$
	and still obtain an equation system in local Hermite normal form.
	We cannot define an order between these successors,
	but any order of them yields an upper triangular coefficient matrix
	because a branch of the variable tree
	cannot use variables of another one.
	
	Let $G$ be the global variable classes of $L$.
	Let $S_G \subseteq S$ be the set of global equations
	in which some variable classes in $G$ have non-zero coefficient
	but all variable classes $\colclass \spless L$ have coefficient zero
	(the others with non-zero coefficient are already in local Hermite normal form).

	The variable classes $G$ are the largest ones
	(with respect to $\spleq$)
	ocurring with non-zero coefficent in $S_G$
	and thus $S_G$ only uses variable classes
	on the path from $L$ to the root of the variable tree.
	Because $\tces$ is tree-like,
	the order on the hyperequations is total on $S_G$.
	Also, the variable classes occuring with non-zero coefficients in $S_G$
	are ordered.
	Thus, we can apply Gaussian elimination on $S_G$
	(recall we are working modulo $2$ in a field)
	to bring the equations in $S_G$ into upper triangular form.
	Note that we only add hyperterms with variables of the same
	root-to-leaf path here, and hence the resulting equation system is tree-like again.
	
	Finally, the coefficient matrix enlarged by the variable classes in $L$
	is in local Hermite normal form.
	Of course, we cannot process the local components one-by-one,
	but we have to process the local components of the
	same level in the tree in parallel.
	This is, as discussed before, possible because they are independent of each other.
\end{proof}

\begin{theorem}
	\label{thm:solvability-tces}
	Solvability of weakly global TCESs over $\ZZ_q$ for prime powers $q$ is CPT-definable.
\end{theorem}
\begin{proof}
	Let $\tces$ be a TCES over $\ZZ_q$.
	If $q$ is a power of an odd prime,
	it can be solved with the procedure for CESs
	(Corollary~\ref{cor:solve-tces-odd-prime-powers}).
	If $q$ is a power of two,
	we first translate $\tces$ into a system of hyperequations,
	apply Lemma~\ref{lem:tces-to-hermite-normal-form}
	to convert it to local Hermite normal form,
	and then solve it with Corollary~\ref{cor:solvability-hermite-normal-form}. 
\end{proof}

We want to note that solving a linear equation system over $\ZZ_{p^\ell}$
can be reduced to solving multiple linear equation systems over $\ZZ_p$
for example by using Hensel's Lemma.
But this technique requires not only to check an equation system
for consistency,
but also to compute a solution.
This cannot be done in CPT in general,
because there are TCESs for which every solution is equivalent to exponentially many other solutions (under automorphisms of the system).

\subsection{Combining Equations Systems}

Given a sequence of CESs, creating a combined CESs to describe solutions satisfying all equations simultaneously is easy. Indeed, two CESs on the same variable set
can just be combined to implement
intersection on their solution spaces
(where we use the order of one CES,
assuming that we can distinguish the two CESs in CPT).
For TCESs the situation is more complicated.
In particular, the tree structures might not be compatible
(so we cannot just take to order of one TCES).
We will now devise a strategy to combine TCESs under certain conditions.
Our solution allows that the variables (and their orders) disagree in the different TCES to a certain extend. 
In the following we write $\solutions{\tces}$ for the solution space of a
TCES.

\begin{definition}
	Let $\tces = (\Vars, S, \spleq)$ be a TCES.
	We say that $\Vars'$ are the \defining{topmost variables} of $\tces$
	if $\Vars'$ is the set of all variables
	of the local component containing the root class of the variable tree~$L_r$
	(formally $\Vars' = \bigcup L_r$).
\end{definition}

\begin{definition}
	\label{def:compatible-vars}
	Let $\tces_i = (\Vars_i, S_i, \spleq_i)$ be two TCESs over $\ZZ_q$ for $i \in [2]$
	using variables $\Vars_i$ and with topmost variables $\Vars'_i$.
	We call $\tces_1$ and $\tces_2$ \defining{compatible},
	if
	$\Vars_1 \cap \Vars_2 = \Vars'_1 \cap \Vars'_2$
	and $\Vars'_1 \cap \Vars'_2$
	is a union of variable classes of $\tces_i$ for all $i\in[2]$.
\end{definition}
The common topmost variables are required to be a union of color classes
of both TCES,
because they can define different orders on them.
Let $\tces_i$ be two TCESs with topmost variables $\Vars'_i$ for all $i \in [2]$
as above and assume that we can distinguish the TCESs in CPT.
That is, there is a CPT term defining the ordered tuple $(\tces_1, \tces_2)$ or equivalently an order $\tces_1 < \tces_2$
(e.g.~the two TCESs are defined by different CPT terms
or can be ordered by their structure).
The union of $\tces_1$ and $\tces_2$ is defined as
\[\tcesunion{\tces_1}{\tces_2} := 
(\Vars_1 \cup \Vars_2, S_1 \cup S_2,
\spleq_1 \cup \spleq_2' \cup \spleq'),\]
where $\spleq_2' = \spleq_2[\Vars_2 \setminus (\Vars_2' \cap \Vars_1')]$
and
$\spleq'$ is defined by 
$\Vars_1' \spless' \Vars_2' \setminus \Vars_1'$ and 
$\Vars'_i \spless' \Vars_j \setminus \Vars'_j$ for
all $i,j \in [2]$.

\begin{lemma}
	\label{lem:union-tces}
	Let $\tces_i = (\Vars_i, S_i, \spleq_i)$
	be two CPT distinguishable TCESs over $\ZZ_q$ for $i \in [2]$
	using variables $\Vars_i$ and with topmost variables $\Vars'_i$.
	If $\tces_1$ and $\tces_2$ are compatible,
	then $\tcesunion{\tces_1}{\tces_2}$
	is again a TCES with topmost variables $\Vars'_1 \cup \Vars'_2$
	and it satisfies
	$\solutions{\tcesunion{\tces_1}{\tces_2}}  = 
	\bigcap_{i \in [2]}	\extExpl{\Vars_1 \cup \Vars_2}{\solutions{\tces_i}}$.
	If both $\tces_1$ and $\tces_2$ are weakly global then
	$\tcesunion{\tces_1}{\tces_2}$ is weakly global, too.
\end{lemma}

\begin{proof}
	The order $\spleq := \spleq_1 \cup \spleq_2' \cup \spleq'$
	as defined above forms a tree on the variable classes:
	by the condition $\Vars_1 \cap \Vars_2 = \Vars'_1 \cap \Vars'_2$
	the only common variables of the two TCESs are the common topmost ones.
	On the common variables, we use the order $\spleq_1$.
	With $\spleq'$ we then order the topmost variables of $\tces_2$
	not common with $\tces_1$ after the topmost variables of $\tces_1$.
	That is, considering $\tces_2$, we just reorder the variable classes of the root local component
	because $\Vars'_1 \cap \Vars'_2$ is a union of $\tces_2$ variable classes.
	Reordering the variable classes of the root local component of $\tces_2$
	does not change its tree structure (i.e., its local components). 
	By construction $\tcesunion{\tces_1}{\tces_2}$ has the topmost variables $\Vars'_1 \cup \Vars'_2$.
	Because the set of equations of $\tcesunion{\tces_1}{\tces_2}$
	is the union of the equations of both TCESs,
	$\solutions{\tcesunion{\tces_1}{\tces_2}}  = 
	\bigcap_{i \in [2]}	\extExpl{\Vars_1 \cup \Vars_2}{\solutions{\tces_i}}$
	follows immediately.
	
	Assume that $\tces_1$ and $\tces_2$ are weakly global.
	Because every local component of $\tces_i$
	is contained in a local component of $\tcesunion{\tces_1}{\tces_2}$
	for all $i \in [2]$,
	every global equation of $\tcesunion{\tces_1}{\tces_2}$
	is a global equation of $\tces_1$ or $\tces_2$.
	It follows that every global variable of $\tcesunion{\tces_1}{\tces_2}$
	is a global variable of $\tces_1$ or $\tces_2$
	and thus that the required constraints of the form $2 \var = 0$
	are present for all global variables.
\end{proof}

We generalize our notation and
write $\tces = (\tces_{p_1}, \dots, \tces_{p_k})$
for a sequence of TCESs over pairwise coprime  prime powers $p_i$
and $\solutions{\tces}$ for the solution space of $\tces$.
A series of TCESs $(\tces_1, \dots , \tces_k)$ with pairwise disjoint variables
has the topmost variables $\Vars' = \Vars'_1 \disunion \cdots \disunion \Vars'_k$
if~$\tces_i$ has the topmost variables $\Vars'_i$ for all $i \in [k]$.

We now want to form the union of two series of TCESs:
Let
$\tces = (\tces_{p_1}, \dots , \tces_{p_{\ell}})$
and $\tces' = (\tces'_{q_1}, \dots, \tces'_{q_{\ell'}})$
be two series of TCESs.
We first ensure that if $p_i$ and $q_j$
or prime powers of the same prime
that then actually $p_i = q_j$.
Assume w.l.o.g.~that $p_i < q_j$.
We turn $\tces_{p_i}$ into a TCES over $\ZZ_{q_j}$
by adding $p_i \cdot \var = 0$ constraints for all variables of $\tces_{p_i}$ embedding $\ZZ_{p_i}$ into~$\ZZ_{q_j}$.

We write $\tcesunion{\tces }{\tces'}$
for the series of TCESs obtained by
forming the union of all~$\tces_{p_i}$  and~$\tces'_{q_j}$ 
whenever $p_i = q_j$.
The remaining TCESs are just copied.
Lemma~\ref{lem:union-tces} generalizes to series of TCESs
by making the assumptions of the lemma for all TCESs
$\tces_{p_i}$  and $\tces'_{q_j}$ for which $p_i$ and $q_i$
are powers of the same prime.

\section{Canonization of Structures with Dihedral Color Classes}
\label{sec:canonization-dihedral}
Recall that for our canonization problem the reduction to normal forms (Theorem~\ref{thm:convert-to-normal-form}) shows that we can assume the input structure to be a dihedral $2$-injective quotient structure.
Our further strategy is as follows: we want to reduce canonization of dihedral $2$-injective quotient structures
to that of structures with abelian color classes
and then apply the canonization procedure
for abelian color classes.
The main idea is to artificially prohibit reflections
in one color class
and then hope that this prohibits reflections in other color classes as well.
For this, we want to
exploit
the classification of $2$-injective subdirect products of dihedral groups
(Theorems~\ref{thm:classification-2inj-sub-dihedral}
and~\ref{thm:classification-2inj-sub-dihedral-cyclic}) saying that
 most $2$-injective subdirect products
are rotate-or-reflect groups.
In particular, if we prohibit reflections in one color class of a 
rotate-or-reflect group
then reflections in the other color classes
are prohibited, too.
This effect of prohibiting reflection continues through
most $2$-injective subdirect products and quotient color classes. However, it does not have to reach all color classes since some $2$-injective subdirect product are not rotate-or-reflect groups (for example if one factor is abelian).
We call the parts of the structure
in which reflections are linked in this way and can only occur simultaneously
\defining{reflection components}.
We analyse how reflection components can depend on each other.
It will turn out, that different reflection components
can indeed only be connected through abelian color classes.
We call these color classes \defining{border color classes}.
Overall, we will follow a two-levelled approach:
on the top level,
we deal with the dependencies between the border (and all other abelian) color classes,
and on the second level
we consider each reflection component on its own 
and how it is embedded in its border color classes.

To ensure that the border color classes are indeed all abelian
we have to forbid the single exception
in Theorem~\ref{thm:classification-2inj-sub-dihedral},
which is not a rotate-or-reflect group, namely the double CFI group.

\begin{definition} [Double-CFI-Free Structure]
	We call a $2$-injective dihedral quotient structure
	\defining{double-CFI-free},
	if for every $\type \in \grouptypes$ the group  $\typegroup{\type}$
	is neither isomorphic to the double CFI group
	$\DoubleCFIgroup$
	nor to
	$\DoubleCFIgroup \cap (\rotsubgroup{\DihedralGroup{4}} \direct \DihedralGroup{4} \direct \DihedralGroup{4})$.
	
\end{definition}

We consider two natural classes of structures 
that are double-CFI-free after applying the preprocessing.
We call a structure $\Struct$ \defining{odd dihedral}
if $\Struct$ is dihedral and for every non-abelian
$\colclass \in \colorclasses{\Struct}$
there is an odd $k$ such that  $\autgroup{\colstruct} \iso \DihedralGroup{k}$.
\begin{lemma}
	\label{lem:odd-dihedral-and-graphs-are-double-CFI-free}
	Let $\Struct$ be a dihedral $q$-bounded structure of arity at most $3$.
	If $\Struct$ satisfies one of the following conditions
	then the $2$-injective quotient structure $\Struct'$
	obtained after applying Theorem~\ref{thm:convert-to-normal-form}
	is double-CFI-free:
	\begin{enumerate}
		\item $\Struct$ is odd dihedral.
		\item $\Struct$ is a graph.
	\end{enumerate}
\end{lemma}
\begin{proof}
	\ 
	\begin{enumerate}
		\item
	Let $\Struct'$
	be the $2$-injective quotient structure
	obtained by the preprocessing steps
	and
	let $\colclass' \in \colorclasses{\Struct'}$
	have a dihedral automorphism group.
	By Theorem~\ref{thm:convert-to-normal-form}
	and Lemma~\ref{lem:reduce-2-inj}
	there is a color class $\colclass \in \colorclasses{\Struct}$
	such that $\autgroup{\Struct'[\colclass']}$ is a section
	of $\autgroup{\Struct[\colclass]}$.
	By assumption $\autgroup{\Struct[\colclass]} \iso \DihedralGroup{k}$
	for an odd $k$ and thus 
	$\autgroup{\Struct'[\colclass']} \iso \DihedralGroup{k'}$
	for an odd $k'$.
	So in $\Struct'$ no color class
	has an automorphism group isomorphic to $\DihedralGroup{4}$
	and thus $\Struct'$ is double-CFI-free.
	\item For every $2$-injective subdirect product occurring in $\Struct$,
	there are two factors
	such that the projection on them is a diagonal subgroup
	by Lemma \ref{lem:reduce-2-inj}.
	This is not the case for the two groups that are forbidden, so they cannot occur.
	\qedhere
\end{enumerate}
\end{proof}

\subsection{Reflection Components}

Let $\Struct = (\groupvertices \disunion \extensionvertices, \rel_1^\Struct, \dots, \rel_k^\Struct, \spleq)$ be an arbitrary
dihedral $2$-injective double-CFI-free quotient structure,
which we fix throughout this section.
Whenever we construct a CPT term in the following,
it does not depend on $\Struct$ but possibly gets $\Struct$ as input
and in particular satisfies the claimed properties for all dihedral $2$-injective double-CFI-free quotient structures.

\newcommand{\reflcomp}{D}
\newcommand{\oriG}[1]{\overline{#1}}
\newcommand{\orientations}{\mathbb{O}}
\newcommand{\other}[1]{\overline{#1}}
\newcommand{\oriASym}{\shortuparrow}
\newcommand{\oriBSym}{\shortdownarrow}
\newcommand{\oriA}[1]{#1^{\oriASym}}
\newcommand{\oriB}[1]{#1^{\oriBSym}}
\newcommand{\ori}{o}
\newcommand{\oriO}[1]{#1^\ori}
\newcommand{\oriOt}[1]{#1^{\other{\ori}}}
\newcommand{\remOrientation}[1]{\tilde{#1}}

We introduce notation:
we use the set $\orientations := \set{\oriASym, \oriBSym}$
to denote orientations
(e.g.,~think of the two possible ways to convert
an undirected cycle to a directed cycle).
For an orientation $\ori \in \orientations$ we set $\other{\ori} := \ori'$  as the reverse orientation, so that $\orientations = \set{\ori, \ori'}$.

\begin{definition}[Orientation]
	We say that a structure \[\Struct' = (\groupvertices \disunion \extensionvertices, \rel_1^\Struct, \dots, \rel_k^\Struct, \spleq')\]
	is an \defining{orientation} of $\Struct$
	if $\spleq'$  refines $\spleq$ with the following property:
	Let $\colclass \in \colorclasses{\Struct}$ be a  color class
	that is split by $\spleq'$,
	then $\autgroup{\Struct[\colclass]}$ 
	is a non-abelian dihedral group
	and $\colclass$ is split into two color classes
	$\oriA{\colclass}$and $\oriB{\colclass}$,
	such that each of the two classes contains one of the two oriented cycles inducing the standard form in $\colclass$.
	We say that $\Struct'$ \defining{orients} $\colclass$.
\end{definition}
By splitting the color class $\colclass$ in the above manner,
we precisely prohibit the reflections in~$\colclass$.
Note that an orientation modifies only the preorder of the structure.
Hence, defining an orientation of $\Struct$ is always canonization
preserving, because we can easily undo the changes in CPT.

For a color class $\colclass$
with dihedral automorphism group
we can define in CPT two orientations $\oriO{\Struct}_\colclass$
for $\ori \in \orientations$
that only orient $\colclass$
(by the two possible orders $\oriO{\colclass} \spless' \oriOt{\colclass}$).
Of course, we cannot choose one orientation canonically.
But the orientation of $\colclass$ 
can canonically be propagated to other color classes
in the following cases:
\begin{enumerate}[label=\alph*)]
	\item Whenever $\colclass$
	is part of a rotate-or-reflect
	group
	(because once we cannot reflect in one component,
	we cannot do so in the others), and 
	\item
	whenever $\colclass'$ is a quotient of $\colclass$
	(because once we remove
	reflections from $\colclass$
	we can also remove remaining reflections from quotient groups).
\end{enumerate}
We now formalize the propagation of orientations.

\newcommand{\simRefl}{\equiv}
\newcommand{\bothRefl}{\parallel}
\begin{definition}
	\label{def:sim-refl}
	We define the relation $\bothRefl$ on the color classes of $\Struct$ as follows: $\colclass_1 \bothRefl \colclass_2$ if and if one of the following conditions hold:
	\begin{enumerate}[label=\alph*)]
		\item  $\colclass_1,\colclass_2 \subseteq \groupvertices$,
		there is (up to reordering of the color classes)
		a group type $\type = (\colclass_1, \colclass_2, \colclass_3) \in \grouptypes$,
		and $\colclass_1$ and $\colclass_2$ are non-abelian.
		\item $\colclass_i \subseteq \groupvertices$ and $\colclass_j \subseteq \extensionvertices$
		and $\colclass_i$ is a non-abelian quotient of $\colclass_j$
		for $\set{i,j} = [2]$.
	\end{enumerate}
	The equivalence relation $\simRefl$ is the transitive closure of $\bothRefl$.
\end{definition}

Note that if $\colclass_1 \simRefl \colclass_2$
and $\colclass_1 \neq \colclass_2$,
then both $\colclass_1$ and $\colclass_2$ are non-abelian.
First, we show that if $\colclass_1 \simRefl \colclass_2$
and given an orientation of $\colclass_1$,
we can define an orientation of both~$\colclass_1$ and~$\colclass_2$ in CPT.
Second, we analyse how the structure $\Struct$
decomposes into $\simRefl$-equivalence classes.

\subsubsection{Propagation of Orientations}

To analyse the effect of orienting one component of
a $2$-injective subdirect product,
we use the classification of
$2$-injective subdirect products of dihedral groups
(Theorems~\ref{thm:classification-2inj-sub-dihedral}
and \ref{thm:classification-2inj-sub-dihedral-cyclic}).

\begin{lemma}
	\label{lem:2inj-sub-aut-graphs}
	Let $\type = (\colclass_1, \colclass_2, \colclass_3) \in \grouptypes$
	be a group type.
	Then one of the following holds:
	\begin{enumerate}
		\item $\typegroup{\type}$ is abelian, in particular $\colclass_i$ is abelian for all $i \in [3]$.
		\item $\typegroup{\type}$ is a rotate-or-reflect group, in particular $\colclass_i$ is non-abelian for all $i \in [3]$.
		\item Up to permutation of the color classes,
		$\colclass_1$ and $\colclass_2$ are non-abelian,
		$\autgroup{\colstruct_3}$ is isomorphic to one of
		$\set{\DihedralGroup{2}, \CyclicGroup{2}, \CyclicGroup{1}}$,
		and $\outproj_{\colclass_3}(\typegroup{\type})$
		is a rotate-or-reflect group.
	\end{enumerate} 
\end{lemma}
\begin{proof}
	If Conclusion 1~does not hold then,
	w.l.o.g.,~$\colclass_1$ is non-abelian.
	By the properties of $2$-injective quotient structures,
	$\typegroup{\type}$ is a $2$-injective subdirect product.
	If $\colclass_i$ is non-abelian for each $i\in[3]$,
	by Theorem~\ref{thm:classification-2inj-sub-dihedral}
	the group $\typegroup{\type}$
	is either a rotate-or-reflect group
	or the double CFI group.
	The later one is impossible 
	because $\Struct$ is double-CFI-free.
	
	Assume $\colclass_2$ is non-abelian and $\colclass_3$ is abelian.
	If $\autgroup{\colstruct_3}$ is a cyclic group,
	then
	$\outproj_{\colclass_3}(\typegroup{\type})$
	is a rotate-or-reflect group
	by Theorem~\ref{thm:classification-2inj-sub-dihedral-cyclic}
	and $\autgroup{\colstruct_3} \in \set{\CyclicGroup{2}, \CyclicGroup{1}}$.
	The double CFI case is not possible again,
	because $\Struct$ is double-CFI-free.
	If $\autgroup{\colstruct_3}$ is not a cyclic group,
	then it is isomorphic to $\DihedralGroup{2}$
	because $\DihedralGroup{2}$ is the only non-cyclic abelian dihedral group.
	
	Finally, it cannot be the case
	that both $\colclass_2$ and $\colclass_3$
	are abelian by Theorem~\ref{thm:classification-2inj-sub-dihedral-cyclic}.
\end{proof}

\begin{lemma}
	\label{lem:orient-2-inj-sub}
	There is a CPT term that, given 
	$\type = (\colclass_1, \colclass_2, \colclass_3) \in \grouptypes$
	and an orientation~$\Struct'$ that orients
	at least one of $\colclass_1$, $\colclass_2$, and $\colclass_3$,
	outputs another orientation $\Struct''$
	orienting all non-abelian color classes of
	$\colclass_1$, $\colclass_2$, and $\colclass_3$.
\end{lemma}
\begin{proof}
	If $\Struct'$ orients all of the non-abelian color classes,
	we just set $\Struct'' := \Struct'$.
	Otherwise 
	let $V := \bigcup_{i\in[3], \colclass_i \text{ not abelian}} \colclass_i$
	be the non-abelian color-classes (in $\Struct$).
	Then $\restrictGroup{\typegroupStruct{\Struct}{\type}}{V}$ 
	is a rotate-or-reflect group
	by Lemma~\ref{lem:2inj-sub-aut-graphs}
	and thus
	$\restrictGroup{\typegroupStruct{\Struct'}{\type}}{V}$ cannot contain any reflections.
	
	We compute and order the $2$-orbits of $\bigcup_{i \in [3]} \colclass_i$ under $\typegroupStruct{\Struct'}{\type}$.
	This is CPT-definable by Lemma~\ref{lem:compute-and-order-orbits}.
	Then, for each color class with dihedral automorphism group,
	the two directed cycles belong to different orbits,
	and we define $\spleq''$ accordingly.
\end{proof}

It remains to consider quotient color classes:

\begin{lemma}
	\label{lem:orienate-quotients}
	There is a CPT term that given
	an extension color class with dihedral automorphism group
	$\colclass \subseteq \extensionvertices$,
	a non-abelian $N$-quotient $\colclass' \subseteq \groupvertices$ of $\colclass$ (where $N \normal \autgroup{\colstruct}$),
	and an orientation $\Struct'$ of one of $\colclass$ and $\colclass'$
	outputs an orientation $\Struct''$ of both $\colclass$ and $\colclass'$.
\end{lemma}
\begin{proof}
	Let $\colclass$,~$\colclass'$,~$N$, and~$\Struct'$ be as above.
	Note that $N \leq \rotsubgroup{\autgroup{\colstruct}}$
	because $\autgroup{\colstruct'}$ is non-abelian.
	Because $N$ does not contain reflections,
	a vertex in one standard form cycle of $\colclass$ is never in the same orbit
	as a vertex of the other cycle.
	In particular, the vertices of one cycle of $\colclass$
	are adjacent via the orbit-map relation
	with the vertices in only one cycle of $\colclass'$
	(and the some for the other cycle).
	
	Hence, if there is an order on the two cycles of $\colclass$
	(or $\colclass'$ respectively),
	we can lift the order to the two cycles of $\colclass'$
	(or $\colclass$ respectively) via the orbit-map.
	This can clearly be defined in CPT.
\end{proof}

\begin{corollary}
	\label{cor:orientate-sim-refl}
	There is a CPT term that
	given an orientation $\Struct'$ of a color class $\colclass_1$
	and another color class with $\colclass_2 \simRefl \colclass_1$
	outputs an orientation $\Struct''$ of $\colclass_1$ and $\colclass_2$.
\end{corollary}
\begin{proof}
	It suffices to show that the claim holds for the relation $\bothRefl$.
	If $\colclass_1$ and $\colclass_2$
	are non-abelian group color classes
	and up to reordering of the color classes
	$(\colclass_1, \colclass_2, \colclass_3) \in \grouptypes$
	is a group type,
	the claim follows by Lemma~\ref{lem:orient-2-inj-sub}.
	Otherwise, w.l.o.g., $\colclass_1$ is a non-abelian quotient of~$\colclass_2$
	and the claim follows by Lemma~\ref{lem:orienate-quotients}.
\end{proof}

\subsubsection{Reflection Components and Border Color Classes}

\tikzstyle{polysix} = [regular polygon, regular polygon sides=6]
\tikzstyle{extClass} = [polysix,thick,  draw=black, fill=none, inner sep=2mm]
\tikzstyle{groupClass} = [polysix, thick, draw=black, fill=none, inner sep=1.5mm]
\tikzstyle{abelian} = [circle]
\tikzstyle{border} =[fill=gray]

\newcommand{\orbitmap}[3]{\path[#3, thick] #1 edge #2;}
\newcommand{\grouptype}[3]{
	\path[-] 
	(#1.center) edge (#3.center)
	($(#1.center) + (0,0.1)$) edge ($(#3.center) + (0,0.1)$)
	($(#1.center) - (0,0.1)$) edge ($(#3.center) - (0,0.1)$);}
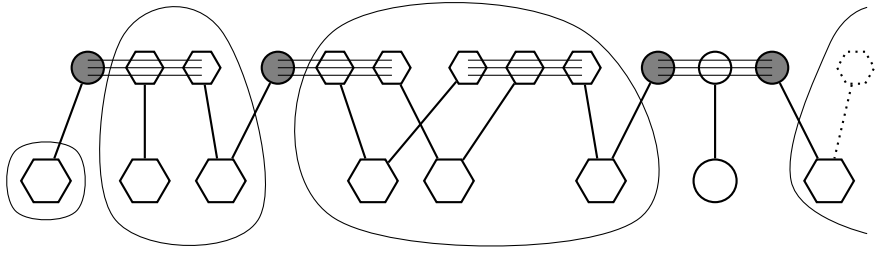
\begin{figure}
	\centering
	\begin{tikzpicture}[>=stealth,
	]
	
		\node[extClass] (E1) at (-0.3,0){};
		\node[extClass] (E2) at (1,0){};
		\node[extClass] (E3) at (2,0){};
		\node[extClass] (E4) at (4,0){};
		\node[extClass] (E5) at (5,0){};
		\node[extClass] (E6) at (7,0){};
		\node[extClass, abelian] (E7) at (8.5,0){};
		\node[extClass] (E8) at (10,0){};
		
		\def \gY{1.5}

		\node[groupClass, abelian, border] (G3) at (0.25,\gY){};
		\node[groupClass] (G4) at (1,\gY){};
		\node[groupClass] (G5) at (1.75,\gY){};
		\node[groupClass,abelian, border] (G6) at (2.75,\gY){};
		\node[groupClass] (G7) at (3.5,\gY){};
		\node[groupClass] (G8) at (4.25,\gY){};
		\node[groupClass] (G9) at (5.25,\gY){};
		\node[groupClass] (G10) at (6.0,\gY){};
		\node[groupClass] (G11) at (6.75,\gY){};
		\node[groupClass, abelian, border] (G12) at (7.75,\gY){};
		\node[groupClass, abelian] (G13) at (8.5,\gY){};
		\node[groupClass, abelian, border] (G14) at (9.25,\gY){};
		\node[groupClass, dotted] (G15) at (10.35,\gY){};
		
		\orbitmap{(G3)}{(E1)}{-}
		\orbitmap{(G4)}{(E2)}{-}
		\orbitmap{(G5)}{(E3)}{-}
		\orbitmap{(G6)}{(E3)}{-}
		\orbitmap{(G7)}{(E4)}{-}
		\orbitmap{(G8)}{(E5)}{-}
		\orbitmap{(G9)}{(E4)}{-}
		\orbitmap{(G10)}{(E5)}{-}
		\orbitmap{(G11)}{(E6)}{-}
		\orbitmap{(G12)}{(E6)}{-}
		\orbitmap{(G13)}{(E7)}{-}
		\orbitmap{(G14)}{(E8)}{-}
		\orbitmap{(G15)}{(E8)}{dotted}
		
		\grouptype{G3}{G4}{G5}
		\grouptype{G6}{G7}{G8}
		\grouptype{G9}{G10}{G11}
		\grouptype{G12}{G13}{G14}
		
		\draw  plot[smooth cycle, tension=.7, thin] coordinates {(0.7,1.9) (2,2) (2.5, -0.5) (0.6,-0.5)};
		
		\draw  plot[smooth cycle, tension=.7] coordinates {(3.5, 2) (6.75,2) (7.5,-0.5) (3.5,-0.5) };
		
		\draw  plot[smooth cycle, tension=.7] coordinates {(-0.7, 0.4) (0.1,0.4) (0.1,-0.4) (-0.7,-0.4) };
		
		\draw  plot[smooth, tension=.7] coordinates {(10.5, -0.7) (9.5, 0) (10,1.8) (10.5, 2.3) };
		
	\end{tikzpicture}
	
	\caption{A 2-injective quotient structure with dihedral colors:
		an abelian color class is drawn as circle, a non-abelian one as hexagon.
		The group color classes are at the top, the extension classes at the bottom.
		An edge between a group class~$\colclass$ and an extension class~$\colclass'$ denotes an orbit-map and~$\colclass$ is a quotient of~$\colclass'$.
		Edges between group color classes indicate relations of arity $3$.
		The reflection components are encircled and border color classes are gray.}
	\label{fig:2-inj-dih-structure}
\end{figure}

Corollary~\ref{cor:orientate-sim-refl} shows that
given an orientation of one color class $\colclass$,
we can orient the $\simRefl$\nobreakdash-equivalence class of $\colclass$ in CPT.
This observation gives rise to the following definition:

\begin{definition}[Reflection Component]
	\label{def:reflection-component}
	We call a $\simRefl$\nobreakdash-equivalence class $\reflcomp \subseteq \colorclasses{\Struct}$
	containing only non-abelian color classes
	a \defining{reflection component}.
\end{definition}
Figure~\ref{fig:2-inj-dih-structure} shows an example with its reflection components.
Recall that by definition of $\simRefl$,
the abelian color classes are all in singleton equivalence classes.
Because all color classes of a reflection component
can be oriented by orienting only a single color class,
we can speak of the two orientations
of a reflection component $\reflcomp$.
They can be defined in CPT by first orienting the smallest
color class (with respect to $\spleq$) contained in $\reflcomp$ and then applying Corollary~\ref{cor:orientate-sim-refl}.
We make some simple observations:

\begin{corollary}
	There is a CPT term defining the set of all reflection components of $\Struct$.
\end{corollary}

\begin{corollary}
	The two orientations $\oriO{\Struct[\reflcomp]}$ of a reflection component $\reflcomp$
	have abelian colors.
\end{corollary}

\begin{corollary}
	\label{cor:refl-comp-rotate-or-reflect}
	$\autgroup{\Struct[\reflcomp]}$
	is a rotate-or-reflect group
	for every reflection component~$\reflcomp$.
\end{corollary}

\begin{corollary}
	If $\Struct$ contains only one reflection component $\reflcomp$,
	then we can canonize $\Struct$ in CPT.
\end{corollary}
\begin{proof}
	We define the two orientations $\oriO{\Struct}$,
	$\ori \in \orientations$,
	of $\Struct$ by orienting the unique reflection component.
	Then we apply the canonization procedure for
	structures with abelian color classes to them
	and obtain two canonizations $\can{\oriO{\Struct}}$.
	From these, we define two canonizations $\oriO{\can{\Struct}}$ of $\Struct$.
	If they are different, we output the lexicographically smaller one.
	Otherwise, we output the unique element contained in the set
	$\set{\oriA{\can{\Struct}},\oriB{\can{\Struct}}}$.
\end{proof}

We now want to deal with the case that $\Struct$
contains multiple reflection components
and thus have to deal with the restrictions
that possible canonizations of one reflection component
impose on other reflection components.
To do so, we want to define
the canonical labellings
(recall, the set of isomorphisms into the canonical copy)
of a reflection component .
Since this set can be exponentially large,
we want to encode it as a solution of an equation system
(as already done in the canonization procedure of abelian color classes).
In fact, we want to use the two equations systems given
for the two orientations of a reflection component $\reflcomp$
to define an equation system for~$\reflcomp$.
However, we are actually not interested in the entire isomorphisms,
but only in their restrictions onto the color classes
connecting a reflection component to the rest of the structure.
We now analyse properties of color classes that connect different reflection components.

\begin{definition}[Border Color Class]
	Let $\reflcomp \subseteq \colorclasses{\Struct}$ be a reflection component.
	We call a color class $\colclass \in \colorclasses{\Struct}$
	a \defining{border color class of $\reflcomp$}
	if $\colclass\notin \reflcomp$
	and $\colclass$ is related to a color class contained in $\reflcomp$.
	We denote with $\borderclasses{\reflcomp}$ the set of all border color classes of $\reflcomp$.
\end{definition}

\begin{corollary}
	There is a CPT term
	that given a reflection component $\reflcomp$
	defines the border color classes $\borderclasses{\reflcomp}$.
\end{corollary}

\begin{lemma}
	\label{lem:border-color-classes-abelian}
	Let $\reflcomp \subseteq \colorclasses{\Struct}$ be a reflection component
	and $\colclass \in \borderclasses{\reflcomp}$
	a border color class of $\reflcomp$.
	Then $\autgroup{\colstruct}$
	is isomorphic to one of
	$\set{\CyclicGroup{1}, \CyclicGroup{2}, \DihedralGroup{2}}$
	and $\colclass$ is a group color class.
\end{lemma}
\begin{proof}
	Let $\colclass' \in \reflcomp$ be related to $\colclass$
	(such a $\colclass' $ exists by the definition of a border color class).
	Because $\colclass' \in\reflcomp$, its automorphism group is non-abelian.
	We first make a case distinction on~$\colclass$.
	
	If $\colclass \subseteq \groupvertices$ is a group color class,
	we make a second case distinction on $\colclass'$.
	If $\colclass' \subseteq \groupvertices$ is a group color class, too,
	let $\colclass_1 = \colclass$, $\colclass_2 = \colclass'$, and $\colclass_3$
	be the three related  group color classes
	forming a $2$-injective subdirect product.
	Because $\colclass'$ is non-abelian,
	there are  by Lemma~\ref{lem:2inj-sub-aut-graphs}
	two non-abelian group color classes of the $\colclass_i$.
	If $\colclass$ is non-abelian, then $\colclass \simRefl \colclass'$
	contradicting that $\colclass \in \borderclasses{\reflcomp}$.
	So $\colclass$ must be abelian and hence 
	$\autgroup{\colstruct}$ is isomorphic to
	$\DihedralGroup{2}, \CyclicGroup{2}$, or $\CyclicGroup{1}$.
	
	If $\colclass' \subseteq \extensionvertices$ is an extension color class,
	then $\colclass$ is a quotient color class of $\colclass'$.
	If $\colclass$ is non-abelian,
	then $\colclass \simRefl \colclass'$
	contradicting $\colclass \in \borderclasses{\reflcomp}$.
	The only abelian quotients of $\colclass$
	are isomorphic to $\DihedralGroup{2}, \CyclicGroup{2}$, or $\CyclicGroup{1}$.
	
	Finally, consider the case that $\colclass \subseteq \extensionvertices$
	is an extension color class.
	Then $\colclass' \subseteq \groupvertices$ is a quotient of $\colclass$
	and because $\colclass'$ is non-abelian so is $\colclass$,
	in particular $\colclass \simRefl \colclass'$.
	Once again, we obtain a contradiction to $\colclass \in \borderclasses{\reflcomp}$.
\end{proof}

So the border color classes of a reflection component $\reflcomp$
are all abelian group color classes
and two reflection components
can only be connected through them
(there could, e.g.,~be other cyclic groups between the border color classes).
In other words,
the reflection components
are embedded in a global abelian part of the structure
(cf.~Figure~\ref{fig:2-inj-dih-structure}).

\begin{definition}
	For a reflection component $\reflcomp$
	we define $\Struct_\reflcomp := \Struct[\borderclasses{\reflcomp} \cup \bigcup \reflcomp]$.
	We denote the two CPT-definable (abelian) orientations of $\Struct_\reflcomp$
	with $\oriO{\Struct_\reflcomp}$, $\ori \in \orientations$.
\end{definition}

\subsubsection{Canonical Labellings of Reflection Components}

\newcommand{\undoorican}[1]{\overline{\mathsf{can}}(#1)}

Let $\reflcomp$ be a reflection component
and assume we are given canonizations $\can{\oriO{\Struct_\reflcomp}}$
for all $\ori \in \orientations$.
We denote with $\undoorican{\oriO{\Struct_\reflcomp}}$
the structure obtained from $\can{\oriO{\Struct_\reflcomp}}$
by undoing the orientation as explained earlier,
i.e., by undoing the refinement of the preorder.
Then $\Struct_\reflcomp \iso \undoorican{\oriO{\Struct_\reflcomp}}$.

Let $<$ be the lexicographical order on canonizations.
We define the canonization $\can{\Struct_\reflcomp}$
to be the $<$-minimal canonization $\undoorican{\oriO{\Struct_\reflcomp}}$ with~$\ori \in \orientations$.
We now analyse the canonical labellings
of a reflection component.

\begin{lemma}
	\label{lem:reflcomp-canon-noniso-orientations}
	If $\can{\oriO{\Struct_\reflcomp}} < \can{\oriOt{\Struct_\reflcomp}}$
	for some $\ori \in \orientations$,
	then we have $\isos{\Struct_\reflcomp} {\can{\Struct_\reflcomp}}  = 
	\isos{\oriO{\Struct_\reflcomp}}{\can{\oriO{\Struct_\reflcomp}}}$.
\end{lemma}
\begin{proof}
	The inclusion $\isos{\oriO{\Struct_\reflcomp}}{\can{\oriO{\Struct_\reflcomp}}} \subseteq 
	\isos{\Struct_\reflcomp} {\can{\Struct_\reflcomp}}$
	holds because $\can{\Struct_\reflcomp} = \undoorican{\oriO{\Struct_\reflcomp}}$ and $\autgroup{\oriO{\Struct_\reflcomp}} \subseteq \autgroup{\Struct_\reflcomp}$.
	Hence we can write 
	$\isos{\Struct_\reflcomp} {\can{\Struct_\reflcomp}} = \phi \autgroup{\Struct_\reflcomp}$ for some (and thus every) $\phi \in \isos{\oriO{\Struct_\reflcomp}}{\can{\oriO{\Struct_\reflcomp}}}$.
	
	For the reverse direction, let $\phi' \in \isos{\Struct_\reflcomp} {\can{\Struct_\reflcomp}}$.
	Then there is a $\psi \in \autgroup{\Struct_\reflcomp}$
	such that $\phi' = \phi \circ \psi$.
	The automorphisms $\psi$ is either a rotation or a reflection on all group
	and extension color classes of $\reflcomp$ simultaneously
	by Corollary~\ref{cor:refl-comp-rotate-or-reflect}.
	If $\psi$ is a reflection on one (and hence all by Corollary~\ref{cor:refl-comp-rotate-or-reflect}) color classes,
	then it exchanges the two cycles of the standard form  within each color class and in particular
	$\oriA{\Struct_\reflcomp} \iso \oriB{\Struct_\reflcomp}$
	contradicting our assumption.
	Hence $\psi$ is a rotation on all color classes
	and thus $\psi \in \autgroup{\oriO{\Struct_\reflcomp}}$,
	i.e., $\autgroup{\oriO{\Struct_\reflcomp}} = \autgroup{\Struct_\reflcomp}$.
	Then $\phi' = \phi \circ\psi \in \phi \autgroup{\oriO{\Struct_\reflcomp}}
	= \isos{\oriO{\Struct_\reflcomp}}{\can{\oriO{\Struct_\reflcomp}}}$.
\end{proof}

\begin{lemma}
	\label{lem:reflcomp-canon-iso-orientations}
	If $\can{\oriA{\Struct_\reflcomp}} = \can{\oriB{\Struct_\reflcomp}}$,
	then $\isos{\Struct_\reflcomp} {\can{\Struct_\reflcomp}}  = 
	\bigcup_{\ori \in \orientations}\isos{\oriO{\Struct_\reflcomp}}{\can{\oriO{\Struct_\reflcomp}}}$.
\end{lemma}
\begin{proof}
	The inclusion $\isos{\oriO{\Struct_\reflcomp}}{\can{\oriO{\Struct_\reflcomp}}} \subseteq 
	\isos{\Struct_\reflcomp} {\can{\Struct_\reflcomp}}$
	follows because $\can{\Struct_\reflcomp} = \undoorican{\oriO{\Struct_\reflcomp}}$ for $\ori \in \orientations$
	as in Lemma~\ref{lem:reflcomp-canon-noniso-orientations}.
	Again, we can write $\isos{\Struct_\reflcomp} {\can{\Struct_\reflcomp}}
	= \phi\autgroup{\Struct_\reflcomp}$
	for some (and thus every) $\phi \in \isos{\oriA{\Struct_\reflcomp}}{\can{\oriA{\Struct_\reflcomp}}}$.
	Let $\phi' \in \isos{\Struct_\reflcomp} {\can{\Struct_\reflcomp}}$.
	Then there is a $\psi \in \autgroup{\Struct_\reflcomp}$ such that
	$\phi' = \phi \circ \psi$.
	Now $\psi$ is either a reflection or a rotation
	on all color classes of $\reflcomp$ simultaneously
	by Corollary~\ref{cor:refl-comp-rotate-or-reflect}.
	If $\psi$ is a rotation everywhere,
	then $\psi \in \autgroup{\oriA{\Struct_\reflcomp}}$
	and thus $\phi' = \phi\circ \psi \in \phi \autgroup{\oriA{\Struct_\reflcomp}} = \isos{\oriA{\Struct_\reflcomp}}{\can{\oriA{\Struct_\reflcomp}}}$.
	If $\psi$ is a reflection everywhere,
	then $\psi \in \isos{\oriA{\Struct_\reflcomp}}{\oriB{\Struct_\reflcomp}} = \isos{\oriB{\Struct_\reflcomp}}{\oriA{\Struct_\reflcomp}}$
	(equality holds because reflections are self-inverse).
	So $\phi' = \phi \circ \psi \in \isos{\oriB{\Struct_\reflcomp}}{\can{\oriB{\Struct_\reflcomp}}}$.
\end{proof}

\subsection{Canonizing Abelian Structures}

We sketch the canonization procedure for abelian color classes.
It proceeds inductively:
it canonizes an induced substructure consisting of $\leq r$ color classes,
where $r$ is the bound on the arity,
and adds it to the already computed partial canonization computed so far.
Because the added substructure is of constant size,
one can consider all possible ordered versions of it, and
pick the minimal one compatible with the existing canonization.
This compatibility check is done
by checking whether an equation system
encoding the canonical labellings of the canonization so far
and another one encoding the canonical labellings of the new substructure
are consistent together.
The following theorem is a precise statement of the canonization procedure
for bounded abelian colors
with all properties needed in this paper.

\begin{theorem}[\cite{AbuZaidGraedelGrohePakusa2014}]
	\label{thm:canonize-abelian-structures}
	There is a CPT term that given a 
	$q$-bounded $\sig$-structure $\Struct$ with transitive and abelian
	colors of arity $r$
	defines a canonical copy $\can{\Struct}$,
	cosets of orderings $\Psi_\colclass = \sigma_\colclass\autgroup{\colstruct} \subseteq \orderings{\colclass}$ with $\sigma_\colclass \in \orderings{\colclass}$ for every color class $\colclass \subseteq \colorclasses{\Struct}$,
	pairwise coprime prime powers $p_1, \dots, p_\ell$,
	distinct variable sets $\Vars_1, \dots, \Vars_\ell$,
	an embedding
	\[\iota : \bigotimes_{\colclass \in \colorclasses{\Struct}} \Psi_\colclass \to \bigotimes_{i \in [\ell]} \ZZ_{p_i}^{\Vars_i},\]
	and a sequence of CESs $\ces = (\ces_{p_1}, \dots, \ces_{p_\ell})$
	such that $\ces_i$ is a CES over $\ZZ_{p_i}$ using variables $\Vars_i$,
	$\solutions{\ces} \subseteq \img{\iota}$ and
	$\iota^{-1}(\solutions{\ces}) = \isos{\Struct}{\can{\Struct}}$.
	The embedding has the following properties:
	\begin{itemize}
		\item Let $\colclass \in \colorclasses{\Struct}$.
		The abelian group $\autgroup{\colstruct}$ is isomorphic to the
		direct sum of cyclic groups of prime power order.
		Let these prime powers be $q_1, \dots, q_k$.
		Then there are variable sets
		$\Vars^\colclass_1, \dots, \Vars^\colclass_k$
		such that $V^\colclass_i$ ranges over $\ZZ_{p_j}$
		with $q_i | p_j$ and there are constraints
		$q_i \cdot v = 0$ for all $v \in \Vars^\colclass_i$ in $\ces_{p_j}$
		embedding $\ZZ_{q_i}$ into $\ZZ_{p_j}$.
		\item There is a CPT term that on input $\colstruct$
		defines the variable sets $\Vars^\colclass_1 , \dots, V^\colclass_k$
		and the cyclic constraints on these variable sets.
		That is, every variable set $\Vars^\colclass_i$ belongs to one color class
		and only depends on this color class.
	\end{itemize}
\end{theorem}

We do  not write the cosets $\Psi_\colclass$
as $\sigma_\colclass\autgroup{\colstruct}$,
because we cannot choose some $\sigma_\colclass \in \Psi_\colclass$.
The set $\bigotimes_{\colclass \in \colorclasses{\Struct}} \Psi_\colclass$
naturally corresponds to a subset of $\orderings{\StructP}$,
namely to the orderings refining the preorder on $\StructP$.
In the following, we just identify these sets.
Now, we show that we can start the canonization procedure
of abelian color classes with a TCES,
that describes initial restrictions on the color classes.
First, we define the set of orderings encoded by a TCES.

\begin{definition}
	Let $\Struct$ be a structure with abelian color classes,
	$\Psi_\colclass$ be the cosets of orderings,
	$p_1,\dots,p_\ell$ the prime powers,
	$\Vars :=\Vars_1 \cup \dots\cup \Vars_\ell$ the variables,
	and $\iota$ the embedding
	given by Theorem~\ref{thm:canonize-abelian-structures}.
	Furthermore, let $\tces$ be a series of TCESs
	using variables $\Vars_\tces$
	such that $\Vars_\tces \cap \Vars$ is contained
	in the topmost variables of $\tces$.
	We say that $\tces$
	\defining{encodes} a set of orderings $\Phi \subseteq \orderings{\StructP}$
	if
	$\extExpl{\Vars}{\restrictVect{\solutions{\tces}}{\Vars}} \subseteq \img{\iota}$
	and $\iota^{-1}(\extExpl{\Vars}{\restrictVect{\solutions{\tces}}{\Vars}}) = \Phi$.
\end{definition}

\begin{lemma}
	\label{lem:canonize-abelian-colors-initial-tces}
	Let $\Struct$ be a structure with transitive abelian color classes
	and $\tces$ a series of weakly global TCESs
	encoding a non-empty set of orderings $\Phi \subseteq \orderings{\StructP}$.
	There is a CPT term that given~$\Struct$ and~$\tces$
	outputs a canonical copy $\can{\Struct_\tces}$
	and a series of CESs~$\ces$ using variables~$\Vars$
	as given by Theorem~\ref{thm:canonize-abelian-structures} such that
	$\ces$ encodes 	$\isos{\Struct}{\can{\Struct_\tces}}$
	and
	$\tcesunion{\tces}{\ces}$ encodes the set
	$\isos{\Struct}{\can{\Struct_\tces}} \cap \Phi \neq \emptyset	$.
\end{lemma}
\begin{proof}[Proof Sketch]	
	We use $\tces$ as an initial equation system.
	Because $\tces$ encodes a nonempty set of orderings,
	we just forbid some orderings initially, but at least one remains.
	So the canonization is done as before:
	We add induced substructures of constant size and
	accumulate all additional constraints in a series of CESs $\ces$
	using variables~$\Vars$.
	Then~$\ces$ encodes the set $\isos{\Struct}{\can{\Struct_\tces}}$.
	The union $\tcesunion{\tces}{\ces}$
	encodes the intersection of both encoded sets
	by Lemma~\ref{lem:union-tces} and is in particular weakly global
	because $\ces$ is weakly global.
	Checks for solvability can be done in CPT by
	Theorem~\ref{thm:solvability-tces}.
\end{proof}

\subsection{Canonization Procedure}
\label{sec:canonization-procedure}

For dihedral colors we want to 
maintain an equation system
encoding all canonical labellings of all abelian 
color classes (and hence including all border color classes)
that extend to canonical labellings of the input structure.
This suffices to encode the dependencies between
different reflection components
because -- as we have seen in the previous section --
they can only be connected via abelian color classes.
As initialization step,
we apply the canonization procedure for abelian colors
to all abelian color classes.
Then we want to inductively add one reflection component in each step
(possibly restricting the canonical labellings of the border color classes).
But a reflection component is not of constant size
so we cannot try out all orderings.
To overcome this limitation,
we want to define a canonical version of the reflection component $\reflcomp$
by taking the existing partial canonization into account.
That is, given an equation system
encoding all canonical labellings of the partial canonization computed so far,
we want to increase both,
the equation system and the canonization,
by $\reflcomp$ in one step.

From now, we assume that the abelian color classes of a structure $\Struct$
are smaller than the non-abelian ones (according to $\spleq$).
If not, we can easily reorder them canonization preservingly.
The canonization procedure is given in Figure~\ref{fig:canonization-procedure},
where $\extExpl{A}{\Phi}$ is shorthand notation for $\extExpl{\scalebox{0.7}{$\bigcup$} A}{\Phi}$.

\begin{figure}
	\LinesNumbered
	\setlength{\interspacetitleruled}{0pt}%
	\setlength{\algotitleheightrule}{0pt}%
\begin{algorithm}[H]
	\vspace{2mm}
	\KwIn{A $2$-injective double-CFI-free quotient structure $\Struct $}
	\KwOut{The canonical copy $\can{\Struct}$ of $\Struct$}
	\label{line:compute-abelian}
	Compute the set $A \subseteq \colorclasses{\Struct}$ of abelian color classes\;
	Compute all reflection components $\reflcomp_1 < \dots < \reflcomp_m$ of $\Struct$\;
	
	Compute $\can{\Struct_0} := \can{\Struct[A]}$ and $\Phi_0 := \isos{\Struct[A]}{ \can{\Struct[A]}}$ using the canonization procedure for abelian colors\;
	\label{line:canonize-abelian-color-classes}
	
	\For{$i \in [m]$}{
		$\reflcomp := \reflcomp_i$\;
		
		Define the two orientations $\oriO{\Struct_\reflcomp}$\; 	\label{line:define-orientations}
		
		Compute $\can{\oriO{\Struct_\reflcomp}}$ and
		$\oriO{\Phi} := \isos{\oriO{\Struct_\reflcomp}}{\can{\oriO{\Struct}_\reflcomp}}$
		such that $ \Phi_{i-1} \cap \extExpl{A}{\oriO{\Phi}} \neq \emptyset$
		with the canonization procedure for abelian colors\;
		\label{line:canonize-orientations}
		
		\eIf{$\can{\oriO{\Struct_\reflcomp}} < \can{\oriOt{\Struct_\reflcomp}}$ for some $\ori \in \orientations$} {
			\label{line:check-orientations-iso}
			
			$\can{\Struct_i} := \can{\Struct_{i-1}} \cup \undoorican{\oriO{\Struct}_\reflcomp}$\;
			
			$\Phi_i := \Phi_{i-1} \cap \extExpl{A}{\restrictVect{\oriO{\Phi}}{A}}$\;
			\label{line:define-isos-not-iso-orientations}
		}
		{
			$\can{\Struct_i} := \can{\Struct_{i-1}} \cup \undoorican{\oriA{\Struct}_\reflcomp} \cup \undoorican{\oriB{\Struct}_\reflcomp}$\;
			\label{line:define-canon-iso-orientations}
			
			$\Phi_i := \Phi_{i-1} \cap \extExpl{A}{\restrictVect{(\oriA{\Phi} \cup \oriB{\Phi})}{A}}$\;
			\label{line:define-isos-iso-orientations}
		}
	}
	$\can{\Struct} := \can{\Struct_m}$\;
\end{algorithm}
	\caption{Canonizing a 2-injective double-CFI-free
	structure $\Struct$ with dihedral colors in CPT.} 
\label{fig:canonization-procedure}
\end{figure}

We fix the input structure $\Struct = (\StructP, \rel_1^\Struct, \dots \rel_k^\Struct, \spleq)$ in the following 
and first argue that the algorithm indeed defines a canonization
and then show how it can be expressed in CPT
(again, our CPT terms will not depend on $\Struct$).
The algorithm maintains canonizations $\can{\Struct_i}$ of
$\Struct_i := \Struct[A \cup \bigcup_{j \in [i]} \reflcomp_j]$
and sets $\Phi_i$ of canonical labellings.
\begin{lemma}
	\label{lem:canonization-procedure}
	For $i \leq m$ the following holds:
	$\Struct_i \iso \can{\Struct_i}$ and 
	$\Phi_i = \restrictVect{\isos{\Struct_i}{\can{\Struct_i}}}{A}$.
\end{lemma}
\begin{proof}
	For the case $i=0$ the claim follows by definition.
	So assume $i > 0$,
	$\Struct_{i-1} \iso \can{\Struct_{i-1}}$, and
	$\Phi_{i-1} = \restrictVect{\isos{\Struct_{i-1}}{\can{\Struct_{i-1}}}}{A}$.
	Let $\reflcomp = \reflcomp_i$ and $\oriO{\Struct_\reflcomp}$
	be the two orientations of $\Struct_\reflcomp$.
	We perform the case distinction as in Line~\ref{line:check-orientations-iso}:
	If $\can{\oriO{\Struct_\reflcomp}} < \can{\oriOt{\Struct_\reflcomp}}$
	for some $\ori \in \orientations$,
	then $\isos{\Struct_\reflcomp} {\can{\Struct_\reflcomp}}  = 
	\isos{\oriO{\Struct_\reflcomp}}{\can{\oriO{\Struct_\reflcomp}}}$
	by Lemma~\ref{lem:reflcomp-canon-noniso-orientations}
	and hence \[\isos{\Struct_i}{\can{\Struct_i}} = 
		\extExpl{\StructP_i}{\isos{\Struct_{i-1}}{\can{\Struct_{i-1}}}} \cap
		\extExpl{\StructP_i}{\isos{\oriO{\Struct_\reflcomp}}{\can{\oriO{\Struct_\reflcomp}}}}\]
	where $\StructP_i$ are the vertices of $\Struct_i$ and
	\[\Phi_i = \Phi_{i-1}\cap \extExpl{A}{\restrictVect{\oriO{\Phi}}{A}} = 
	\restrictVect{\isos{\Struct_{i-1}}{\can{\Struct_{i-1}}}}{A}
		\cap \extExpl{A}{\restrictVect{\isos{\oriO{\Struct_\reflcomp}}{\can{\oriO{\Struct_\reflcomp}}}}{A}}.\]
 	The reflection component $\reflcomp$
	is only connected to its border color classes ${\borderclasses{\reflcomp} \subseteq A}$.
	So we finally have
	\begin{align*}
	&\quad\restrictVect{\left(\extExpl{\StructP_i}{\isos{\Struct_{i-1}}{\can{\Struct_{i-1}}}} \cap
		\extExpl{\StructP_i}{\isos{\oriO{\Struct_\reflcomp}}{\can{\oriO{\Struct_\reflcomp}}}}\right)}{A} \\
	&= \restrictVect{\isos{\Struct_{i-1}}{\can{\Struct_{i-1}}}}{A}
	\cap \extExpl{A}{\restrictVect{\isos{\oriO{\Struct_\reflcomp}}{\can{\oriO{\Struct_\reflcomp}}}}{A}}
	\end{align*}
	and so $\Phi_i =  \restrictVect{\isos{\Struct_{i}}{\can{\Struct_{i}}}}{A}$.
	The case $\can{\oriA{\Struct_\reflcomp}} = \can{\oriB{\Struct_\reflcomp}}$ 
	proceeds similarly using Lemma~\ref{lem:reflcomp-canon-iso-orientations}.
\end{proof}

We cannot compute with the sets $\Phi_i$
directly in CPT because they can be exponentially large.
Hence we adapt the approach of the canonization procedure for abelian color classes
and encode them with sequences of weakly global TCESs $\tces_i$.
We maintain that the variables~$\Vars_A$ of the abelian color classes~$A$
are contained in the topmost variables of the~$\tces_i$.

It is clear that Lines~\ref{line:compute-abelian} to~\ref{line:define-orientations}
can be defined in CPT.
In Line~\ref{line:canonize-abelian-color-classes}
we just use the canonization procedure for abelian color classes.
It outputs, apart from the canonization, a series of CESs~$\tces_0$
encoding~$\Phi_0$
using variables~$\Vars_A$ (and hence has topmost variables~$\Vars_A$)
by Theorem~\ref{thm:canonize-abelian-structures}.

For Line~\ref{line:canonize-orientations}
we use Lemma~\ref{lem:canonize-abelian-colors-initial-tces}
to define canonizations $\can{\oriO{\Struct_D}}$
that are compatible with the canonization of $\Struct_{i-1}$ computed so far.
We use $\tces_{i-1}$ as an initial equation system
to canonize $\oriO{\Struct_\reflcomp}$,
which has topmost variables $\Vars_A$.
Because there is at least one canonical labelling,
$\emptyset \neq \solutions{\tces_{i-1}} \subseteq \solutions{\tces_0}$.
So Lemma~\ref{lem:canonize-abelian-colors-initial-tces} can be applyied
and we obtain two sequences of CESs $\oriO{\ces_\reflcomp}$.
Note that the variables $\Vars_{\borderclasses{\reflcomp}}$
of the border color classes of~$\reflcomp$
are contained in~$\Vars_A$
and the variables of $\oriO{\ces_\reflcomp}$
for all $\ori \in \orientations$.

Lines~\ref{line:check-orientations-iso} to~\ref{line:define-canon-iso-orientations}
can be defined in CPT:
in Line~\ref{line:define-isos-not-iso-orientations}
we set
$\tces_i := \tcesunion{\tces_{i-1}}{\oriO{\ces_\reflcomp}}$.
Because the common topmost variables of
$\tces_{i-1}$ and $\oriO{\ces_\reflcomp}$
are $\Vars_{\borderclasses{\reflcomp}}$,
we can apply Lemma~\ref{lem:union-tces}:
$\Vars_A$ is contained in the topmost variables of $\tces_i$ and 
$\tces_i$ encodes $\Phi_{i-1} \cap \extExpl{A}{\restrictVect{\oriO{\Phi}}{A}}$
and is weakly global because $\tces_{i-1}$ is weakly global.

The intersection in Line~\ref{line:define-isos-iso-orientations}
can be performed again due to Lemma~\ref{lem:union-tces}.
So we are only left to show that we can define a series of weakly global TCESs
with~$\Vars_{\borderclasses{\reflcomp}}$ contained in its topmost variables
encoding $\extExpl{A}{\restrictVect{(\oriA{\Phi} \cup \oriB{\Phi})}{A}}$.

\subsection{Equation Systems for Reflection Components}
Let $\tces := \tces_{i-1}$  for some $1<i \leq m$
be the series of weakly global TCESs for the canonization constructed so far
and $\reflcomp := \reflcomp_i$ be the next reflection component
to canonize.
Let the variables of the border color classes of $\reflcomp$ be
$\BVars = \BVars_1 < \dots < \BVars_k$.
Recall that these variables only depend on the border classes
(Theorem~\ref{thm:canonize-abelian-structures})
and that $\tces$ already contains cyclic constrains
for these variables (because they belong to the abelian color classes).
Also, the set of variables $\BVars$
is contained in the topmost variables $\Vars_A$ of $\tces$.

Then we define  
the two possible orientations $\oriO{\Struct_\reflcomp}$ of $\reflcomp$
for both $\ori \in \orientations$. We define 
$\can{\oriO{\Struct_\reflcomp}}$
to be the canonizations compatible with $\tces$ 
and let $\oriO{\ces}$ be the series of CESs
encoding the sets $\oriO{\Phi} =
\isos{\oriO{\Struct_\reflcomp}}{\can{\oriO{\Struct_\reflcomp}}}$
as given by Lemma~\ref{lem:canonize-abelian-colors-initial-tces}.
Note that the two series of CESs use the same variables
for the border color classes (as they are equal in both orientations).
We consider the case $\can{\oriA{\Struct_\reflcomp}} = \can{\oriB{\Struct_\reflcomp}}$
and $\can{\Struct_\reflcomp} = \undoorican{\can{\oriA{\Struct_\reflcomp}}} =  \undoorican{\can{\oriB{\Struct_\reflcomp}}}$.
We cannot fix an isomorphism in
$ \isos{\oriA{\Struct_\reflcomp}}{\oriB{\Struct_\reflcomp}} =
\isos{\oriB{\Struct_\reflcomp}}{\oriA{\Struct_\reflcomp}}$
canonically.
But we now show that we can fix
an isomorphism
contained in
$\restrictVect{\isos{\oriA{\Struct_\reflcomp}}{\oriB{\Struct_\reflcomp}}}{\borderclasses{\reflcomp}}$
canonically
(which possibly extends to multiple isomorphisms between
the orientations).

Note that
by Lemma~\ref{lem:border-color-classes-abelian}
the border color classes
have automorphism groups $\CyclicGroup{2}^\ell$ for $\ell \in \set{0,1,2}$
(and are in particular all abelian).
Hence, all variables for the border color classes range over $\ZZ_2$
by Theorem~\ref{thm:canonize-abelian-structures}.
(Precisely, $\BVars$ are variables over $\ZZ_p$ for~$p$ a power of~$2$
and there are constraints $2v = 0$ for all $v \in \BVars$
embedding $\ZZ_2$ in $\ZZ_p$,
as already discussed for TCES).

We adapt both series of CESs
such that their variables are different,
but such that a variable of a border color class of $\oriO{\ces}$
can still be identified with a variable of a border color class of $\oriOt{\ces}$.
This can for example be achieved by renaming the variables as follows: $v \mapsto (v, \oriO{\Struct_\reflcomp})$.
We denote with 
$\oriO{\Vars}$ (and $\oriO{\BVars}$ respectively)
the changed variables for $\ori \in \orientations$.
Then, because the variables can be identified one-to-one,
we write for simplicity e.g.~${\oriO{x}=\oriOt{x}}$
for two vectors indexed by the variables of the 
border color class of the two series of CESs.

\begin{lemma}
	\label{lem:vector-for-reflection-on-border-classes}
	There is a CPT term defining two vectors
	$\oriO{\vectA} = (\oriO{\vectA}_1, \dots, \oriO{\vectA}_k)
	\in \ZZ_2^{\oriO{\BVars_1}} \times \dots \times \ZZ_2^{\oriO{\BVars_k}}$
	for both $\ori \in \orientations$
	(that are equal up to identification of the variables,
	i.e.,~$\oriA{\vectA} = \oriB{\vectA}$ in the above notational convention)
	such that if $\oriO{\vectB}$ is a solution of~$\oriO{\ces}$,
	then there is a solution~$\oriOt{\vectB}$ of~$\oriOt{\ces}$
	such that
	$\restrictVect{\oriO{\vectB}}{\oriO{\BVars}} + \oriO{\vectA} = \restrictVect{\oriOt{\vectB}}{\oriOt{\BVars}}$.
\end{lemma}
\begin{proof}

	We first show that
	$\restrictVect{\oriO{\vectB}}{\oriO{\BVars}} - \restrictVect{\oriOt{\vectB}}{\oriOt{\BVars}}$
	has the desired property
	when~$\oriO{\vectB}$ is a solution of~$\oriO{\ces}$ for both
	$\ori \in \orientations$.
	We then show how we can define such two vectors in CPT.
	The vectors~$\oriO{\vectB}$
	encode isomorphisms $\oriO{\phi} \colon \oriO{\Struct_\reflcomp} \to \can{\oriO{\Struct_\reflcomp}}$
	via the embedding $\iota$ given by the abelian canonization procedure
	(cf.~Theorem~\ref{thm:canonize-abelian-structures}).
	They define two isomorphisms
	$\oriO{\psi} := \oriO{\phi} \circ \inv{(\oriOt{\phi})} : \oriO{\Struct_\reflcomp} \to \oriOt{\Struct_\reflcomp}$
	(recall that $\can{\oriO{\Struct_\reflcomp}}=\can{\oriOt{\Struct_\reflcomp}}$).
	We cannot describe their actions on the reflection component $\reflcomp$
	by vectors
	(because $\reflcomp$ is not abelian),
	but on the border color classes the action is precisely given by
	$\restrictVect{\oriO{\vectB}}{\oriO{\BVars}} - \restrictVect{\oriOt{\vectB}}{\oriOt{\BVars}}$.
	Here, if $\oriO{\vectB}(\varA) - \oriOt{\vectB}(\varA)=1$
	for $\varA \in \oriO{\BVars}_i$ and some $i \in[k]$
	then the $i$-th $\CyclicGroup{2}$-group
	has to be flipped
	and if $\oriO{\vectB}(\varA) -  \oriOt{\vectB}(\varA)=0$
	then it stays constant.
	Now let $\oriO{\vectB'}$ be another solution of $\oriO{\ces}$.
	It again encodes an isomorphism
	$\oriO{\phi'} \colon \oriO{\Struct_\reflcomp} \to \can{\oriO{\Struct_\reflcomp}}$.
	Then
	 $\oriOt{\phi'} := \oriO{\phi'} \circ \oriOt{\psi}
	 \colon \oriOt{\Struct_\reflcomp} \to \can{\oriOt{\Struct_\reflcomp}}$
	is an isomorphism 
	that is encoded by a solution $\oriOt{\vectB}$ of $\oriOt{\ces}$.
	It satisfies $\inv{(\oriOt{\phi'})}\circ \oriO{\phi'}
	= \inv{(\oriO{\phi'} \circ \oriOt{\psi})} \circ \oriO{\phi'} = \inv{(\oriOt{\psi})} = \oriO{\psi}$.
	Hence, we have
	that $\restrictVect{\oriO{\vectB}}{\oriO{\BVars}} - \restrictVect{\oriOt{\vectB}}{\oriOt{\BVars}} = \restrictVect{\oriO{\vectB'}}{\oriO{\BVars}} - \restrictVect{\oriOt{\vectB'}}{\oriOt{\BVars}}$.

	Note that $\oriO{\vectA}_i(\varA) = \oriO{\vectA}_i(\varB)$ must hold 
	for $\varA,\varB \in \oriO{\BVars_i}$,
	because by the cyclic constraints $\oriO{\vectB}(\varA) \neq \oriO{\vectB}(\varB)$ if $\varA \neq \varB$.
	Hence the desired vector $\oriO{\vectA}$ corresponds to a string in $\ZZ_2^k$.
	To define~$\oriO{\vectA}$ in CPT,
	we want to choose the lexicographically minimal one.
	This cannot be done straightforwardly, because $k$ is not bounded.
	
	We construct the vectors $\oriO{\vectA}$ inductively
	and fix the entries for variable set $\oriO{B_i}$ per step.
	Assume we have defined $(\oriO{\vectA}_1, \dots , \oriO{\vectA}_{j-1})$
	such that we can complete
	$\oriO{\vectA}$ to have the desired form
	$\restrictVect{\oriO{\vectB}}{\oriO{\BVars}} - \restrictVect{\oriOt{\vectB}}{\oriOt{\BVars}}$
	for $\oriO{\vectB}$ and $\oriOt{\vectB}$.
	To define $\oriO{\vectA}_j$, we want to solve the equation system
	\begin{align*}
	\oriA{\vectB}& \in \solutions{\oriA{\ces}}\\
	\oriB{\vectB}& \in \solutions{\oriB{\ces}}\\
	\oriA{\vectB}(\varA) - \oriB{\vectB}(\varA) &=  \oriA{\vectA}_i(\varA) & i \in [j], \varA \in \BVars_i\\
	\oriB{\vectB}(\varA) - \oriA{\vectB}(\varA) &=  \oriB{\vectA}_i(\varA) & i \in [j], \varA \in \BVars_i
	\end{align*}
	for the two possible values of $\oriO{\vectA}_j$
	(or equivalently of $\oriOt{\vectA}_j$,
	note that $\oriA{\vectA}_i(\varA) = \oriA{\vectB}(\varA) - \oriB{\vectB}(\varA) = \oriB{\vectB}(\varA) - \oriA{\vectB}(\varA) = \oriB{\vectA}_i(\varA)$ because we are working in $\ZZ_2$).
	But the equations above do not define a series of TCESs
	because there is no order between the variables $\oriA{\BVars_i}$ and $\oriB{\BVars_i}$.
	Hence we solve four equations systems,
	two for each possible value of $\oriO{\vectA}_i$
	and two for each possible order setting $\oriA{\BVars_i} < \oriB{\BVars_i}$ or vice versa (for which the equation system has obviously the same solutions). 
	
	For at least one possible value of $\oriO{\vectA}_i$
	the equation system is consistent by induction hypothesis.
	If it is consistent for both, we choose $\oriO{\vectA}_i(\varA) = 1$.
	The enlarged vector $(\oriO{\vectA}_1, \dots , \oriO{\vectA}_{j})$
	satisfies the induction hypothesis, too,
	because by construction there are 
	solutions
	$\oriO{\vectB} \in \solutions{\oriO{\ces}}$
	that satisfy 
	$\oriO{\vectA} = \restrictVect{\oriO{\vectB}}{\oriO{\BVars}} - \restrictVect{\oriOt{\vectB}}{\oriOt{\BVars}}$
	(as part of the solution of the equation system above).
\end{proof}

We now use the vectors $\oriO{\vectA}$
to represent the canonical labellings of the border color classes,
which additionally extend to canonical labellings of the reflection component,
as a TCES.

\begin{lemma}
	\label{lem:tces-for-reflection-comp}
	There is a CPT term defining a series of weakly global TCESs $\tces_\reflcomp$
	with the following properties:
	\begin{enumerate}
		\item $B$ is contained in the topmost variables of $\tces_\reflcomp$.
		\item $\tces_\reflcomp$ encodes the set
		$\restrictVect{\isos{\Struct_\reflcomp}{\can{\Struct_\reflcomp}}}{\borderclasses{\reflcomp}}$.
		\item The size of $\tces_\reflcomp$ is polynomial in $|\reflcomp|$.
	\end{enumerate}
\end{lemma}
\begin{proof}
	Let $\oriO{\vectA}$ be the two vectors given
	by Lemma~\ref{lem:vector-for-reflection-on-border-classes}.
	We define a set of two variables
	$\BVars_\refl := \set{\oriA{\refl},\oriB{\refl}}$
	(and set $\oriO{\refl} := \oriO{\Struct_\reflcomp}$),
	$V_\reflcomp := \BVars \cup \BVars_\refl \cup \oriA{\Vars} \cup \oriB{\Vars}$,
	and $\spleq_\reflcomp$
	such that it respects the orders on $B$ and $\oriO{\Vars}$
	and $\BVars \spless \BVars_\refl \spless \oriO{\Vars}$ for all $\ori \in \orientations$. 
	The variable sets $\oriA{\Vars}$ and $\oriB{\Vars}$ are incomparable.
	
	We want to define a TCES $\tces_\reflcomp$ enforcing that if
	$\vectC \in \solutions{\tces_\reflcomp}$,
	then there is an $\ori \in \orientations$ and a
	$\oriO{\vectB} \in \solutions{\oriO{\ces}}$
	such that $\vectC = \restrictVect{\oriO{\vectB}}{\BVars}$.
	To do so, we guess two solutions $\oriO{\vectB} \in \solutions{\oriO{\ces}}$
	(one for each $\ori \in \orientations$)
	with the property that $\restrictVect{\oriO{\vectB}}{\BVars} + \oriO{\vectA} = \restrictVect{\oriOt{\vectB}}{\BVars}$ (Lemma~\ref{lem:vector-for-reflection-on-border-classes}).
	Then we want to ensure that 
	$\vectC = \restrictVect{\oriA{\vectB}}{\BVars}$ or $\vectC = \restrictVect{\oriB{\vectB}}{\BVars}$.
	To allow that one equality does not hold,
	we use the additional variables $\oriA{\refl}$
	to express the constraints
	$\vectC = \restrictVect{\oriO{\vectB}}{\BVars} + \oriO{\refl} \cdot  \oriO{\vectA}$.
	By enforcing that exactly one of $\oriA{\refl}$ and $\oriB{\refl}$ is
	$1$, we obtain the desired system.
	Finally, to make the system linear, we encode the multiplication
	$\oriO{\refl} \cdot  \oriO{\vectA}$.
	This is possible, because $\oriO{\vectA}$ does not depend on $\oriO{\vectB}$
	and can be defined before defining the following TCES:
	\begin{align*}
	\oriA{\vectB} &\in \solutions{\oriA{\ces}}\\
	\oriB{\vectB} &\in \solutions{\oriB{\ces}}\\
	\vectC(\varA) &= \oriA{\vectB}(\varA) = \oriB{\vectB}(\varA) & \text{if }\oriA{\vectA}(\varA) = \oriB{\vectA}(\varA) = 0, \varA \in \BVars\\
	\vectC(\varA) &= \oriA{\vectB}(\varA)+ \oriA{\refl} = \oriB{\vectB}(\varA) + \oriB{\refl}  & \text{if } \oriA{\vectA}(\varA) = \oriB{\vectA}(\varA) = 1, \varA \in \BVars\\
	1 &=\oriA{\refl} + \oriB{\refl} 
	\end{align*}
	where $\oriO{\vectB}$ is indexed by $\oriO{V}$
	and $\vectC$ is indexed by $B$ and ranges over $\ZZ_2$.
	If the variable~$\oriO{\refl}$ is assigned to~$1$,
	then $\vectC = \oriO{\vectB}|_{\oriO{\BVars}} + \oriO{\vectA}$
	and $\vectC =  \oriO{\vectB}|_{\oriO{\BVars}} $ otherwise.
	Because of the cyclic constraint $1 =\oriA{\refl} + \oriB{\refl}$,
	we add the vector $\oriO{\vectA}$
	to a solution $\oriO{\vectB}$ of $\oriO{\ces}$
	for exactly one orientation $\ori \in \orientations$.
	
	Clearly, $\tces_\reflcomp$ has topmost variables $\BVars$
	and, because the size of the $\oriO{\ces}$ is polynomial $|\reflcomp|$,
	so is the size of $\tces_\reflcomp$.
	We argue that $\tces_D$ is weakly global: 
	The only global equations are the equations relating
	$\vectC(\varA)$ and $\oriO{\vectB}(\varA)$.
	These variables and so all global variables are over $\ZZ_2$
	(embedded in $\ZZ_{2^\ell}$ by $2\var = 0$ equations, as discussed earlier).
	It suffices to show that $\tces_\reflcomp$
	encodes the set
	\[\restrictVect{ \isos{\Struct_\reflcomp} {\can{\Struct_\reflcomp}}}{\borderclasses{\reflcomp}}
	=
	\bigcup_{\ori \in \orientations}\restrictVect{\isos{\oriO{\Struct_\reflcomp}}{\can{\oriO{\Struct_\reflcomp}}}}{\borderclasses{\reflcomp}}\]
	by Lemma~\ref{lem:reflcomp-canon-iso-orientations}.
	So we have to show that
	\[\restrictVect{\solutions{\tces_\reflcomp}}{\BVars} = 
	\restrictVect{\solutions{\oriA{\ces}}}{\oriA{\BVars}}
	\cup
	\restrictVect{\solutions{\oriB{\ces}}}{\oriB{\BVars}},\]
	because $\oriO{\ces}$ encodes $\isos{\oriO{\Struct_\reflcomp}}{\can{\oriO{\Struct_\reflcomp}}}$.
	Let $\vectC \in \restrictVect{\solutions{\tces_\reflcomp}}{\BVars}$,
	so there is a solution consisting of 
	$\vectC, \oriA{\refl},\oriB{\refl},\oriA{\vectB},$ and $\oriB{\vectB}$
	of $\tces_\reflcomp$
	with $\oriO{\refl} = 0$ for some $\ori \in \orientations$
	(which must be the case by the cyclic constraint on $\oriA{\refl}$ and $\oriB{\refl}$).
	Then in particular $\restrictVect{\oriO{\vectB}}{\oriO{\BVars}} = \vectC$,
	$\oriO{\vectB} \in \solutions{\oriO{\ces}}$,
	and thus $\vectC \in \restrictVect{\solutions{\oriO{\ces}}}{\oriO{\BVars}}$.
	
	For the reverse direction,
	let $\oriO{\vectB} \in \solutions{\oriO{\ces}}$.
	Then by Lemma~\ref{lem:vector-for-reflection-on-border-classes}
	there is a solution $\oriOt{\vectB} \in \solutions{\oriO{\ces}}$
	such that $\restrictVect{\oriO{\vectB}}{\oriO{\BVars}} - \restrictVect{\oriOt{\vectB}}{\oriOt{\BVars}} = \oriO{\vectA}$.
	So $\vectC := \restrictVect{\oriO{\vectB}}{\oriO{\BVars}}$, $\oriO{\refl} = 0$, $\oriOt{\refl}=1$,
	$\oriO{\vectB}$, and 
	$\oriOt{\vectB}$
	form a solution of $\tces_\reflcomp$,
	in particular $\vectC = \restrictVect{\oriO{\vectB}}{\oriO{\BVars}} \in \restrictVect{\solutions{\tces_\reflcomp}}{\BVars}$.
\end{proof}

Now, we finally defined all operation on TCESs needed
for our canonization procedure and conclude:

\begin{theorem}
	\label{thm:canonize-double-CFI-free}
	Canonization of $2$-injective double-CFI-free $q$-bounded gadget quotient structures is CPT-definable.
\end{theorem}
\begin{proof}
	By Lemma~\ref{lem:canonization-procedure}
	the canonization procedure yields indeed a canonization
	given a CPT-definable encoding of the sets $\Phi_i$.
	The discussion after Lemma~\ref{lem:canonization-procedure}
	and Lemma~\ref{lem:tces-for-reflection-comp}
	show that the sets can be encoded by series of TCESs.
\end{proof}

\begin{proof}[Proof of Theorem~\ref{thm:canonize-structures-CPT}]
Theorem~\ref{thm:canonize-structures-CPT}
now follows from Lemma~\ref{lem:odd-dihedral-and-graphs-are-double-CFI-free}
and Theorem~\ref{thm:canonize-double-CFI-free}.
\end{proof}

Since~CPT subsumes IFP, it follows that CPT is a logic that captures \PTime{} on the class of structures mentioned in the theorem.

\section{Conclusion}

We separated a relational structure into $2$-injective subdirect products and quotients,
gave a classification of all $2$-injective subdirect products of dihedral and cyclic groups,
and used this classification to canonize relational structures with bounded dihedral colors of arity at most~$3$.
We showed that the structure decomposes into reflection components
and that in these components either all color classes have to be reflected or none.
If we exclude a single $2$-injective subdirect product,
namely the double CFI group,
the reflection components can only have abelian dependencies.
This is always true for graphs,
because the said group cannot be realized by graphs with dihedral colors.
In fact, we demonstrated the increase of complexity when considering structures of arity $3$ instead of $2$.
Apart from the fact that the double CFI group does not appear,
a classification of $1$-injective $2$-factor subdirect products of dihedral groups is much easier.
Considering higher arity,
already $3$-injective $4$-factor subdirect products of dihedral groups
cannot be classified to be (almost) abelian or reflect-or-rotate groups.
If one instead tries to reduce the arity of the structures,
one needs not only to work with a class of groups closed under taking quotients and subgroups (which is the case for dihedral and cyclic groups),
but also closed under taking direct products.

One natural way to exclude the double CFI group is a restriction to odd dihedral colors.
The difficulty with even dihedral groups
might indicate that looking at odd (non-dihedral) groups could be a reasonable next step.
A natural graph class with odd automorphism groups are tournaments.
Since such groups are solvable there is hope for an inductive approach exploiting the abelian case.
It could be possible that the techniques developed in this paper
transfer to this case.
Just like dihedral groups, odd groups are closed under taking quotients and subgroups.
However, they are also closed under direct products (and are solvable),
which would allow a reduction of the arity.
Thus, it is possible to apply our reduction to quotients and $2$-injective groups.
As a next step, one could try to follow a similar strategy as for dihedral colors:
identify components of the graph,
in which the complexity of all color classes decreases simultaneously,
when a single color class is made easier
(similar to reflection components).
This might not immediately result in abelian groups,
but recursion on the complexity of the groups could be a reasonable option,
e.g.~on the length of the composition series or on the nilpotency class.
All the mentioned avenues remain as future work.

\bibliographystyle{plainurl}
\bibliography{canon}

\end{document}